\documentclass[a4paper,11pt]{article}

\usepackage{authblk}

\usepackage[
    left=25mm,
    right=25mm,
    top=25mm,
    bottom=25mm
]{geometry}
\usepackage[x11names]{xcolor}
\definecolor{DarkRed}{rgb}{0.6471, 0.1098, 0.1882}

\usepackage{standalone}
\usepackage{hyperref}
\hypersetup{
    colorlinks, citecolor=DarkRed, linkcolor=DarkRed
}

\usepackage{graphicx}
\usepackage{tikz}
\usepackage{tikz-cd}
\usepackage{quantikz}
\usepackage[compat=0.3]{yquant}
\usetikzlibrary{arrows.meta}
\usetikzlibrary{decorations.pathmorphing}
\usetikzlibrary{shapes.geometric}
\usetikzlibrary{shapes.symbols}
\usepackage{lscape}
\usepackage{float}
\usepackage{subcaption}
\usepackage{wrapfig}
\usepackage{amssymb}
\usepackage{amsfonts}
\usepackage{bm}
\usepackage{youngtab}
\usepackage{mathdots}
\usepackage{mathtools}
\usepackage{relsize}

\newcommand{\GL}{\mathrm{GL}}
\newcommand{\U}{\mathrm{U}}

\newcommand{\tr}{\mathrm{tr}}
\newcommand{\Span}{\mathrm{span}}

\newcommand{\Res}{\mathrm{Res}}
\newcommand{\Ind}{\mathrm{Ind}}
\newcommand{\Enc}{\mathrm{Enc}}

\newcommand{\Fou}{\mathfrak{F}}
\newcommand{\F}{\mathbb{F}}
\newcommand{\C}{\mathbb{C}}

\newcommand{\Hom}{\mathrm{Hom}}
\newcommand{\Rep}{\mathrm{R}}
\newcommand{\Path}{\mathcal{P}}

\newcommand{\pt}{\mathbin{\vdash}}

\newcommand{\N}{\mathcal{N}}
\newcommand{\act}[1]{{^{(#1)}}}
\newcommand{\inv}{^{-1}}
\newcommand{\idN}{1_N}

\newcommand{\idSn}{1_{S_n}}
\newcommand{\vla}
{{\boldsymbol{\lambda}}}

\newcommand{\vnu}{{\boldsymbol{\nu}}}
\newcommand{\vsig}{{\boldsymbol{\sigma}}}

\newcommand{\la}{\lambda}
\providecommand{\id}{\mathrm{id}}
\DeclareMathOperator{\tw}{\operatorname{twist}}
\renewcommand{\P}{\mathcal{P}}

\newcommand{\hatN}{\widehat{N}}
\newcommand{\hatG}{\widehat{G}}

\newcommand{\orbreps}{{\mathcal{O}(\hatN)}}
\newcommand{\repO}{{r}_{\orbreps}}

\newcommand{\Inertia}[1]{I_{#1}}
\newcommand{\Little}[1]{H_{#1}}
\newcommand{\T}{\mathcal{T}}
\newcommand{\Transversal}[1]{\T_{#1}}
\newcommand{\TransversalElt}[1]{g_{#1}}

\newcommand{\ExtRep}{\Rep^{\mathrm{ext}}}
\newcommand{\InfRep}{\widetilde{\Rep}}

\newcommand{\IsoComp}[1]{V(#1)}

\newcommand{\Z}{\mathbb{Z}}
\newcommand{\I}{\mathbb{I}}

\newcommand{\tp}{^{\mathsf{T}}}
\newcommand{\x}{\otimes}

 \let\U\relax \DeclareMathOperator{\U}{U} \DeclareMathOperator{\SL}{SL} \DeclareMathOperator{\im}{im}\DeclareMathOperator{\No}{N} \DeclareMathOperator{\St}{St}
\newcommand{\psm}[1]{\begin{psmallmatrix} #1 \end{psmallmatrix}}
\newcommand{\w}{\psm{0&1\\1&0}}

\DeclareMathOperator{\Stab}{Stab}
\DeclareMathOperator{\orb}{orb}
\newcommand{\UBtransv}{\U_{B,\mathrm{transv}}}
\NewDocumentCommand{\IndTrans}{o}{\mathrm{U}_{\mathrm{Ind}\IfValueT{#1}{,#1}}}

\usepackage{amsthm}
\usepackage{zref-clever}
\zcsetup{nameinlink, cap, noabbrev}
\newcommand{\cref}[1]{\zcref{#1}}
\newcommand{\Cref}[1]{\zcref[S]{#1}}

\newtheoremstyle{plainstyle}{3pt}{3pt}{\normalfont}{}{\bfseries}{.}{ }{}

\theoremstyle{plain}
\newtheorem{definition}{Definition}[section]
\newtheorem{theorem}[definition]{Theorem}
\newtheorem{lemma}[definition]{Lemma}

\newtheorem{proposition}[definition]{Proposition}
\newtheorem{remark}[definition]{Remark}
\theoremstyle{plainstyle}
\newtheorem{observation}[definition]{Observation}

\AddToHook{env/definition/begin} {\zcsetup{countertype={definition=definition}}}
\AddToHook{env/observation/begin}{\zcsetup{countertype={definition=observation}}}
\AddToHook{env/theorem/begin}    {\zcsetup{countertype={definition=theorem}}}
\AddToHook{env/lemma/begin}      {\zcsetup{countertype={definition=lemma}}}
\AddToHook{env/corollary/begin}  {\zcsetup{countertype={definition=corollary}}}
\AddToHook{env/proposition/begin}{\zcsetup{countertype={definition=proposition}}}
\AddToHook{env/remark/begin}     {\zcsetup{countertype={definition=remark}}}

\zcRefTypeSetup{observation}{
  Name-sg = Observation , name-sg = observation ,
  Name-pl = Observations , name-pl = observations ,
}

\usepackage{array}
\usepackage{booktabs}

\usepackage{nomencl}
\makenomenclature

\usepackage{algorithm}
\usepackage{algpseudocode}

\usepackage{adjustbox}

\usepackage[backend=biber,style=alphabetic]{biblatex}
\tikzset{snake it/.style={decorate, decoration=snake}}

\newcommand{\defeq}{\vcentcolon=}

\newcolumntype{M}[1]{>{\centering\arraybackslash}m{#1}}

\makeatletter
\pgfdeclareshape{reptensor}{
    \savedmacro{\reptensor@width}{\def\reptensor@width{1.1cm}}
    \savedmacro{\reptensor@rectheight}{\def\reptensor@rectheight{0.6cm}}
    \savedmacro{\reptensor@arcHeight}{\def\reptensor@arcHeight{0.5cm}}

    \anchor{center}{
        \pgfpoint{0}{-0.2*\reptensor@rectheight}
    }
    \anchor{west}{
        \pgfpoint{-0.5*\reptensor@width}{-0.2*\reptensor@rectheight}
    }
    \anchor{east}{
        \pgfpoint{0.5*\reptensor@width}{-0.2*\reptensor@rectheight}
    }
    \anchor{south}{
        \pgfpoint{0}{-0.8*\reptensor@rectheight}
    }
    \anchor{text}{
        \pgfmathsetlength{\pgf@x}{-0.5*\wd\pgfnodeparttextbox}
        \pgfmathsetlength{\pgf@y}{-0.5*\ht\pgfnodeparttextbox}
    }

    \backgroundpath{
        \pgfpathmoveto{\pgfpoint{-0.5*\reptensor@width}{0.2*\reptensor@rectheight}}
        \pgfpathlineto{\pgfpoint{-0.5*\reptensor@width}{-0.8*\reptensor@rectheight}}
        \pgfpathlineto{\pgfpoint{0.5*\reptensor@width}{-0.8*\reptensor@rectheight}}
        \pgfpathlineto{\pgfpoint{0.5*\reptensor@width}{0.2*\reptensor@rectheight}}
        \pgfpatharc{0}{180}{0.5*\reptensor@width and 0.5*\reptensor@arcHeight + 0.2*\reptensor@rectheight}
        \pgfpathclose
    }
}
\makeatother

\usepackage{xstring} \usepackage{ifthen}

\newcommand{\drawsquare}[3][]{\StrBefore{#2}{,}[\x]
  \StrBehind{#2}{,}[\y]
  \def\option{#1}
  \def\labeltext{#3}
  \pgfmathsetmacro{\xl}{\x - \w/2}
  \pgfmathsetmacro{\xr}{\x + \w/2}
  \pgfmathsetmacro{\yt}{\y + \w/2}
  \pgfmathsetmacro{\yb}{\y - \w/2}
  \pgfmathsetmacro{\xlp}{\x - \w/2 + \p}
  \pgfmathsetmacro{\xrp}{\x + \w/2 - \p}
  \pgfmathsetmacro{\ytp}{\y + \w/2 - \p}
  \pgfmathsetmacro{\ybp}{\y - \w/2 + \p}

  \ifthenelse{\equal{\option}{dashed}}{\draw[dashed] (\xl,\yt) rectangle (\xr,\yb);
  }{\draw (\xl,\yt) rectangle (\xr,\yb);
  }

  \ifthenelse{\equal{\option}{cross}}{\draw (\xl,\yt) -- (\xr,\yb);
    \draw (\xr,\yt) -- (\xl,\yb);
  }{}

  \ifthenelse{\not\equal{\labeltext}{}}{\ifthenelse{\equal{\option}{cross}}{\draw[fill=white,draw=white] (\xlp,\ytp) rectangle (\xrp,\ybp);
    }{}
    \node at (\x,\y) {\(\labeltext\)};
  }{}
}

\newcommand\blfootnote[1]{\begingroup
  \renewcommand\thefootnote{}\footnote{#1}\addtocounter{footnote}{-1}\endgroup
}

\title{Quantum Fourier transform toolbox}

\author[1,2,*]{Carli Bruinsma}
\author[3,*]{Pietro M. Posta}
\author[1,4,*]{Joppe Stokvis}
\author[1,2,5]{\\Dmitry Grinko}
\author[1,2,5]{Maris Ozols}

\affil[1]{QuSoft, Amsterdam, The Netherlands}
\affil[2]{Institute for Logic, Language and Computation, University of Amsterdam, The Netherlands}
\affil[3]{Department of Mathematical Sciences, University of Copenhagen, Denmark}
\affil[4]{Mathematical Institute, Leiden University, The Netherlands}
\affil[5]{Korteweg-de Vries Institute for Mathematics, University of Amsterdam, The Netherlands}

\date{}

\begin{document}
\maketitle

\blfootnote{$^*$ These authors contributed equally}
\blfootnote{\raggedright Email addresses: \href{mailto:carlibruinsma@gmail.com}{carlibruinsma@gmail.com},
\href{mailto:pmp@math.ku.dk}{pmp@math.ku.dk},
\href{mailto:j.a.stokvis@math.leidenuniv.nl}{j.a.stokvis@math.leidenuniv.nl},
\href{mailto:d.grinko@uva.nl}{d.grinko@uva.nl}, and \href{mailto:marozols@gmail.com}{marozols@gmail.com}.}

\begin{abstract}
    Quantum Fourier transforms (QFTs) are essential primitives in quantum algorithms.
    While abelian groups admit efficient QFT circuits, with circuit size polynomial in the logarithm of the group order, efficient constructions are known for relatively few non-abelian families.
    We develop two new approaches to QFT circuit construction, based on Mackey theory and Clifford theory, respectively, and use them to show exponential improvement in circuit cost for specific group families.
    Using the Mackey-theoretic approach, we obtain explicit quantum circuits for the QFT over $\GL_2(\F_q)$ that scale polynomially in $\log q$, rather than polynomially in $q$.
    Using the Clifford-theoretic approach, we obtain QFT circuits for wreath products $F\wr S_n$, whose cost depends on the cost of a QFT over $F$ and the size of its representation registers.
    This removes the restriction $|F|=\operatorname{poly}(n)$ required by previous generic constructions and can yield exponential improvements when $F$ itself has an efficient QFT.
    Together, these methods provide new systematic tools to construct QFTs for broad classes of finite groups.
\end{abstract}

\tableofcontents

\section{Introduction}\label{chap: introduction}

Fourier analysis is one of the most fundamental tools in mathematics and its applications \cite{Terras1999}.
In particular, many quantum algorithms rely on it \cite{RevModPhys.82.1}.
A Quantum Fourier Transform (QFT) changes a quantum state so that hidden symmetries can be exposed through interference.
The abelian QFT is a central ingredient in Shor's factoring and discrete-logarithm algorithms \cite{Shor1994}.
Many problems, however, carry non-abelian symmetries, for which the corresponding transforms are more complicated.
Efficient circuits have been constructed for particular families, including certain metacyclic groups \cite{Hoyer97}, non-abelian $2$-groups with a cyclic normal subgroup of index two \cite{PRB99}, and the symmetric groups $S_n$ \cite{Beals,SnQFT_Speedup}.
Understanding how to implement these transforms is thus a basic question in quantum algorithm design: it often determines which symmetry-based computational ideas can be turned into efficient quantum circuits.
Efficient QFTs have previously been exploited in several different contexts.

\textit{Hidden subgroups and Fourier sampling.}
The QFT became central to quantum algorithms through the abelian hidden subgroup problem, which includes the periodicity problems behind factoring and discrete logarithms \cite{Shor1994,nonAbelianHSP,EttingerHoyerKnill04,RevModPhys.82.1}.
Historically, a major motivation for studying non-abelian QFTs was the hope that the same strategy would extend through weak and strong Fourier sampling \cite{Beals,HRT03,MRRS07}.
Weak sampling records only the irreducible-representation label and recovers hidden normal subgroups, while strong sampling also measures coordinates within a representation block and succeeds for some affine groups and hidden-shift problems \cite{HRT03,MRRS07}.
This approach did not produce a general algorithm for the non-abelian HSP.
In particular, even strong sampling of individual coset states fails for the symmetric-group instances associated with graph isomorphism \cite{FourierNotForHSPSn2005}.
Successful algorithms for other non-abelian groups use more elaborate measurements or exploit additional group structure, so an efficient QFT is only one part of an HSP algorithm \cite{Kuperberg2005,BaconChildsvanDam05,IvanyosSanselmeSantha,KroviRotteler08}.
Through its non-abelian variants, the HSP is also related to graph and code isomorphism, hidden shifts, and certain lattice problems \cite{EttingerHoyer1999Graph,DinhMooreRussell2011CodeEquivalence,Regev2004QuantumLattices,vDHIhiddenshift}.
The same group-action viewpoint also has implications in cryptographic constructions of quantum money \cite{BostanciNehoranZhandry2025}.

\textit{Representation-theoretic and algebraic quantities.}
Generalized phase estimation \cite{Harrow2005} uses QFTs to resolve representations into isotypic components and supports quantum algorithms for computing representation-theoretic multiplicities and characters \cite{PRXQuantum.5.010329,paperMultiplicities,BravyiCharacters}.
Closely related Schur transform methods apply this strategy to general plethysm and branching multiplicities \cite{christandl2026plethysmbqp}.
Moreover, QFTs are used in quantum algorithms for character sums and other arithmetic quantities \cite{vDS02,Bruin,Dam2004quantum}.

\textit{Quantum simulation and approximate QFTs.}
QFTs arise in digital simulations of gauge theories, where they connect useful descriptions of the system and motivate explicit circuits for the finite groups used in such models \cite{LammLawrenceYamauchi19,MurairiQFT,SigmaSU3_72}.
Previous work on the abelian QFT has produced low-cost, hardware-aware circuits and low-rank tensor-network representations of the transform \cite{Coppersmith1994,NamSuMaslov2020,BaumerSutterWoerner2025AQFT,DolgovKhoromskijSavostyanov2012,ChenStoudenmireWhite2023QFTEntanglement,ChenLindsey2026DFTMPO}.

\textit{Quantum machine learning.}
Recent work explores QFTs as inductive biases in quantum machine learning \cite{belis2026spectralmethods}.
Examples include spectral and IQP generative models, Fourier neural operators, and equivariant classifiers \cite{huang2026spectralborn,lerch2026iqpinitialization,jain2023quantumfouriernetworks,MarcandelliEtAl2025PartitionedQFNO,west2024rotational,ChirkovLobanov2026PixelTranslation,tuysuz2026quantumfouriergenerative,BanksEtAl2026QuditIQP}.
Other proposals use finite-group QFTs for group convolution or probabilistic models over permutations, although these works do not yet establish end-to-end quantum advantages \cite{CastelazoEtAl2022GroupConvolution,belis2026permutations}.

\subsection{Previous work}\label{subsec: intro known}

Fast Fourier transforms on finite groups were first developed classically \cite{Beth87,Clausen89,classicalFFT}.
The central technique is \emph{separation of variables}: the transform is computed recursively along a subgroup chain $\{e\}=G_0<G_1<\dots<G_m=G$, and each step combines the transforms on the cosets of $G_{i-1}$ using the matrices of the coset representatives, the \emph{twiddle factors} \cite{MaslenRockmore97,MaslenRockmoreAdapted}.
An adapted basis records how an irreducible representation branches along the chain, often by a path in the corresponding Bratteli diagram.
In this basis the twiddle factors become sparse and block structured.
This sparsity gives classical algorithms with $O(|G|\operatorname{polylog}|G|)$ operations for the symmetric groups \cite{ClausenBaum1993,Maslen1998} and for many other families \cite{MaslenRockmore97}.
The same architecture is natural for a quantum circuit because lower-level transforms can be reused coherently.
The quantum target is more demanding: a QFT circuit is \emph{efficient} if it uses $\operatorname{polylog}|G|$ gates, that is, polynomially many in the $\lceil\log_2|G|\rceil$ qubits of the input register.
\Cref{tab:_qft_landscape} lists the asymptotic gate counts of the known circuits for infinite families of groups, and we discuss them family by family below.

\begin{table}[htbp]
{\centering
\small
\setlength{\tabcolsep}{4pt}
\renewcommand{\arraystretch}{1.05}
\begin{tabular}{@{}>{\raggedright\arraybackslash}p{0.31\linewidth}
                  >{\raggedright\arraybackslash}p{0.10\linewidth}
                  >{\raggedright\arraybackslash}p{0.35\linewidth}
                  >{\raggedright\arraybackslash}p{\dimexpr0.24\linewidth-6\tabcolsep\relax}@{}}
\toprule
\textbf{Group} & \textbf{Order} & \textbf{Gates} & \textbf{Reference} \\
\midrule
\multicolumn{4}{@{}l}{\textit{Abelian groups}} \\
$\Z_{2^n}$
    & $2^n$ & $\Theta(n^2)$ & \cite{Coppersmith1994} \\
    & & $O(n\log^2 n\log\log n)$ & \cite{CleveWatrous2000} \\
    & & $O(n^{1+\delta})$, any $\delta>0$ & \cite{KahanamokuMeyerYao2024} \\
    & & $O(n\ell)$ & \cite{Coppersmith1994} \\
    & & $O(n(\log\ell)^2\log\log\ell)$ & \cite{AhokasCleveHales2003} \\
    & & $O(n\ell)$, depth $O(\log n+\log\log\frac{1}{\varepsilon})$ & \cite{CleveWatrous2000} \\
    & & $O(n\ell)$ $T$ gates & \cite{NamSuMaslov2020}\,\textsuperscript{a} \\
finite abelian $A$
    & $|A|$ & $\operatorname{poly}(\log|A|,\log\frac{1}{\varepsilon})$ & \cite{Kitaev1995} \\
    & & $\operatorname{poly}\log|A|$ & \cite{MoscaZalka2004} \\
\addlinespace[3pt]
\multicolumn{4}{@{}l}{\textit{Groups close to abelian}} \\
Pauli group $E_n$
    & $2\cdot4^n$ & $O(n)$ & \cite{Hoyer97} \\
metacyclic $\Z_m\rtimes\Z_q$, $q$ prime
    & $mq$ & $\operatorname{polylog}|G|$ & \cite{Hoyer97}\,\textsuperscript{b} \\
$2$-groups with a cyclic subgroup of index~$2$
    & $2^n$ & $O(n^2)$ & \cite{PRB99} \\
Heisenberg group $H_p$, $p$ prime
    & $p^3$ & $O(\log^3 p)$ & \cite{RadhakrishnanRottelerSen2005} \\
Weyl--Heisenberg $\Z_p^{n+1}\rtimes\Z_p^n$
    & $p^{2n+1}$ & $O(n\log p\log\log p)$ & \cite{KroviRotteler08}\,\textsuperscript{c} \\
Clifford-algebra group $\mathrm{CL}_n$
    & $2^{n+1}$ & $\operatorname{polylog}|G|$ & \cite{GenQFT} \\
metabelian $A\rtimes B$
    & --- & $\operatorname{polylog}|G|$ & \cite{GenQFT} \\
\addlinespace[3pt]
\multicolumn{4}{@{}l}{\textit{Symmetric groups and wreath products}} \\
$S_n$
    & $n!$ & $\widetilde O(n^4)$ & \cite{Beals,GenQFT}\,\textsuperscript{d} \\
    & & $O(n^4)$ & \cite{KS2013} \\
    & & $\widetilde O(n^{7/2})$ & \cite{SnQFT_Speedup}\,\textsuperscript{e} \\
    & & $\widetilde O(n^3)$ & \cite{BruinsmaGrinkoOzols2026} \\
$H\wr S_n$, $|H|=\operatorname{poly}(n)$
    & $|H|^n n!$ & $\operatorname{polylog}|G|$ & \cite{GenQFT} \\
\addlinespace[3pt]
\multicolumn{4}{@{}l}{\textit{Groups of Lie type: no efficient circuits known}} \\
$\SL_2(\F_q)$, $\mathrm{PSL}_2(\F_q)$
    & $\Theta(q^3)$ & $\Theta(q\log q)$ & \cite{GenQFT,DenneyMooreRussell2010} \\
rank $k$, e.g.\ $\GL_k(\F_q)$, $\SL_k(\F_q)$
    & $q^{\Theta(k^2)}$ & $q^{O(k)}$ & \cite{GenQFT} \\
\midrule
\multicolumn{4}{@{}l}{\textit{This work}} \\
$\GL_2(\F_q)$
    & $\Theta(q^4)$ & $\operatorname{poly}(\log q,\log\frac{1}{\varepsilon})$ & \cref{thm: GL2 QFT} \\
$F\wr S_n$
    & $|F|^n n!$ & $nC_F(\frac{\varepsilon}{2n})+\widetilde O(n^3+n^2L_F)$ & \cref{thm:wreath-main} \\
\bottomrule
\end{tabular}
\par\smallskip
\begin{minipage}{\linewidth}
\footnotesize
Bounds that depend on $\varepsilon$ approximate the QFT to error $\varepsilon$ in operator norm (in diamond norm for \cite{BruinsmaGrinkoOzols2026}); the other bounds are exact in the gate model of the respective source.
$\widetilde O$ hides polylogarithmic factors, including those in $1/\varepsilon$, and $\ell=\log(n/\varepsilon)$.
Rows without a group name continue the row above, and --- means that the source does not state the quantity.
$C_F$ and $L_F$ are defined in \cref{thm:wreath-main}.
\par
\textsuperscript{a}~Gates: for $\varepsilon\geq 2^{-n}$; the general bound is given in the text.
\textsuperscript{b}~Computes the QFT up to phases on the Fourier basis.
\textsuperscript{c}~Proof sketched in the source.
\textsuperscript{d}~Gates: as analyzed in \cite{BruinsmaGrinkoOzols2026}.
\textsuperscript{e}~Gates: counted in one- and two-qubit gates for a corrected version of the circuit \cite{BruinsmaGrinkoOzols2026}; \cite{SnQFT_Speedup} claim $O(n^3\log n)$ gates in their own gate model, or $O(n^3)$ with threshold gates.
\end{minipage}
}
\caption{{Asymptotic gate counts of known quantum Fourier transform circuits for infinite families of finite groups.}}
\label{tab:_qft_landscape}
\end{table}

\textit{Abelian groups.}
For $G=\Z_{2^n}$ the textbook circuit uses $n$ Hadamard gates and $n(n-1)/2$ controlled phase rotations on $n$ qubits, followed by a reversal of the qubit order \cite{Coppersmith1994,NielsenChuang}.
Coppersmith observed that keeping only the $O(\log(n/\varepsilon))$ largest rotations acting on each qubit leaves $O(n\log(n/\varepsilon))$ gates and approximates the transform to operator-norm error $\varepsilon$ \cite{Coppersmith1994}.
\textcite{CleveWatrous2000} reduced the exact gate count to $O(n\log^2 n\log\log n)$ by means of fast integer multiplication with $O(n)$ ancillas.
\textcite{KahanamokuMeyerYao2024} implement the same recursion with an ancilla-free multiplication circuit, which gives an exact transform with $O(n^{1+\delta})$ gates for every $\delta>0$ and no ancillas.
Combining truncation with fast multiplication, Ahokas, Cleve and Hales reduced the approximate gate count to $O(n(\log\log(n/\varepsilon))^2\log\log\log(n/\varepsilon))$ \cite{AhokasCleveHales2003,Ahokas2004}.
Cleve and Watrous also showed that depth $O(\log n+\log\log(1/\varepsilon))$ suffices for the approximate transform, and proved a matching $\Omega(\log n)$ lower bound for constant error.
Their log-depth circuit holds $O(\log(n/\varepsilon))$ copies of an $n$-qubit Fourier state and therefore uses $O(n\log(n/\varepsilon))$ qubits.
Later constructions reach logarithmic depth with $3n$ qubits under all-to-all connectivity, and with $4n$ qubits on a line using mid-circuit measurements \cite{BaumerSutterWoerner2025AQFT}.
In the Clifford${}+T$ model, \textcite{NamSuMaslov2020} replace the controlled rotations by additions into a reusable phase-gradient state.
Their circuit uses $O\bigl(n\log(n/\varepsilon)+\log(n/\varepsilon)\log(\log(n/\varepsilon)/\varepsilon)\bigr)$ $T$ gates on $n+O(\log(n/\varepsilon))$ qubits, which is $O(n\log(n/\varepsilon))$ whenever $\varepsilon\geq 2^{-n}$.
When the transform is immediately followed by a measurement, as in Shor's algorithm, every two-qubit gate can be replaced by a one-qubit gate conditioned on earlier measurement outcomes \cite{GriffithsNiu1996}.
For arbitrary cyclic orders, and hence for all finite abelian groups, \textcite{Kitaev1995} gave approximate circuits of size $\operatorname{poly}(\log|G|,\log(1/\varepsilon))$.
\textcite{MoscaZalka2004} made them exact using amplitude amplification, in a gate model whose rotation angles are computed efficiently from the group order.

\textit{Groups close to abelian.}
The first non-abelian circuits are due to \textcite{Hoyer97}.
He treated metacyclic groups $\Z_m\rtimes\Z_q$ with $q$ prime, a class that contains dihedral and affine groups.
His circuit consists of two cyclic QFTs and a diagonal twiddle layer, and it computes the transform only up to phases on the Fourier basis.
The same paper gives an exact circuit for the Pauli group $E_n$ of order $2\cdot 4^n$ that uses $2n+1$ qubits and $O(n)$ Hadamard and controlled-Hadamard gates \cite[Thm.~18]{Hoyer97}.
\textcite{PRB99} turned the classical recursion for solvable groups into a systematic quantum method.
They applied it to the four families of non-abelian $2$-groups with a cyclic subgroup of index two, namely the dihedral, quaternion, quasidihedral and $QP_{2^n}$ groups.
For groups of order $2^n$ this gives exact circuits with $O(n^2)$ one-qubit and CNOT gates on $n$ qubits.
\textcite{RadhakrishnanRottelerSen2005} gave an $O(\log^3 p)$ circuit for the Heisenberg group of order $p^3$ on $3\lceil\log_2 p\rceil$ data qubits.
\textcite{KroviRotteler08} sketched the extension to the Weyl--Heisenberg groups $\Z_p^{n+1}\rtimes\Z_p^n$ of order $p^{2n+1}$, which generalizes H\o yer's Pauli-group circuit and uses $O(n\log p\log\log p)$ gates.

\textit{A generic framework.}
\textcite{GenQFT} showed that the separation-of-variables recursion can be quantized in general.
They call $G$ \emph{polynomially uniform} with respect to a subgroup chain if a classical algorithm can navigate the Bratteli diagram and factor coset representatives over a fixed generating set in $\operatorname{polylog}|G|$ time.
For such $G$ they obtain circuits of size $\operatorname{poly}(I\cdot D\cdot M\cdot\log|G|)$.
Here $I=\max_i[G_i:G_{i-1}]$ is the largest index in the chain, $D$ is the adapted diameter of the generating set, and $M$ is the largest multiplicity in the restrictions from $G_i$ to $G_{i-1}$.
The factor $I$ arises because the circuit sums over the cosets of $G_{i-1}$ in $G_i$ one at a time.
For homothetic extensions by abelian groups, a parallel summation in the style of Coppersmith's circuit removes this factor, and a related argument covers the metabelian groups.
Besides recovering the constructions above, the framework gives efficient circuits for the wreath products $H\wr S_n$ with $|H|=\operatorname{poly}(n)$, where $I=\max(n,|H|)$, $D=O(n^2\log|H|)$ and $M=O(|H|)$.
It also covers the Clifford-algebra groups $\mathrm{CL}_n$ generated by $x_1,\dots,x_n$ with $x_i^2=1$ and $x_ix_j=-x_jx_i$.\footnote{That should not be confused with the \emph{other} Clifford group.}
To simplify the description, the circuit stores the coset representatives and the two Bratteli paths in separate registers, which doubles the number of qubits of a compact encoding.

\textit{Symmetric groups.}
\textcite{Beals} gave the first efficient circuit for $S_n$, using the Young orthogonal basis and the adjacent transpositions as generators.
Through the framework of \cite{GenQFT}, with $I=n$, $D=O(n^2)$ and $M=2$, this construction uses $\widetilde O(n^4)$ gates \cite{BruinsmaGrinkoOzols2026}.
\textcite{KS2013} gave the first fully explicit circuit, which sums over cosets in parallel rather than one at a time, with depth $O(n^4)$ on $O(n\log n)$ qubits.
They later claimed $O(n^3\log n)$ elementary gates, or $O(n^3)$ with threshold gates \cite{SnQFT_Speedup}.
Concurrent work by a subset of the present authors gives an elementary-gate construction of the $S_n$ QFT and a detailed resource analysis \cite{BruinsmaGrinkoOzols2026}.
Counted in one- and two-qubit gates, a corrected version of the circuit of \cite{SnQFT_Speedup} uses $\widetilde O(n^{7/2})$ gates, $\widetilde O(n^3)$ depth and $\widetilde O(n^{3/2})$ qubits.
Beals' construction with a suitable choice of transversals needs only $\widetilde O(n^3)$ gates, $\widetilde O(n^3)$ depth and $\widetilde O(n^{3/2})$ qubits.
Compared with \cite{KS2013}, this saves a factor $\widetilde O(n)$ in gates at the cost of a factor $\widetilde O(n^{1/2})$ in qubits.

\textit{Groups of Lie type.}
Beyond these families the known bounds are no longer polylogarithmic.
For $\GL_k(\F_q)$, $\SL_k(\F_q)$, their projective versions, and the other finite groups of Lie type, the parabolic subgroup chain has index $I=q^{k-1}$ and multiplicity $M=q^{O(k)}$.
The theorem of \cite{GenQFT} therefore gives circuits of size $q^{O(k)}=|G|^{O(1/k)}$.
For fixed $q$ this equals $\exp(O(\sqrt{\log|G|}))$.
This bound is subexponential in $|G|$ and improves on the best classical bound $|G|^{1+\Theta(1/k)}$, but it is superpolynomial in the number of qubits.
For fixed rank and growing $q$ it is polynomial in $q$.
Childs and van Dam \cite{RevModPhys.82.1} list the classical groups, such as $\GL_n(\F_q)$, among the few important groups for which efficient QFTs are not known.
\textcite{DenneyMooreRussell2010} note that the natural quantum version of the fastest classical transform for $\SL_2(\F_q)$ and $\mathrm{PSL}_2(\F_q)$ takes time $\Theta(q\log q)$, which is polynomial rather than polylogarithmic in $|G|=\Theta(q^3)$.
This cost was the original motivation for their hidden subgroup algorithm for these groups, which avoids the transform by reducing to the affine group $\mathrm{AGL}(1,q)$.
\textcite{GenQFT} had already asked whether the finer classical techniques behind the fastest transforms for $\SL_2(\F_q)$, $S_n$ and its wreath products can be adapted to the quantum setting.
To our knowledge, no circuit of size $\operatorname{polylog}|G|$ was previously known for any infinite family of non-abelian finite groups of Lie type.

\textit{Finite subgroups of $\mathrm{SU}(2)$ and $\mathrm{SU}(3)$.}
Digital simulations of lattice gauge theories replace the gauge group by a finite subgroup and apply a QFT on every link in every Trotter step to diagonalize the kinetic term \cite{LammLawrenceYamauchi19}.
Without a fast transform, the QFT dominates the cost of such simulations \cite{SigmaSU3_72}.
\textcite{MurairiQFT} constructed explicit fast circuits for the binary tetrahedral and binary octahedral groups $\mathbb{BT}$ and $\mathbb{BO}$, and for $\Delta(27)$, $\Delta(54)$ and $\Sigma(36\times3)$, following the method of \cite{PRB99}, and reported $T$ counts after Clifford${}+T$ synthesis.
For $\Sigma(36\times3)$ their circuit is orders of magnitude cheaper than a compiled dense unitary of the transform, at the price of additional qubits, and it substantially lowers the total cost of a shear-viscosity simulation.
For $\Sigma(72\times3)$ only a transform without this fast structure has been reported \cite{SigmaSU3_72}.
For the larger groups $\Sigma(216\times3)$ and $\Sigma(360\times3)$, which are relevant to quantum chromodynamics, only extrapolated cost estimates are available \cite{SigmaSU3_72}.
Since these are individual groups of fixed order, \cref{tab:_qft_landscape}, which lists asymptotic bounds for infinite families, does not include them.

\textit{Beyond group algebras.}
\textcite{FoxmanNehoranDing} extended the separation-of-variables approach to semisimple algebras and gave circuits of size $\operatorname{poly}(n,\log d,\log(1/\varepsilon))$ for the partition, Brauer and walled Brauer algebras.
Since the Fourier transform of such an algebra need not be unitary, they implement a unitary approximation that is accurate when the parameter $d$ is large.

\textit{Exact and approximate transforms.}
The smallest known circuits for the $\Z_{2^n}$ transform are approximate: dropping small rotations gives $O(n(\log\log n)^2\log\log\log n)$ gates at polynomially small error \cite{AhokasCleveHales2003}, compared with $O(n\log^2 n\log\log n)$ gates for the exact transform \cite{CleveWatrous2000}.
The same transform also admits accurate low-rank matrix-product-operator approximations \cite{ChenStoudenmireWhite2023QFTEntanglement,ChenLindsey2026DFTMPO}.
For non-abelian groups such approximations are largely unexplored, and \textcite{MurairiQFT} propose them as an open direction.
For semisimple algebras they are forced, since the transform need not be unitary \cite{FoxmanNehoranDing}.
In this work we construct exact circuits.
The general constructions in \cref{sec: first induction,sec: beals,sec: Mackey induction transform,sec: clifford theory} express the transform as an explicit product of smaller unitaries, without approximation.
The error $\varepsilon$ in \cref{thm: GL2 QFT} comes only from compiling three kinds of operations into a finite gate set: cyclic QFTs of arbitrary order, discrete logarithms in finite fields, and rotations whose angles are efficiently computable functions of the input labels.
In a gate model with exact rotations of efficiently computable angles, all three are exact \cite{MoscaZalka2004}, and so is our circuit for $\GL_2(\F_q)$.
Similarly, the error in \cref{thm:wreath-main} comes only from the QFT over $F$ and from such rotations, so the circuit for $F\wr S_n$ is exact whenever the QFT over $F$ is.
In both theorems the compilation overhead is polynomial in $\log(1/\varepsilon)$.

\textit{Mackey and Clifford theory in previous constructions.}
Several circuits above recurse through a normal subgroup $N\lhd G$ and use the action of $G$ on the irreducible representations of $N$.
\textcite{PRB99} base their recursion on Clifford's theorem for $N$ of prime index: each irreducible representation of $N$ either extends to $G$ or induces an irreducible representation of $G$.
The same extension method underlies the circuits of \cite{MurairiQFT,SigmaSU3_72}.
The Heisenberg-group circuit of \cite{RadhakrishnanRottelerSen2005} recurses through the normal subgroup $\Z_p^2$ and uses the action of $H_p$ on its characters, and the homothetic extensions of \cite{GenQFT} are defined through the same action.
All these constructions require the quotient $G/N$ to be abelian.
The little group algorithm of \cref{sec: clifford theory} removes this restriction whenever the irreducible representations of $N$ extend to their inertia groups; for $F\wr S_n$ the quotient is $S_n$.
On the Mackey side, the second induction relation in the circuit of \cite{KS2013} follows from Mackey theory \cite{BruinsmaGrinkoOzols2026}.
Both theories also have classical precedents in fast Fourier transforms on finite groups.
\textcite{Rockmore1990} used Clifford theory to obtain fast transforms when $G/N$ is abelian, a class that contains all metabelian groups.
For the wreath products $F\wr S_n$, \textcite{Rockmore1995} built adapted bases from the Clifford-theoretic description of their irreducible representations and from the Mackey subgroup theorem, and obtained fast transforms for this family.
\textcite{MaslenRockmore2000} decomposed the Fourier transform over the double cosets $H\backslash G/K$ of two subgroups, which generalizes the Cooley--Tukey algorithm, and applied this to $\SL_2(\F_q)$ through its Bruhat decomposition.

Most known non-abelian QFT constructions recurse along a subgroup chain.
For each inclusion $H<G$, the QFT over $H$ is followed by an induction transform that converts the $H$-adapted decomposition into the Fourier basis of $G$.
We study three ways to implement this induction transform.
The first is the direct subgroup-adapted construction of \cite{Beals,GenQFT}, which treats the cosets of $H$ one at a time and serves as our baseline (\cref{sec: beals}).
It was recently revisited for the symmetric group \cite{PRXQuantum.5.010329,paperMultiplicities}.
The other two build circuits for the induction transform from Mackey theory and Clifford theory, respectively.
The Mackey construction splits the transform over the double cosets $H\backslash G/H$ and reduces it to induction transforms for the intersection subgroups $H\cap\omega H\omega^{-1}$ (\cref{sec: Mackey induction transform}).
When $H$ is normal in $G$, the little group construction instead uses the orbits of $G$ on the irreducible representations of $H$ and their inertia groups (\cref{sec: clifford theory}).
The three methods exploit different structure.
The direct construction suits inclusions of small index.
The Mackey construction suits inclusions with few double cosets whose intersection subgroups already lie in the chain, as for $\GL_2(\F_q)$.
The little group construction suits normal subgroups whose irreducible representations extend to their inertia groups, as for $F\wr S_n$.

\subsection{Main results}\label{subsec: intro results}

Mackey theory and Clifford theory describe how group representations behave under induction and restriction.
To our knowledge, they have not previously been developed as general circuit frameworks for QFTs.
We give a self-contained circuit formulation of subgroup induction, recover the direct subgroup-recursive construction \cite{Beals}, and derive two new implementations of the induction transform.
We then apply these constructions to $\GL_2(\F_q)$ and $F\wr S_n$, where generic subgroup recursion leaves specific computational bottlenecks.

\begin{theorem}[Mackey algorithm, informal]
For every subgroup inclusion satisfying the standing assumptions of \cref{sec: first induction}, and with the lower induction transforms, untwisting maps, and adapted branching data appearing in the Mackey decomposition supplied, \cref{alg: Mackey induction transform} implements the induction transform for every irreducible representation of the subgroup.
The circuit decomposes the calculation over the double cosets of $H$ in $G$, invokes smaller induction transforms associated with the intersection subgroups, and finishes with explicitly defined unitary blocks on multiplicity spaces.
\end{theorem}

The exact circuit identity and the multiplicity-space blocks are given in \cref{thm: Mackey induction transform circuit}.
As an application of this framework, we consider $G = \GL_2(\F_q)$.
The Bruhat decomposition has only two relevant double cosets, and the required intersection subgroup is already present in the chain from the diagonal torus through the Borel subgroup.
This structure yields the following result.

\begin{theorem}[QFT over $\GL_2(\F_q)$]
For every prime power $q$ and every $\varepsilon>0$, there is a quantum circuit of size $\operatorname{poly}(\log q,\log(1/\varepsilon))$ that implements the QFT over $\GL_2(\F_q)$ to operator-norm error at most $\varepsilon$.
\end{theorem}

This is proved in \cref{thm: GL2 QFT}.
In contrast, the general finite-group construction has a subgroup-index dependence that is polynomial in $q$ for this family \cite{GenQFT}.
The two Bruhat cells make the Mackey decomposition short, but they do not by themselves make the circuit efficient: one multiplicity block has dimension $q$, and synthesizing it as a dense unitary would cost polynomially many gates in $q$.

\begin{theorem}[Little group algorithm, informal]
Let $N$ be a normal subgroup of $G$, and suppose that every irreducible representation of $N$ extends to its inertia group.
Then \cref{alg:little-group-qft} implements the QFT over $G$ in the little group basis.
\end{theorem}

The circuit uses the QFT over $N$, QFTs over the little groups, reversible orbit and factorization procedures, controlled orbit-transport intertwiners, and controlled extension matrices.
Correctness is proved in \cref{thm:little-group-qft-correctness}.

As an application of this framework, we consider $N=F^n$ inside $F\wr S_n$, where $F$ is a finite group for which an efficient QFT is known.
The little groups in this case are Young subgroups and the required extensions are explicit.
We write $C_F$, $D_F$ and $Q_F$ for the circuit cost, depth and number of qubits of a QFT over $F$, and $L_F$ for the number of qubits required to encode irreducible-representation registers and elements of $\widehat F$.

\begin{theorem}[QFT over a permutation wreath product, informal]
For every $n\geq 1$ and $\varepsilon>0$, the little group circuit implements the QFT over $F\wr S_n$ to operator-norm error at most $\varepsilon$ with gate count
$\widetilde O\!\left(nC_F + n^3+n^2L_F\right)$, depth $\widetilde O(D_F + n^3+n^2L_F)$ and qubit count of $  \widetilde O\!\left(nQ_F+n^2\right).$
\end{theorem}

The complete statement and proof can be found in \cref{thm:wreath-main} and \cref{lem:wreath-qubit-count}.
The Beals generic construction gives an efficient QFT over $F\wr S_n$ under the condition $|F|=\operatorname{poly}(n)$ \cite{GenQFT}.
Our construction does not have this restriction, and our bound depends instead on the complexity of the QFT for the group $F$ and the size of its representation registers.
It therefore covers all cases in which $|F|$ is not polynomial in $n$ while $C_F$, $D_F$, and $L_F$ remain efficient in the chosen parameters.
For instance, taking $F=S_m$ gives a circuit polynomial in the two independent wreath-product parameters $m$ and $n$.

\subsection{Organization of the paper}\label{subsec: intro outline}

\cref{sec: preliminaries} introduces the specific representation-theoretic background and circuit conventions.
\cref{sec: first induction} develops subgroup induction and the induction transform, and \cref{sec: beals} describes the canonical Beals algorithm for implementing the induction transform.
\cref{sec: Mackey induction transform} describes the Mackey algorithm, followed by the $\GL_2(\F_q)$ application in \cref{sec: GL2}.
\cref{sec: clifford theory} develops the little group QFT, which is specialized to wreath products in \cref{sec: wreath}.

\section{Preliminaries and notation}\label{sec: preliminaries}
In this section we fix the conventions for representations, their duals, and Fourier transforms, and refer to Ref.~\cite{serre} for further background.
Throughout this manuscript, groups are always assumed to be finite and representations are finite-dimensional complex unitary representations. We use the Dirac notation for vectors and their duals.

Our notation for matrix-entries of a linear map $A$ between spaces with specified orthonormal bases is
\begin{equation}\label{eq: matrix entry convention}
[A]^j_i:=\bra j A\ket i=A_{ji},
\end{equation}
so the upper index is the output or row index, and the lower index is the input or column index.

\subsection{Representations, labels, and dual actions}\label{subsec: representations and duals}

A \emph{representation} of a group~$G$ on a complex Hilbert space $V$ is a group homomorphism $\Rep:G\to\U(V)$, where $\U(V)$ are all unitary operators on $V$.
We call the underlying vector space $V$ a $G$- or \emph{$\C[G]$-module}, and sometimes we refer to it as the \emph{representation space} for $\Rep$.
Sometimes, with a slight abuse of notation we will refer to the pair $(\Rep,V)$ as a \emph{representation}.
A subspace $W\subseteq V$ is called \emph{$G$-invariant} if $\Rep(g)W = W$ for all $g\in G$.
A representation $\Rep$ on $V$ is called \emph{irreducible} if $V\neq \{0\}$ and its only invariant subspaces are $\{0\}$ and $V$.
Representations $(\Rep,V)$ and $(\Rep',V')$ are \emph{equivalent}, and we write $\Rep \cong \Rep'$ and $V\cong V'$, if there is an invertible linear map $A:V\to V'$ satisfying $A\Rep(g)=\Rep'(g)A$ for every $g\in G$.
We call such a map an \emph{intertwiner} between $V$ and $V'$.
For unitary representations, the intertwiner $A$ can always be chosen to be unitary.

The \emph{unitary dual} of a group $G$, denoted by $\hatG$, is the set of equivalence classes of irreducible unitary representations of $G$.
It is finite, and its cardinality is the same as that of the set of conjugacy classes of $G$.
The elements of $\hatG$ are called \emph{labels}, and we generally denote them using Greek letters.
For each label $\lambda\in\hatG$, we fix a $G$-module $V_\lambda$, a representation $\Rep_\lambda:G\to\U(V_\lambda)$, and an orthonormal basis $\mathcal B_\lambda:=\{\ket{i_\lambda}: i=1,\dots,d_\lambda\}$ for $V_\lambda$, where $d_\lambda:=\dim V_\lambda$.
We denote the \emph{dual} space of $V_\lambda$ by $V_\lambda^\vee$ and its basis by $\mathcal B_\lambda^\vee:=\{\ket{i_\lambda^\vee}: i=1,\dots,d_\lambda\}$.

In \cref{sec: first induction,sec: GL2,sec: beals,sec: Mackey induction transform} we choose \emph{Gelfand--Tsetlin} bases for irreducible representations. See \cref{subsec: bratteli} for details on this notation.

\begin{definition}[Dual representation]\label{def: dual representation}
Let $(\Rep,V)$ be a representation of $G$, and let $V^\vee$ be the dual space with basis $\mathcal B^\vee = \{\ket{j^\vee} : j = 1,\dotsc, \dim V\}$.
The \emph{dual} or \emph{contragredient representation} $\Rep^\vee:G\to\U(V^\vee)$ is defined on elements of $\mathcal B^\vee$ by
\begin{equation}\label{eq:dual-irrep-convention}
\Rep^\vee(g):\ket{j^\vee}\longmapsto\Rep(g\inv)^{\mathsf T} \ket{j^\vee},\quad g\in G.
\end{equation}
\end{definition}
The dual of an irreducible representation is itself irreducible.
Finite group representations are completely reducible, in the following sense:

\begin{theorem}[Maschke]\label{thm: Maschke}
Every finite-dimensional complex representation $V$ of a finite group $G$ is a direct sum of irreducible representations.
In our unitary setting, the decomposition can be chosen to be orthogonal and written as
\begin{equation}
V\cong\bigoplus_{\lambda\in\hatG}
\underbrace{V_\lambda\oplus V_\lambda\oplus\cdots\oplus V_\lambda}_{m_\lambda\text{ times}}
\cong\bigoplus_{\lambda\in\hatG}V_\lambda\otimes\C^{m_\lambda}.
\end{equation}
Here $m_\lambda\geq0$ is the multiplicity of $V_\lambda$, a summand with $m_\lambda=0$ is the zero space, and $G$ acts trivially on the multiplicity space $\C^{m_\lambda}$.
\end{theorem}

\subsection{Regular representations and the Fourier transform}\label{subsec: Sn regular reps}
Consider the group algebra $\C[G]$, equipped with the inner product for which the group-element basis $\{\ket g:g\in G\}$ is orthonormal.
\begin{definition}[Regular representations]
The left and right regular representations are the commuting unitary actions on $\C[G]$ given on the group-element basis by
\begin{equation}
L(x)\ket g:=\ket{xg},
\qquad R(y)\ket g:=\ket{gy\inv},
\qquad x,y,g\in G.
\end{equation}
Together they define the $G\times G$ action on $\C[G]$ given by $(x,y)\ket g=\ket{xgy\inv}$.
\end{definition}

\begin{theorem}[Peter--Weyl]\label{thm: finite Peter Weyl}
Under the left and right regular actions, $\C[G]$ has the multiplicity-free decomposition into irreducible $G\times G$-modules
\begin{equation}\label{eq: group algebra Peter-Weyl decomposition}
\C[G]\cong\bigoplus_{\lambda\in\hatG}V_\lambda\otimes V_\lambda^\vee,
\end{equation}
where $(x,y)\in G\times G$ acts in the $\lambda$-sector by $\Rep_\lambda(x)\otimes\Rep_\lambda^\vee(y)$.
In particular, $V_\lambda$ occurs in the left regular representation with multiplicity $d_\lambda$, and $|G|=\sum_{\lambda\in\hatG}d_\lambda^2$.
The right side of \cref{eq: group algebra Peter-Weyl decomposition} has basis
\begin{equation}\label{eq: Peter-Weyl basis}
\{\ket{\lambda,i_\lambda,j_\lambda^\vee}:\lambda\in\hatG,\ i,j=1,\dots,d_\lambda\},
\end{equation}
where $\ket{\lambda,i_\lambda,j_\lambda^\vee}$ denotes $\ket{i_\lambda}\otimes\ket{j_\lambda^\vee}$ in the summand labeled by $\lambda$.
\end{theorem}

We can now define Fourier transforms as unitary intertwiners for the two regular actions:
\begin{definition}[Fourier transform]\label{def: finite group Fourier transform}
A Fourier transform for $G$ is a unitary map $\Fou_G:\C[G]\to\bigoplus_{\lambda\in\hatG}V_\lambda\otimes V_\lambda^\vee$ that intertwines the left and right regular actions with the actions in \cref{thm: finite Peter Weyl}, namely
\begin{equation}\label{eq: Fourier transform block-diagononalizes}
\begin{aligned}
\Fou_G L(x)\Fou_G^\dagger
&=\bigoplus_{\lambda\in\hatG}\Rep_\lambda(x)\otimes I_{V_\lambda^\vee},\\
\Fou_G R(y)\Fou_G^\dagger
&=\bigoplus_{\lambda\in\hatG}I_{V_\lambda}\otimes\Rep_\lambda^\vee(y).
\end{aligned}
\end{equation}
\end{definition}

Up to a scalar phase in each summand, the matrix entries of a Fourier transform for $G$ are determined by those of the chosen irreducible representations.

\begin{proposition}[Matrix formula and phase convention]\label{prop: Fourier transform intertwining}
For the chosen representatives and orthonormal bases $\mathcal B_\lambda$, $\lambda\in\hatG$, a Fourier transform is given by
\begin{equation}\label{eq: Fourier transform}
\Fou_G\ket g=\sum_{\lambda\in\hatG}\sqrt{\frac{d_\lambda}{|G|}}
\sum_{i,j=1}^{d_\lambda}[\Rep_\lambda(g)]^i_j\ket{\lambda,i_\lambda,j_\lambda^\vee}.
\end{equation}
\end{proposition}

Each Fourier transform in \cref{def: finite group Fourier transform} differs from this one by a scalar phase in each $\lambda$-summand.
Throughout this paper, $\Fou_G$ denotes the choice in \cref{eq: Fourier transform}, with all these phases equal to one.

\begin{proof}[Proof of \cref{thm: finite Peter Weyl,prop: Fourier transform intertwining}]
Let $\chi_\lambda(g):=\tr(\Rep_\lambda(g))$ and $\chi_L(g):=\tr(L(g))=|G|\delta_{g,1_G}$ be the characters of $\Rep_\lambda$ and $L$, respectively.
By \cref{thm: Maschke} and character orthogonality, the multiplicity of $V_\lambda$ in the left regular representation is
\begin{equation}
m_\lambda=\frac1{|G|}\sum_{g\in G}\chi_L(g)\overline{\chi_\lambda(g)}
=\overline{\chi_\lambda(1_G)}=d_\lambda.
\end{equation}
Thus $|G|=\sum_\lambda d_\lambda^2$, so the matrix $\Fou_G$ defined by \cref{eq: Fourier transform} is square.
Schur orthogonality gives
\begin{equation}
\sum_{g\in G}[\Rep_\lambda(g)]^i_j\,
\overline{[\Rep_\mu(g)]^{i'}_{j'}}
=\frac{|G|}{d_\lambda}\delta_{\lambda\mu}\delta_{ii'}\delta_{jj'}.
\end{equation}
Hence its rows are orthonormal, and $\Fou_G$ is unitary.

The representation law and the dual-matrix convention give
\begin{align*}
[\Rep_\lambda(xg)]^i_j
&=\sum_a[\Rep_\lambda(x)]^i_a[\Rep_\lambda(g)]^a_j,\\
[\Rep_\lambda(gy\inv)]^i_j
&=\sum_b[\Rep_\lambda(g)]^i_b[\Rep_\lambda^\vee(y)]^j_b.
\end{align*}
Substituting these identities into \cref{eq: Fourier transform} proves the two intertwining identities in \cref{eq: Fourier transform block-diagononalizes}.
Thus \cref{eq: Fourier transform} defines a Fourier transform, proving \cref{prop: Fourier transform intertwining}, and realizes the claimed $G\times G$ decomposition.
Each summand is irreducible as an external tensor product of irreducible representations.
The summands are pairwise inequivalent: restriction to $G\times\{1_G\}$ gives $d_\lambda$ copies of $V_\lambda$, which determines $\lambda$.
This proves \cref{thm: finite Peter Weyl}.

Finally, if $\Fou_G^{\prime}$ is another Fourier transform, then $\Fou_G^{\prime}\Fou_G^\dagger$ commutes with the $G\times G$ action on this multiplicity-free direct sum.
Schur's lemma makes it a scalar on each summand, and unitarity makes each scalar have absolute value one.
This proves the stated freedom to choose one phase per summand.
\end{proof}

\subsection{Restriction and induction}

The relationship between representations of a group $G$ and those of its subgroup $H \leq G$ is central to our work.
A simple way of constructing representations of $H$ from those of $G$ is by restriction.

\begin{definition}[Restricted representation]\label{def: restricted representation}
Let $H\leq G$ and let $(\Rep,V)$ be a representation of $G$.
The \emph{restriction} $\Res_H^G\Rep$ is the representation of $H$ on $V$ obtained by restricting the action $\Rep$ from $G$ to $H$.
\end{definition}

Induction is an operation that goes the other way, constructing representations of $G$ from those of $H$.
To describe induction explicitly, we first need to choose coset representatives.

\begin{definition}[Transversal]
Let $H\leq G$.
A \emph{left transversal} $\T$ for $H$ in $G$ is a subset of $G$ that contains exactly one representative of each \emph{left coset} $gH=\{gh:h\in H\}$.
A \emph{right transversal} contains exactly one representative of each \emph{right coset} $Hg=\{hg:h\in H\}$.
\end{definition}

If $H\unlhd G$, the left and right cosets of $H$ coincide.
Consequently, a subset of $G$ is a left transversal of $H$ if and only if it is a right transversal.

\begin{definition}[Induced representation]\label{def: induced representation}
Let $H\leq G$, choose a left transversal $\T$ containing $1_G$, and let $(\Rep,V)$ be a representation of $H$.
The induced representation $\Ind_H^G\Rep$ acts on the space
\begin{equation}
    \Ind_H^G V:=\C[\T]\otimes V,
\end{equation}
with the action of $g \in G$ given by
\begin{equation}
    (\Ind_H^G\Rep)(g)(\ket t\otimes\ket i)
    =\ket{t_g}\otimes(\Rep(h_g)\ket i),
    \qquad t \in \T, \ket i \in V,
\end{equation}
where $t_g\in\T$ and $h_g\in H$ are uniquely determined by $gt=t_gh_g$.
For an orthonormal basis $\{\ket i\}_{i=1}^{\dim V}$ of $V$, the matrix entries in the corresponding product basis are
\begin{equation}\label{eq: induced representation matrix entries}
    [(\Ind_H^G\Rep)(g)]^{s,j}_{t,i}
    =\delta_{s,t_g}[\Rep(h_g)]^j_i,
    \qquad s,t\in\T,\qquad i,j\in \{1,2,\dots,\dim V\}.
\end{equation}
\end{definition}

Thus $\Ind_H^G V$ consists of $[G:H]$ copies of $V$ indexed by the chosen coset representatives, and has dimension $[G:H]\dim V$.
The factorization in the definition uses a left transversal; given a right transversal, its set of inverses is a left transversal and gives the same construction.
Different choices of transversal for $H$ yield $G$-equivalent induced representations.
In the dual of the chosen coset basis, the dual induced representation satisfies $(\Ind_H^G\Rep)^\vee\cong\Ind_H^G(\Rep^\vee)$, with the same transversal $\T$ and the dual basis of $V$.

The relation between restriction and induction is governed by Frobenius reciprocity:
\begin{theorem}[Frobenius reciprocity \protect{\cite[Section~7.1]{serre}}]\label{thm: Frobenius reciprocity}
Let $H\leq G$, and let $W$ and $V$ be irreducible representations of $H$ and $G$, respectively.
The multiplicity of $V$ in $\Ind_H^G W$ equals the multiplicity of $W$ in $\Res_H^G V$.
\end{theorem}

\subsection{Successors and predecessors of a representation}\label{subsec: notation}

To conveniently denote the set of all irreps contained in a given (reducible) representation, we introduce the concept of irreducible support.

\begin{definition}[Irreducible support]\label{def: irreducible support}
For a representation $\Rep:H\to\U(V)$, define its \emph{irreducible support} by
\begin{equation}
    \N(\Rep):=\{\lambda\in\widehat H:\Hom_H(V_\lambda,V)\neq0\}.
\end{equation}
When the representation is specified by its module or by a label, we use the same notation for its support.
\end{definition}

If $H \leq G$, the restriction from $G$ to $H$ can be used to establish a correspondence between irreps of $G$ and $H$.

\begin{definition}[Successors and predecessors]\label{def: successors and predecessors}
For $H\leq G$, define the \emph{successors} of $\lambda\in\widehat H$ as
\begin{equation}
    \N_{H\leq G}^{+}(\lambda)
    :=\{\mu\in\widehat G:\lambda\in\N(\Res_H^G\Rep_\mu)\}
    \subseteq \widehat G
\end{equation}
and the \emph{predecessors} of $\mu\in\widehat G$ as
\begin{equation}
    \N_{H\leq G}^{-}(\mu)
    :=\{\lambda\in\widehat H:\lambda\in\N(\Res_H^G\Rep_\mu)\}
    \subseteq \widehat H.
\end{equation}
When $H\leq G$ is clear from context, we abbreviate these sets to $\N^+(\lambda)$ and $\N^-(\mu)$.
\end{definition}

This effectively establishes a bipartite graph between $\widehat H$ and $\widehat G$, in which an edge between $\lambda \in \widehat H$ and $\mu \in \widehat G$ is present whenever one of these two equivalent conditions holds:
\begin{equation}
    \mu\in\N_{H\leq G}^{+}(\lambda)
    \quad\Longleftrightarrow\quad
    \lambda\in\N_{H\leq G}^{-}(\mu).
\end{equation}

\begin{proposition}[Multiplicity-free induction]\label{prop: multiplicity-free induction}
Let $H\leq G$ and assume that $\Res_H^G V_\mu$ is multi\-plicity-free for every $\mu\in\widehat G$.
Then, for every $\lambda\in\widehat H$, the induced representation $\Ind_H^G V_\lambda$ is multiplicity-free, with decomposition into irreducible $G$-modules
\begin{equation}\label{eq: multiplicity-free induced decomposition}
    \Ind_H^G V_\lambda
    \cong\bigoplus_{\mu\in\N^+(\lambda)}V_\mu.
\end{equation}
In particular, for any left or right transversal $\T$ of $H$ in $G$,
\begin{equation}\label{eq: multiplicity-free induction dimension}
    |\T|\,d_\lambda
    = \sum_{\mu\in\N^+(\lambda)}d_\mu.
\end{equation}
\end{proposition}

\begin{proof}
By Frobenius reciprocity (\cref{thm: Frobenius reciprocity}), the multiplicity of $V_\mu$ in $\Ind_H^G V_\lambda$ equals the multiplicity of $V_\lambda$ in $\Res_H^G V_\mu$.
By the definition of $\N^+(\lambda)$ and the multiplicity-free assumption, this multiplicity is one when $\mu\in\N^+(\lambda)$ and zero otherwise.
Maschke's theorem (\cref{thm: Maschke}) therefore gives the decomposition in \cref{eq: multiplicity-free induced decomposition}.
Taking dimensions gives \cref{eq: multiplicity-free induction dimension}, since $\dim(\Ind_H^G V_\lambda)=[G:H]d_\lambda$ and every left or right transversal has cardinality $[G:H]$.
\end{proof}

\section{QFT using subgroups}
\label{sec: first induction}
In this section, we describe an inductive approach that can be used for implementing the quantum Fourier transform for any finite group that has a chain of subgroups with multiplicity-free restrictions.
This approach is not new; it has been used to implement both classical fast Fourier transforms, see Refs.~\cite{classicalFFT,ClausenBaum1993,Maslen1998},
and quantum Fourier transforms; see Refs.~\cite{Beals, GenQFT, SnQFT_Speedup}.

\subsection{Bratteli diagram and Gelfand--Tsetlin basis}
\label{subsec: bratteli}

Let
\begin{equation*}
    \{e\} = G_1 < G_2 < \cdots < G_n = G
\end{equation*}
be a chain of finite groups.  Its \emph{Bratteli diagram} is the graded directed multigraph with $n$ levels whose vertices at level $m$ are the irreducible representation labels $\lambda\in\widehat{G}_m$.  The number of edges $\lambda\to\mu$ from $\lambda\in\widehat{G}_{m-1}$ to $\mu\in\widehat{G}_m$ is the branching multiplicity
\begin{equation*}
    \dim\Hom_{G_{m-1}}
    \!\left(V_\lambda,\Res_{G_{m-1}}^{G_m}V_\mu\right).
\end{equation*}
Thus, for a multiplicity-free chain, there is an edge $\lambda\to\mu$ precisely when $V_\lambda$ occurs in $\Res_{G_{m-1}}^{G_m}V_\mu$.

\sloppy For every edge, choose an isometric embedding of the corresponding copy of $V_\lambda$ into $\Res_{G_{m-1}}^{G_m}V_\mu$.  Iterating these embeddings from the trivial representation at level $1$ gives, for each $\mu\in\widehat{G_m}$, an orthonormal \emph{Gelfand--Tsetlin basis} of $V_\mu$.  Its vectors are indexed by paths
\begin{equation*}
    P=(\lambda^{(1)}\to\lambda^{(2)}\to\cdots
    \to\lambda^{(m)}=\mu)
\end{equation*}
in the Bratteli diagram, and we write $\Path(\mu)$ for the set of such paths and $\ket{P}$ for the corresponding basis vector.
We also use the notation $Q=Q'\to\lambda$ to denote extending the truncated path $Q'$ to $\lambda$. This basis is subgroup adapted: for each fixed tail $\lambda^{(k)}\to\cdots\to\lambda^{(m)}$, the span of the basis vectors having that tail is a $G_k$-invariant copy of $V_{\lambda^{(k)}}$.
The basis depends on the chosen embeddings, whereas its path labels and branching structure are determined by the subgroup chain.

\subsection{Fourier transform for a group and a subgroup}\label{subsec:F for group and subgroup}
Let $G$ be a finite group and consider a proper subgroup $H < G$.
Additionally, assume that the restrictions of all irreps of $G$ to $H$ are multiplicity-free, and that the irreducible representations are unitary and subgroup adapted. Let $\T$ be a left transversal, so that each group element can be written uniquely as $g = th$ with $t\in \T, h\in H$.

We can define the following encoding map for group elements.
\begin{definition}[Encoding unitary]\label{def: encoding tensor}
    Let $\T$ be a left transversal for $H < G$.
    The \emph{encoding unitary} is the linear map $\Enc : \C^{G} \to \C^{\T} \otimes \C^{H}$ acting on basis vectors $\ket{g}$ with $g \in G$ as
    \begin{equation}
        \Enc \ket{g} = \ket{t} \otimes \ket{h},
        \quad t \in \T,\,\,h \in H,\, g = th.
    \end{equation}
\end{definition}

$\Enc$ is a ``classical'' unitary operation, as it only relabels basis vectors, and the corresponding matrix is a permutation matrix.

Using the encoding unitary $\Enc$, we can express the Fourier transform of $G$ in terms of that of $H$. After applying $\Enc$, we perform the Fourier transform on the subgroup $H$ to decompose $\C[H]$ into irreducible parts, and then we apply a unitary $\IndTrans$ to recover the full Fourier transform of $G$. The total transform is then given by

\begin{equation}
    \C[G] \xrightarrow{\Enc}
    \C[\T]\otimes \C[H] \xrightarrow{I_{\T} \otimes \Fou_H}
    \C[\T]\otimes \left(\bigoplus_{\lambda \in \widehat{H}} V_{\lambda}\otimes V_{\lambda}^\vee\right) \xrightarrow{\IndTrans}
    \bigoplus_{\mu \in \hatG} V_{\mu} \otimes V_{\mu}^\vee,
\end{equation}
where $\IndTrans$ is the \emph{induction transform} from $H$ to $G$.

\begin{definition}[Induction transform]\label{def: IndTrans}
    Let $\T$ be a left transversal for $H < G$, with assumptions as in the beginning of this section.
    The \emph{induction transform} from $H$ to $G$ is the controlled operation (i.e., a direct sum) given by
    \begin{equation}\label{eq: UInd}
        \IndTrans := \sum_{\lambda \in \widehat{H}} \proj{\lambda} \x \IndTrans[\lambda].
    \end{equation}
    For all $\lambda \in \widehat{H}$, the linear map $\IndTrans[\lambda] : \C[\T] \otimes V_\lambda \to \bigoplus_{\mu \in \N^+(\lambda)} V_\mu$ is defined by its matrix entries as
    \begin{equation}\label{eq: induction transform matrix entries}
        [\IndTrans[\lambda]]^{\mu, Q'}_{t, P} = \sqrt{\frac{d_\mu}{d_\lambda|\T|}}\left[\Rep_{\mu}(t)\right]^{Q}_{P \to \mu},
    \end{equation}
    where $t \in \T$ and $P \in \Path(\lambda)$ label a basis of the input space and $\mu \in  \N^+(\lambda)$ and $Q \in \Path(\mu)$ label a basis of the output space.
    Here $Q'$ denotes the truncation of $Q = Q' \to \mu$.

\end{definition}

The following lemma shows that $\IndTrans[\lambda]$ realizes the decomposition in \cref{eq: multiplicity-free induced decomposition} as a unitary $G$-intertwiner.

\begin{lemma}\label{lemma: IndTrans basis transform}
For every $\lambda\in\widehat H$, the induction transform is a unitary $G$-intertwiner
\begin{equation}
    \IndTrans[\lambda]:\Ind_H^G V_\lambda
    \xrightarrow{\ \sim\ }
    \bigoplus_{\mu\in\N^+(\lambda)}V_\mu.
\end{equation}
In particular, for every $g\in G$,
\begin{equation}\label{eq: IndTrans intertwines induced rep}
    \IndTrans[\lambda] \cdot \Ind_H^G\Rep_\lambda(g)
    =\left(\bigoplus_{\mu\in\N^+(\lambda)}\Rep_\mu(g)\right) \cdot \IndTrans[\lambda].
\end{equation}
\end{lemma}
\begin{proof}
    To prove \cref{eq: IndTrans intertwines induced rep}, fix $g\in G$ and $t\in\T$, and write the unique transversal factorization $gt=t_g h_g$, with $t_g\in\T$ and $h_g\in H$.
    On both sides of \cref{eq: IndTrans intertwines induced rep}, we will compare the matrix entries with row index $(\mu,Q')$ and column index $(t,P)$ where $\mu\in\N^+(\lambda)$, $Q=Q'\to\mu\in\Path(\mu)$, and $P\in\Path(\lambda)$.
    Expanding the product on the left-hand side over the intermediate basis indices $(s,R)\in\T\times\Path(\lambda)$ gives
    \begin{align}
        \left[\IndTrans[\lambda] \cdot \Ind_H^G\Rep_\lambda(g)\right]^{\mu,Q'}_{t,P}
        &=
        \sum_{s\in\T}\sum_{R\in\Path(\lambda)}
        [\IndTrans[\lambda]]^{\mu,Q'}_{s,R}
        [\Ind_H^G\Rep_\lambda(g)]^{s,R}_{t,P} \nonumber \\
        &=
        \sqrt{\frac{d_\mu}{d_\lambda|\T|}}
        \sum_{R\in\Path(\lambda)}
        [\Rep_\mu(t_g)]^Q_{R\to\mu}
        [\Rep_\mu(h_g)]^{R\to\mu}_{P\to\mu} \nonumber \\
        &=
        \sqrt{\frac{d_\mu}{d_\lambda|\T|}}
        [\Rep_\mu(gt)]^Q_{P\to\mu}.
        \label{eq: LHS}
    \end{align}
    The second equality uses \cref{eq: induced representation matrix entries,eq: induction transform matrix entries}, together with $[\Rep_\lambda(h_g)]^R_P=[\Rep_\mu(h_g)]^{R\to\mu}_{P\to\mu}$, which follows from subgroup adaptation since $h_g\in H$.
    For the final equality, all matrix entries of $\Rep_\mu(h_g)$ from this copy of $V_\lambda$ to its orthogonal complement vanish, so the sum extends to all of $\Path(\mu)$ and the homomorphism property of $\Rep_\mu$ applies.
    On the right-hand side of \cref{eq: IndTrans intertwines induced rep}, the direct sum preserves the irrep label $\mu$ and acts as $\Rep_\mu(g)$ on the $\mu$-summand. Writing $S=S'\to\mu\in\Path(\mu)$ for each intermediate path, we obtain
    \begin{align}
        \left[\left(\bigoplus_{\mu'\in\N^+(\lambda)}\Rep_{\mu'}(g)\right)\cdot\IndTrans[\lambda]\right]^{\mu,Q'}_{t,P}
        &=\sum_{S\in\Path(\mu)}
        [\Rep_\mu(g)]^Q_S
        [\IndTrans[\lambda]]^{\mu,S'}_{t,P}\nonumber\\
        &=
        \sqrt{\frac{d_\mu}{d_\lambda|\T|}}
        \sum_{S\in\Path(\mu)}
        [\Rep_\mu(g)]^Q_S[\Rep_\mu(t)]^S_{P\to\mu}\nonumber\\
        &=
        \sqrt{\frac{d_\mu}{d_\lambda|\T|}}
        [\Rep_\mu(gt)]^Q_{P\to\mu}.
    \end{align}
    Here the second equality uses \cref{eq: induction transform matrix entries}, and the last uses the homomorphism property of $\Rep_\mu$.
    Since this agrees with \cref{eq: LHS}, it proves the intertwining relation \eqref{eq: IndTrans intertwines induced rep}.

    To prove unitarity of $\IndTrans[\lambda]$, let
    \begin{equation*}
        U:=\IndTrans[\lambda],\qquad
        A(g):=\Ind_H^G\Rep_\lambda(g),\qquad
        B(g):=\bigoplus_{\mu\in\N^+(\lambda)}\Rep_\mu(g).
    \end{equation*}
    The intertwining relation \eqref{eq: IndTrans intertwines induced rep} gives $UA(g)=B(g)U$. Since $A(g)$ is unitary,
    \begin{equation*}
        B(g)UU^\dagger B(g)^\dagger
        =(B(g)U)(B(g)U)^\dagger
        =(UA(g))(UA(g))^\dagger
        =UU^\dagger.
    \end{equation*}
    As $B(g)$ is also unitary, $UU^\dagger$ commutes with $B(g)$ for every $g\in G$.
    Since $B(g)$ is a direct sum of pairwise inequivalent irreps, Schur's lemma gives
    \begin{equation*}
        UU^\dagger=\bigoplus_{\mu\in\N^+(\lambda)}c_\mu I_{V_\mu}.
    \end{equation*}
    For each $\mu\in\N^+(\lambda)$, let
    $\pi_\mu:\bigoplus_{\mu'\in\N^+(\lambda)}V_{\mu'}\to V_\mu$
    be the canonical projection onto the $\mu$-summand, and define
    $U_\mu:=\pi_\mu\circ U:\C[\T]\otimes V_\lambda\to V_\mu$.
    By \cref{eq: induction transform matrix entries},
    \begin{equation*}
        U_\mu(\ket t\otimes\ket P)
        =\sqrt{\frac{d_\mu}{d_\lambda|\T|}}\,
        \Rep_\mu(t)\ket{P\to\mu}.
    \end{equation*}
    Since $\Rep_\mu(t)$ is unitary, each of the $|\T|d_\lambda$ columns of $U_\mu$ has squared norm $d_\mu/(d_\lambda|\T|)$. Hence
    \begin{equation*}
        c_\mu d_\mu
        =\tr(U_\mu U_\mu^\dagger)
        =\sum_{\substack{t\in\T\\P\in\Path(\lambda)}}
          \bigl\|U_\mu(\ket t\otimes\ket P)\bigr\|^2
        =|\T|d_\lambda\,\frac{d_\mu}{d_\lambda|\T|}
        =d_\mu.
    \end{equation*}
    Thus $c_\mu=1$ for every $\mu$, and $UU^\dagger=I$ implying that $U^\dagger$ is an isometry.
    By \cref{eq: multiplicity-free induction dimension}, $U$ has equal input and output dimensions, so $U$ is unitary.
\end{proof}

\begin{lemma}\label{lemma: induction step of first induction}
For every finite group $G$ and subgroup $H < G$ satisfying the assumptions above, the Fourier transform of $G$ factors as
\begin{equation}\label{eq: QFT subgroup induction step}
    \Fou_G
    =\IndTrans\bigl(\mathrm{I}_{\C[\T]}\otimes\Fou_H\bigr)\Enc.
\end{equation}
\end{lemma}
Here the full induction transform also carries along the second Fourier path, which labels the dual register.  More precisely, for $P,R'\in\Path(\lambda)$, it is defined by
\begin{equation}\label{eq: def full IndTrans}
    \IndTrans\ket{t}\ket{\lambda,P,R'^{\,\vee}}
    =\sum_{\mu\in\N^+(\lambda)}\sum_{Q\in\Path(\mu)}
      [\IndTrans[\lambda]]^{\mu,Q'}_{t,P}
      \ket{\mu,Q,(R'\to\mu)^\vee}.
\end{equation}
\begin{proof}
    It is enough to compare the two sides on a group-basis vector.  Write $g=th$ with $t\in\T$ and $h\in H$.  After applying the encoding and the Fourier transform of $H$, we obtain
    \begin{equation*}
        \bigl(\mathrm{I}_{\C[\T]}\otimes\Fou_H\bigr)\Enc\ket{g}
        =\sum_{\lambda\in\widehat H}
         \sqrt{\frac{d_\lambda}{|H|}}
         \sum_{P,R'\in\Path(\lambda)}
         [\Rep_\lambda(h)]^P_{R'}
         \ket{t}\ket{\lambda,P,{R'^{\,\vee}}}.
    \end{equation*}
    Fix an output basis vector $\ket{\mu,Q,R^{\vee}}$, and write its second path uniquely as $R=R'\to\mu$, where $R'\in\Path(\lambda)$ and $\lambda\to\mu$.  By \cref{def: IndTrans,eq: def full IndTrans}, its coefficient after applying $\IndTrans$ is
    \begin{align*}
        &\sqrt{\frac{d_\lambda}{|H|}}
        \sqrt{\frac{d_\mu}{d_\lambda|\T|}}
        \sum_{P\in\Path(\lambda)}
        [\Rep_\mu(t)]^Q_{P\to\mu}
        [\Rep_\lambda(h)]^P_{R'} \\
        &\qquad=
        \sqrt{\frac{d_\mu}{|G|}}
        \sum_{P\in\Path(\lambda)}
        [\Rep_\mu(t)]^Q_{P\to\mu}
        [\Rep_\mu(h)]^{P\to\mu}_{R'\to\mu} =
        \sqrt{\frac{d_\mu}{|G|}}
        [\Rep_\mu(th)]^Q_R.
    \end{align*}
    Here we used $|G|=|\T||H|$ and the subgroup adapted form of $\Rep_\mu(h)$.  The final expression is exactly the coefficient of $\ket{\mu,Q,R^{\vee}}$ in $\Fou_G\ket{g}$, proving \cref{eq: QFT subgroup induction step}.
\end{proof}

As we decomposed the Fourier transform into a product of unitary matrices, we can define a corresponding quantum circuit, as shown in \cref{fig: induction step of first induction}.

\begin{figure}[H]
\centering
\begin{quantikz}[row sep=0.35cm, column sep=0.45cm]
	\setwiretype{n}    &                       &                                   &               &                 &&                  &                                         & \gate[3][1.5cm]{\Fou_H}\gateoutput{$\widetilde R^{\vee}$}   &                  \setwiretype{q}                     & \rstick[2]{$\ket{{R'^{\,\vee}}}$} \\
	\setwiretype{n}    & \gate[3][1.5cm]{\Fou_G} & \rstick{$\ket{R'^{\,\vee}}$}  \qw &               & \setwiretype{n} &&                  &  \gate[3][1.5cm]{\Enc} \gateoutput{$h$} & \gateoutput{$\lambda$}                      \setwiretype{q}   & \ctrl{1}                                             &   \\
	\lstick{$\ket{g}$} &                       & \rstick{$\ket{\mu}$}              &\setwiretype{n}& \midstick{$=$}  &&\lstick{$\ket{g}$}&    \qw                                  & \gateoutput{$
		P$}                                  & \gate[2][1.5cm]{\mathrm{U}_{\mathrm{Ind},\lambda}}  \setwiretype{q} & \rstick{$\ket{\mu}$}                      &\setwiretype{n}&                                 \\
	\setwiretype{n}    &                       & \rstick{$\ket{Q'}$}  \qw &               & \setwiretype{n} &&                  &   \gateoutput{$t$}                      &                                             \setwiretype{q}   &                                                      & \rstick{$\ket{Q'}$}              &\setwiretype{n}&
\end{quantikz}
\caption{\cref{lemma: induction step of first induction} as a quantum circuit. The Fourier transform $\Fou_G$ is decomposed into a subgroup Fourier transform $\Fou_H$ and the induction transform $\IndTrans$.}
\label{fig: induction step of first induction}
\end{figure}

\subsection{Induction along a multiplicity-free chain}\label{subsec: First induction}

The induction along a subgroup can be repeated along a chain of subgroups.
Suppose that $G$ has a multiplicity-free subgroup chain $\{e\} = G_1 < \cdots < G_n = G$.  For $2\leq m\leq n$, let $\T_m$ be a left transversal for $G_{m-1} < G_m$, let $\Enc_m$ be the corresponding encoding unitary, and let $\IndTrans[m]$ be the full induction transform of \cref{eq: def full IndTrans} for this inclusion.  Thus a group element has the recursive encoding
\begin{equation*}
    g=t^{(n)}t^{(n-1)}\cdots t^{(2)},
    \qquad t^{(m)}\in\T_m.
\end{equation*}

\begin{theorem}\label{thm: first induction}
Define $\mathcal Q_1=\Fou_{G_1}=\mathrm I$ and, recursively, define
\begin{equation}\label{eq: recursive subgroup-chain QFT}
    \mathcal Q_m
    =\IndTrans[m]
     \bigl(\mathrm I_{\C[\T_m]}\otimes\mathcal Q_{m-1}\bigr)\Enc_m,
    \qquad 2\leq m\leq n.
\end{equation}
Then $\mathcal Q_n=\Fou_G$.  Its output basis is the Gelfand--Tsetlin basis indexed by pairs of paths in the Bratteli diagram.
The corresponding circuit is shown in \cref{fig: theorem first induction full circuit}.
\end{theorem}

\begin{proof}
    We argue by induction on $m$.  Since $G_1=\{e\}$, its group algebra is one-dimensional and $\mathcal Q_1=\Fou_{G_1}=\mathrm I$.  Suppose that $\mathcal Q_{m-1}=\Fou_{G_{m-1}}$.  Applying \cref{lemma: induction step of first induction} to the inclusion $G_{m-1} < G_m$ gives
    \begin{equation*}
        \Fou_{G_m}
        =\IndTrans[m]
         \bigl(\mathrm I_{\C[\T_m]}\otimes\Fou_{G_{m-1}}\bigr)\Enc_m
        =\mathcal Q_m.
    \end{equation*}
    Hence the claim holds for every $m$, and in particular $\mathcal Q_n=\Fou_{G_n}=\Fou_G$.
\end{proof}

\begin{figure}[H]
\centering
\resizebox{\textwidth}{!}{\begin{quantikz}[row sep={-0.7cm}]
		&\gate[style=transparent][0cm][1cm]{\,}\setwiretype{n}    &                                &\gate[4][2cm]{\mathrm{U}_{\mathrm{Ind},\lambda}^{(1)}}\gateoutput{$\mu$} &\ctrl{3}  \setwiretype{q}                                                                          &                                                                                           &                                                                                           &                                                           &                                                                                                               &\rstick{$\ket{P_1}$}                                       &\gate[style=transparent][0.4cm]{\,}\setwiretype{n}&\rstick[16]{$\ket{P^{\vee}}$}\\
		&\gate[style=transparent][0cm][1cm]{\,}\setwiretype{n}    &                                &                                                        &                                                                                                   &                                                                                           &                                                                                           &                                                           &                                                                                                               &                                                           &\setwiretype{n}&\\
		&\gate[style=transparent][0cm][1cm]{\,}\setwiretype{n}    &                                &                                                        &                                                                                                   &                                                                                           &                                                                                           &                                                           &                                                                                                               &                                                           &\setwiretype{n}&\\
		\lstick[24]{$\ket{g}$}  &\gate[style=transparent][0cm][1cm]{\,}\setwiretype{n}    &\lstick{$\ket{t_1 = e}$}        &\gateinput{$t$}\setwiretype{q}                          &\gate[4][2cm]{\mathrm{U}_{\mathrm{Ind},\lambda}^{(2)}}\gateoutput{$\mu$}\setwiretype{n}    &\ctrl{3}\setwiretype{q}                                                                    &                                                                                           &                                                           &                                                                                                               &\rstick{$\ket{P_2}$}                                       &\setwiretype{n}&\\
		&\gate[style=transparent][0cm][1cm]{\,}\setwiretype{n}    &                                &                                                        &                                                                                                   &                                                                                           &                                                                                           &                                                           &                                                                                                               &                                                           &\setwiretype{n}&\\
		&\gate[style=transparent][0cm][1cm]{\,}\setwiretype{n}    &                                &                                                        &                                                                                                   &                                                                                           &                                                                                           &                                                           &                                                                                                               &                                                           &\setwiretype{n}&\\
		&\gate[style=transparent][0cm][1cm]{\,}\setwiretype{n}    &\lstick{$\ket{t^{(2)}}$}        &\setwiretype{q}                                         &                 \gateinput{$t$}   \gateoutput{$Q'$}                                              &\gate[5][2cm]{\mathrm{U}_{\mathrm{Ind},\lambda}^{(3)}}\gateinput{$P$}\gateoutput{$\mu$} \setwiretype{q} &\ctrl{3}                                                                                   &                                                           &                                                                                                               &\rstick{$\ket{P_3}$}                                       &\setwiretype{n}&\\
		&\gate[style=transparent][0cm][1cm]{\,}\setwiretype{n}    &                                &                                                        &                                                                                                   &                                                                                           &\setwiretype{n}                                                                            &                                                           &                                                                                                               &                                                           &\setwiretype{n}&\\
		&\gate[style=transparent][0cm][1cm]{\,}\setwiretype{n}    &                                &                                                        &                                                                                                   &                                                                                           &                                                                                           &                                                           &                                                                                                               &                                                           &\setwiretype{n}&\\
		&\gate[style=transparent][0cm][1cm]{\,}\setwiretype{n}    &                                &                                                        &                                                                                                   &                                                               \gateoutput[2]{$Q'$}          &\gate[6][2cm]{\mathrm{U}_{\mathrm{Ind},\lambda}^{(4)}}\gateinput[2]{$P$}\gateoutput{$\mu$}  \setwiretype{q}  &                                                           &                                                                                                               &\rstick{$\ket{P_4}$}                                       &\setwiretype{n}&\\
		&\gate[style=transparent][0cm][1cm]{\,}\setwiretype{n}    &\lstick{$\ket{t^{(3)}}$}        &\setwiretype{q}                                         &                                                                                                   &                                       \gateinput{$t$}                                  &                                                        \setwiretype{q}  &\setwiretype{n}                                            &                                                                                                               &                                                           &\setwiretype{n}&\\
		&\gate[style=transparent][0cm][1cm]{\,}\setwiretype{n}    &                                &                                                        &                                                                                                   &                                                                                           &                                                                                           &\setwiretype{n}                                            &                                                                                                               &                                                           &\setwiretype{n}&\\
		&\gate[style=transparent][0cm][1cm]{\,}\setwiretype{n}    &                                &                                                        &                                                                                                   &                                                                                           &                                                      \gateoutput[3]{$Q'$}                   &\gate[12,style=transparent][1cm]{\ddots}\setwiretype{q}    &\setwiretype{n}                                                                                                &\rstick{$\,\,\,\vdots$}                                    &\setwiretype{n}&\\
		&\gate[style=transparent][0cm][1cm]{\,}\setwiretype{n}    &                                &                                                        &                                                                                                   &                                                                                           &                                                                                                  &\setwiretype{q}                                            &\setwiretype{n}                                                                                                &                                                           &\setwiretype{n}&\\
		&\gate[style=transparent][0cm][1cm]{\,}\setwiretype{n}    &\lstick{$\ket{t^{(4)}}$}        &\setwiretype{q}                                         &                                                                                                   &                                                                                           &                                     \gateinput{$t$}                                    &\setwiretype{q}                                            &\setwiretype{n}                                                                                                &                                                           &\setwiretype{n}&\\
		&\gate[style=transparent][0cm][1cm]{\,}\setwiretype{n}    &                                &                                                        &                                                                                                   &                                                                                           &                                                                                           &                                                           &\ctrl{3}\setwiretype{q}                                                                                        &\rstick{$\ket{P_{n-1}}$}                                   &\setwiretype{n}&\\
		&\gate[style=transparent][0cm][1cm]{\,}\setwiretype{n}    &                                &                                                        &                                                                                                   &                                                                                           &                                                                                           &                                                           &\setwiretype{n}                                                                                                &                                                           &\setwiretype{n}&\\
		&\gate[style=transparent][0cm][1cm]{\,}\setwiretype{n}    &                                &                                                        &                                                                                                   &                                                                                           &                                                                                           &                                                           &\setwiretype{n}                                                                                                &                                                           &\setwiretype{n}&\\
		&\gate[style=transparent][0cm][1cm]{\,}\setwiretype{n}    &                                &                                                        &                                                                                                   &                                                                                           &                                                                                           &                                                           &\gate[9][2cm]{\mathrm{U}_{\mathrm{Ind},\lambda}^{(n)}}\gateinput[6]{$P$}\gateoutput{$\mu$}              \setwiretype{q} &\rstick{$\ket{\mu} = \ket{Q_n}$}                           &\setwiretype{n}&\\
		&\gate[style=transparent][0cm][1cm]{\,}\setwiretype{n}    &                                &                                                        &                                                                                                   &                                                                                           &                                                                                           &                                                           &                                                                                                               &                                                           &\setwiretype{n}&\\
		&\gate[style=transparent][0cm][1cm]{\,}\setwiretype{n}    &                                &                                                        &                                                                                                   &                                                                                           &                                                                                           &                                                           &                                                                                                                        &                                                           &\setwiretype{n}&\\
		&\gate[style=transparent][0cm][1cm]{\,}\setwiretype{n}    &                                &                                                        &                                                                                                   &                                                                                           &                                                                                           &                                                           &                                        \gateoutput[6]{$Q'$}          \setwiretype{q} &                                                           &\setwiretype{n}&\rstick[6]{$\ket{Q'}$}\\
		&\gate[style=transparent][0cm][1cm]{\,}\setwiretype{n}    &                                &                                                        &                                                                                                   &                                                                                           &                                                                                           &                                                           &                                                                                      \setwiretype{q} &\rstick[3, brackets=none]{$\,\,\vdots$}    \setwiretype{n} &\setwiretype{n}&\\
		&\gate[style=transparent][0cm][1cm]{\,}\setwiretype{n}    &                                &                                                        &                                                                                                   &                                                                                           &                                                                                           &                                                           &                                                                      \setwiretype{q} &\setwiretype{n}                                            &\setwiretype{n}&\\
		&\gate[style=transparent][0cm][1cm]{\,}\setwiretype{n}    &                                &                                                        &                                                                                                   &                                                                                           &                                                                                           &                                                           &                                                                                                                   &  \setwiretype{q}                                          &\setwiretype{n}&\\
		&\gate[style=transparent][0cm][1cm]{\,}\setwiretype{n}    &                                &                                                        &                                                                                                   &                                                                                           &                                                                                           &                                                           &                                                                                                                   & \setwiretype{q}                                           &\setwiretype{n}&\\
		&\gate[style=transparent][0cm][1cm]{\,}\setwiretype{n}    &\lstick{$\ket{t^{(n)}}$}        &\setwiretype{q}                                         &                                                                                                   &                                                                                           &                                                                                           &                                                           &                                       \gateinput{$t$}                                                                  &                                  \setwiretype{q}     &\setwiretype{n}&
	\end{quantikz}
}
\caption{Quantum circuit for the Fourier transform $\Fou_G$. Group elements $g \in G$ are encoded with transversal elements $t^{(m)} \in \T_m$ as $g = t^{(n)}\dots t^{(3)}t^{(2)}$. For Fourier basis vectors, we let $P = P_1 \to P_2 \to \dots \to P_n$ and $Q = Q_1 \to Q_2 \to \dots \to Q_n$ encode paths in the Bratteli diagram, where $P_m, Q_m \in \hatG_m$. The $P$-path labels the dual factor and the $Q$-path labels the ordinary representation factor.}

\label{fig: theorem first induction full circuit}
\end{figure}

\section{Beals' algorithm for induction transform}\label{sec: beals}

One way of implementing the induction transform of \cref{def: IndTrans} is via the algorithm proposed by Beals~\cite{Beals}. First implemented in the QFT for the symmetric group, this approach has been generalized and extended to other non-abelian groups in~\cite{GenQFT}. In this section, we develop a representation-theoretic analysis and explanation of this general method.

As in \cref{subsec:F for group and subgroup}, let $H < G$ be finite groups such that the restriction of every $G$-irrep to $H$ is multiplicity-free, and use unitary, subgroup-adapted representations.
Write $n = [G:H]$, and fix a left transversal $\T = \{\tau_1,\ldots,\tau_n\}$ for $H$ in $G$ containing $e$.
Consider the Bratteli diagram for some subgroup chain from $\{e\}$ to $H$, extended by the inclusion $H<G$, with bases chosen as in \cref{subsec: bratteli}.
In particular, $P\in\Path(\lambda)$ is a complete path ending at the $H$-irrep $\lambda$, and $P\to\mu$ extends it by the unique edge from $\lambda$ to a successor $\mu\in\N^+(\lambda)$.

Following \cref{def: IndTrans,eq: UInd}, the first register of $\IndTrans$ is the \emph{control register}, which stores the irrep label $\ket{\lambda}$.
The second is the \emph{transversal register}, whose basis states $\ket{t}$ label transversal elements $t \in \T$ and include a special symbol $\ket{\star}$ indicating that $\IndTrans$ has been applied.
The third is the \emph{path extension register}, which stores the new end point $\ket{\mu}$ with $\mu \in \N^+(\lambda)$, since the output paths are one edge longer than the input paths.
It is initialized as $\ket{0}$, a distinguished state orthogonal to all irrep labels $\ket{\mu}$.
The fourth is the \emph{path register} with state space $\bigoplus_{\lambda\in\widehat H}V_{\lambda}$ storing paths to all $H$-irreps.
Initially, when the control register is in state $\ket{\lambda}$, the path register lies in $V_\lambda$, with basis vectors $\ket{P} \in V_\lambda$ indexed by $P\in\Path(\lambda)$, but its endpoint may change during the implementation of $\IndTrans[\lambda]$ while the control label $\lambda$ is retained.
For fixed $\lambda$, the transversal and path registers initially encode the induced module $\Ind^G_H V_\lambda \cong \C[\T]\otimes V_\lambda$, on which the controlled transform $\IndTrans[\lambda]$ acts.

In this encoding, the induction transform $\IndTrans$ from \cref{def: IndTrans} acts on the input basis vectors as
\begin{equation}\label{eq: Beals transform indtrans action}
    \IndTrans \ket{\lambda}\ket{\tau_k}\ket{0}\ket{P} = \ket{\lambda}\ket{\star}\ket{\Psi_{k, P}},
\end{equation}
where the transversal register $\ket{\tau_k}$ is reset to $\ket{\star}$ and the dependence on $\tau_k$ is transferred to
\begin{equation}
    \ket{\Psi_{k, P}} \defeq
    \sum_{\mu \in \N^+(\lambda)}
    \sqrt{\frac{d_\mu}{n d_\lambda}}\ket{\mu}\otimes
    \sum_{Q \in \Path(\mu)}
    [\Rep_\mu(\tau_k)]^{Q}_{P \to \mu}\ket{Q'}.
\end{equation}
Here, $\ket{\mu}\ket{Q'}$ encodes the path $Q = Q' \to \mu$, where the truncated path $Q'$ may end at an $H$-irrep other than $\lambda$.
Note that for each $\mu$ the second register has unit norm: its coefficients form the $(P\to\mu)$-column of the unitary matrix $\Rep_\mu(\tau_k)$.
Since the labels $\ket{\mu}$ are orthonormal, $\ket{\Psi_{k,P}}$ is a unit vector by \cref{eq: multiplicity-free induction dimension}.

Beals' algorithm uses three different quantum gates, $U$, $V_k$ and $R(\tau_k)$, which we now specify.
First, for each $\lambda\in\widehat H$, define the following weighted superposition of its successors:
\begin{equation}\label{eq:path extension superposition state}
    \ket{\phi_\lambda}
    \defeq
    \sum_{\mu\in\N^+(\lambda)}
    \sqrt{\frac{d_\mu}{n d_\lambda}}
    \ket{\mu}
\end{equation}
which has unit norm by \cref{eq: multiplicity-free induction dimension}.
For $P\in\Path(\lambda)$, define $U$ to act on basis vectors as
\begin{align}\label{eq: U gate Beals}
    U\ket{\lambda}\ket{\star}\ket{0}\ket{P}
    &:= \ket{\lambda}\ket{\star}\ket{\phi_\lambda}\ket{P},
\end{align}
and let $U$ act as the identity when the transversal register is not $\ket{\star}$.
We complete the definition by choosing, for each $\lambda$, a unitary on the path extension register that maps $\ket{0}$ to $\ket{\phi_\lambda}$ and applying it throughout the $\star$-labelled subspace, with the path register left unchanged.

Next, define a gate $V_k$ that interchanges $\star$ and $\tau_k$ within the transversal register when the path extension register is in its initial state $\ket{0}$ and the path ends at the irrep $\lambda$ stored in the control register:
\begin{equation}\label{eq: V gate Beals}
    \begin{aligned}
        V_k \ket{\lambda}\ket{\star}\ket{0}\ket{P' \to \lambda} &\defeq \ket{\lambda}\ket{\tau_k}\ket{0}\ket{P' \to \lambda}, \\
        V_k \ket{\lambda}\ket{\tau_k}\ket{0}\ket{P' \to \lambda} &\defeq \ket{\lambda}\ket{\star}\ket{0}\ket{P' \to \lambda}.
    \end{aligned}
\end{equation}
The gate $V_k$ acts as identity on the orthogonal complement of the span of these pairs.
In particular, it fixes states whose path does not end at $\lambda$, even when the path extension register is $\ket{0}$.

Finally, $R(\tau_k)$ is the controlled version of the irreducible representation matrix $\Rep_\mu(\tau_k)$ evaluated at the transversal element $\tau_k \in \T$.
Controlled by the label $\mu$ in the third register and the label $\star$ in the second, it applies $\Rep_\mu(\tau_k)$ only to the fourth register, the path register, leaving the other three registers unchanged.
For $P=P'\to\mu \in \Path(\mu)$ and $\mu \in \hatG$, its action on basis vectors is
\begin{equation}\label{eq: R gate Beals}
    R(\tau_k) \ket{\lambda}\ket{\star}\ket{\mu}\ket{P'}
    = \ket{\lambda}\ket{\star}\ket{\mu}
    \Biggl(\sum_{M \in \Path(\mu)} [\Rep_{\mu }(\tau_k)]^{M}_{P}\ket{M'}\Biggr).
\end{equation}
It leaves the control register $\ket{\lambda}$ unchanged and acts as identity on the orthogonal complement of the subspaces specified in \cref{eq: R gate Beals}, in particular when the transversal register is not $\ket{\star}$.

The following lemma plays a key role in proving the correctness of Beals' algorithm.

\begin{lemma}\label{lemma:key_lemma_for_beals}
For every $g \in G$ and $P, Q \in \Path(\lambda)$ with $\lambda \in \widehat{H}$, it holds that
\begin{equation}
        \sum_{\mu \in \N^+(\lambda)} \frac{d_\mu}{n d_\lambda} [\Rep_\mu(g)]^{Q \to \mu}_{P\to \mu} = [\Ind^G_H \Rep_\lambda(g)]^{e,Q}_{e,P}.
\end{equation}
In particular, if $g \notin eH$, then
\begin{equation}
        \sum_{\mu \in \N^+(\lambda)} \frac{d_\mu}{n d_\lambda} [\Rep_\mu(g)]^{Q \to \mu}_{P \to \mu} = 0.
\end{equation}
\end{lemma}

\begin{proof}
    By \cref{lemma: IndTrans basis transform}, $\IndTrans[\lambda]$ is a unitary intertwiner from $\Ind_H^G\Rep_\lambda$ to $\bigoplus_{\mu\in\N^+(\lambda)}\Rep_\mu$.
    Since $e\in\T$ and $\Rep_\mu(e)=I$, \cref{def: IndTrans} gives, for every $P\in\Path(\lambda)$,
    \begin{equation*}
        \IndTrans[\lambda]\ket{e,P}
        =
        \sum_{\mu \in \N^+(\lambda)}
        \sqrt{\frac{d_\mu}{n d_\lambda}}\ket{P \to \mu}.
    \end{equation*}
    Inserting the intertwining relation between the identity-coset vectors therefore yields
    \begin{align*}
        [\Ind_H^G \Rep_\lambda(g)]^{e,Q}_{e,P}
        &= \bra{e,Q}\IndTrans[\lambda]^\dagger
        \left(\bigoplus_{\mu\in\N^+(\lambda)}\Rep_\mu(g)\right)
        \IndTrans[\lambda]\ket{e,P}\\
        &= \sum_{\mu \in \N^+(\lambda)}\frac{d_\mu}{n d_\lambda} [\Rep_{\mu}(g)]^{Q\to \mu}_{P \to \mu}.
    \end{align*}
    To evaluate this coefficient directly, write $g=th$ with $t\in\T$ and $h\in H$.
    By \cref{def: induced representation},
    \begin{equation*}
        \Ind_H^G \Rep_\lambda(g)\ket{e,P}
        =\ket{t}\otimes\Rep_\lambda(h)\ket{P}.
    \end{equation*}
    Since $t=e$ exactly when $g\in H$, this gives
    \begin{equation}
        \bra{e,Q}\Ind_H^G \Rep_\lambda(g)\ket{e,P}
        =
        \begin{cases}
            [\Rep_\lambda(g)]^{Q}_{P}, & g\in H,\\
            0, & g\notin H.
        \end{cases}
    \end{equation}
    Together, these identities prove both claims.
\end{proof}

Beals' algorithm implements the induction transform $\IndTrans$ by a sequence of gates
\begin{equation}\label{eq:beals block}
    W_k \defeq R(\tau_k) U V_k U^\dagger R(\tau_k)^\dagger,
\end{equation}
where $k\in[n]$ and the gates $U$, $V_k$, and $R(\tau_k)$ are defined in \cref{eq: U gate Beals,eq: V gate Beals,eq: R gate Beals}, respectively.

\begin{theorem}[Beals' construction~\cite{Beals,GenQFT}]
\label{thm:beals-induction-transform}
    On the input subspace specified in \cref{eq: Beals transform indtrans action}, the induction transform $\IndTrans$ can be decomposed as
    \begin{equation}\label{eq:beals algorithm}
        \IndTrans = \prod_{k=1}^n W_k = W_n W_{n-1}\cdots W_1.
    \end{equation}
    Here and below, products of the $W_k$ are ordered so that the factors act in increasing order of $k$.
    The circuit implementing $\IndTrans$ is shown in \cref{fig:beals-induction-circuit}.
\end{theorem}

\begin{figure}[H]
\centering
	{\small
		\begin{quantikz}[column sep=0.42cm]
		\lstick{$\ket{\lambda}$}
		&
        &
        &\ctrl{1}
		&\ctrl{1}
		&\gate[style=transparent]{\dots}
		&\ctrl{1}
		&\rstick{$\ket{\lambda}$}\\
		\lstick{$\ket{\tau_\ell}$}
        &
        &
		&\gate[3]{W_1}
		&\gate[3]{W_2}
		&\gate[style=transparent]{\dots}
		&\gate[3]{W_n}
		&\rstick{$\ket{\star}$}\\
        \lstick{$\ket{0}$}
        &
        &
		&
		&
		&\gate[style=transparent]{\dots}
		&
		& \rstick{$\ket{\mu}$}\\
        \lstick{$\ket{P'\to\lambda}$}
		&
        &
        &
		&
		&\gate[style=transparent]{\dots}
		&
		& \rstick{$\ket{Q'}$}
		\end{quantikz}

		\vspace{0.8cm}

		\begin{quantikz}[row sep=0.38cm,column sep=0.31cm]
			\lstick{$\ket{\lambda}$}
            &
			&\gategroup[4,steps=5,style={inner ysep=0.25cm}]{$W_k$}
			&\ctrl{1}
			&\ctrl{1}
			&\ctrl{1}
			&
			&
            &\rstick{$\ket{\lambda}$}\\
			\lstick{$\T$}
			&
            &\ctrl{1}\gategroup[1,style={transparent,inner sep=0cm},label style={label position=above,xshift=-0.12cm,yshift=-0.4cm}]{$\scriptstyle\star$}
			&\ctrl{1}\gategroup[1,style={transparent,inner sep=0cm},label style={label position=above,xshift=-0.12cm,yshift=-0.4cm}]{$\scriptstyle\star$}
			&\gate{V_k}
			&\ctrl{1}\gategroup[1,style={transparent,inner sep=0cm},label style={label position=above,xshift=-0.12cm,yshift=-0.4cm}]{$\scriptstyle\star$}
			&\ctrl{1}\gategroup[1,style={transparent,inner sep=0cm},label style={label position=above,xshift=-0.12cm,yshift=-0.4cm}]{$\scriptstyle\star$}
			&
            &\rstick{$\T$}\\
            \lstick{$\ket{0}$ or $\ket{\mu}$}
			&
            &\ctrl{1}\gategroup[1,style={transparent,inner sep=0cm},label style={label position=above,xshift=-0.14cm,yshift=-0.4cm}]{$\scriptstyle\mu$}
			&\gate{U^\dagger}
			&\octrl{-1}\gategroup[1,style={transparent,inner sep=0cm},label style={label position=above,xshift=-0.14cm,yshift=-0.4cm}]{$\scriptstyle 0$}
			&\gate{U}
			&\ctrl{1}\gategroup[1,style={transparent,inner sep=0cm},label style={label position=above,xshift=-0.14cm,yshift=-0.4cm}]{$\scriptstyle\mu$}
			&
            &\rstick{$\ket{0}$ or $\ket{\mu}$}\\
            \lstick{$\mathcal P$}
			&
            &\gate{\Rep_\mu(\tau_k)^\dagger}
			&
			&\ctrl{-1}
			&
			&\gate{\Rep_\mu(\tau_k)}
			&
            &\rstick{$\mathcal P$}
		\end{quantikz}
	}
	\caption{Circuit of Beals' algorithm for the induction transform from \cref{thm:beals-induction-transform}.
        The top circuit applies $W_1,W_2,\ldots,W_n$ from left to right, with $W_k$ defined in \cref{eq:beals block}.
        Both circuits follow the register order of \cref{eq: UInd}, with the path extension register inserted: remembered label $\lambda$, transversal, path extension, and path, as in \cref{eq: U gate Beals,eq: V gate Beals,eq: R gate Beals}.
        The lower circuit expands $W_k$.
        The gates $U$ and $U^\dagger$ target only the third register, controlled by the remembered label $\lambda$ and the transversal label $\star$.
        The boxes $\Rep_\mu(\tau_k)^\dagger$ and $\Rep_\mu(\tau_k)$ on the fourth wire, controlled by $\mu$ on the third wire and $\star$ on the second, implement the global gates $R(\tau_k)^\dagger$ and $R(\tau_k)$, respectively; these act as identity when the extension register is $0$.
        Controls by irrep labels select the corresponding operation coherently, while the $\star$-labelled controls activate only on $\ket{\star}$.
        In the $V_k$ column, the box denotes the transposition $\star\leftrightarrow\tau_k$, the open control marked $0$ selects the extension state $\ket{0}$, and the filled controls on the $\lambda$ and $\mathcal P$ wires jointly require that the path endpoint equals the remembered label $\lambda$.
        Thus all three control registers are preserved, as in \cref{eq: V gate Beals}.
        On an input with transversal label $\tau_\ell$, only $W_\ell$ acts nontrivially.
        It creates a superposition of extended paths $Q=Q'\to\mu$ and resets the transversal register to $\ket{\star}$.
        The earlier blocks fix the input by their controls, and the later blocks fix the transformed state by \cref{lemma:key_lemma_for_beals}.
}
	\label{fig:beals-induction-circuit}
\end{figure}

\begin{proof}
    Every gate preserves the control register, so fix $\lambda\in\widehat H$ and suppress $\ket{\lambda}$ for the rest of the proof.
    It suffices to verify the claimed action on every basis vector $\ket{\tau_\ell,0,P}$ with $\ell\in[n]$ and $P\in\Path(\lambda)$.
    Recall from \cref{eq: Beals transform indtrans action} that the induction transform maps this input to $\ket{\star}\ket{\Psi_{\ell, P}}$.
    We prove the theorem by showing that the decomposition on the RHS of \cref{eq:beals algorithm} also maps the input vectors to this state.

    We show that, on this input, only the block $W_\ell$ changes the state.
    Hence, we start by showing that
    \begin{equation}\label{eq: proof of Beals k lt l}
       \left(\prod_{k=1}^{\ell-1} W_k\right) \ket{\tau_\ell, 0, P} = \ket{\tau_\ell, 0, P}.
    \end{equation}
    Then, we show that at step $k = \ell$, we get the state
    \begin{equation}\label{eq: proof of Beals k eq l}
        W_\ell \ket{\tau_\ell, 0, P} = \ket{\star}\ket{\Psi_{\ell, P}}.
    \end{equation}
    Lastly, we show that
    \begin{equation}\label{eq: proof of Beals k gt l}
       \left(\prod_{k=\ell+1}^{n} W_k\right) \ket{\star}\ket{\Psi_{\ell, P}} = \ket{\star}\ket{\Psi_{\ell, P}}.
    \end{equation}

    First, to prove \cref{eq: proof of Beals k lt l}, it suffices to show that for all $k < \ell$, we have
    \begin{equation}\label{eq: Beals proof step k lt l}
         W_k \ket{\tau_\ell, 0, P} = \ket{\tau_\ell, 0, P}.
    \end{equation}
    We follow the action of the gates making up $W_k$ on the state $\ket{\tau_\ell, 0, P}$ in order.
    Recall that $R(\tau_k)^\dagger$ only acts non-trivially when the transversal register is $\ket{\star}$, which is not the case.
    Then, $U^\dagger$ acts trivially as well for the same reason.
    $V_k$ only acts on the subspace spanned by $\ket{\star}$ and $\ket{\tau_k}$, while $\ket{\tau_\ell}$ is in the orthogonal complement because $k < \ell$.
    Using the same arguments as before, $U$ and $R(\tau_k)$ also act trivially, which proves \cref{eq: Beals proof step k lt l}.

    Second, to prove \cref{eq: proof of Beals k eq l}, note that the gates $R(\tau_\ell)^\dagger$ and $U^\dagger$ act trivially again.
    However, $\ket{\tau_\ell}$ is now changed to $\ket{\star}$ by $V_\ell$.
    Then, $U$ changes the state to
    \begin{equation*}
       U V_\ell U^\dagger R(\tau_\ell)^\dagger \ket{\tau_\ell, 0, P} = \ket{\star}\ket{\phi_\lambda}\ket{P},
    \end{equation*}
    with $\ket{\phi_\lambda}$ from \cref{eq:path extension superposition state}.
    Finally, $R(\tau_\ell)$ applies irreps at $\tau_\ell$ on the path register, obtaining the state $R(\tau_\ell)\ket{\star}\ket{\phi_\lambda}\ket{P} = \ket{\star}\ket{\Psi_{\ell, P}}$.

    Third, to prove \cref{eq: proof of Beals k gt l}, we need to show that for all $k > \ell$, we have
    \begin{equation}\label{eq: Beals proof step k gt l}
        W_k \ket{\star}\ket{\Psi_{\ell,P}} = \ket{\star}\ket{\Psi_{\ell,P}}.
    \end{equation}
    The gate $R(\tau_k)^\dagger$ changes the state to
    \begin{align}\label{eq: Beals state k eq l after L dag}
        R(\tau_k)^\dagger \ket{\star}\ket{\Psi_{\ell,P}}
        &= \ket{\star}\otimes
        \sum_{\mu \in \N^+(\lambda)}
        \sqrt{\frac{d_\mu}{nd_\lambda}}\ket{\mu}\otimes
        \sum_{Q \in \Path(\mu)}
        [\Rep_\mu(\tau^{-1}_k \tau_\ell)]^Q_{P\to\mu}\ket{Q'},
    \end{align}
    where we used the identity $\Rep_\mu(\tau_k)^\dagger\Rep_\mu(\tau_\ell)=\Rep_\mu(\tau_k^{-1}\tau_\ell)$ for unitary representations.

    Next, we show that $U V_k U^\dagger$ acts as identity on the state $R(\tau_k)^\dagger \ket{\star}\ket{\Psi_{\ell,P}}$.
    By definition, $V_k$ is the identity on the orthogonal complement of
    \begin{equation*}
        \mathcal{B}_{k,\lambda}
        \defeq \Span\{\ket{\star,0,Q},\ket{\tau_k,0,Q} \mid Q\in\Path(\lambda)\}.
    \end{equation*}
    Since $U$ is unitary, $UV_kU^\dagger$ is therefore the identity on $\mathcal{A}_{k,\lambda}^{\perp}$, where
    \begin{equation*}
        \mathcal{A}_{k,\lambda}
        \defeq U\mathcal{B}_{k,\lambda}
        = \Span\{\ket{\star}\ket{\phi_\lambda}\ket{Q},\ket{\tau_k,0,Q} \mid Q\in\Path(\lambda)\}.
    \end{equation*}
    This description uses only the prescribed action of $U$ in \cref{eq: U gate Beals}, so it holds for every allowed unitary completion.

    The state in \cref{eq: Beals state k eq l after L dag} has transversal label $\star$, so it is orthogonal to every $\ket{\tau_k,0,Q}$.
    It remains to compute its projection onto
    \begin{equation*}
        \mathcal{A}_{k,\lambda}^\star
        \defeq \Span\{\ket{\star}\ket{\phi_\lambda}\ket{Q} \mid Q\in\Path(\lambda)\}.
    \end{equation*}
    Since $\ket{\phi_\lambda}$ has unit norm and the path states $\ket{Q}$ are orthonormal, the corresponding projector is
    \begin{equation*}
        \Pi_{\mathcal{A}_{k,\lambda}^\star}
        = \ket{\star}\bra{\star}\otimes
        \ket{\phi_\lambda}\bra{\phi_\lambda}\otimes
        \sum_{Q\in\Path(\lambda)}\ket{Q}\bra{Q}.
    \end{equation*}
    Writing $g=\tau_k^{-1}\tau_\ell$ and using \cref{eq:path extension superposition state}, we obtain
    \begin{align*}
        &\Pi_{\mathcal{A}_{k,\lambda}^\star}
        R(\tau_k)^\dagger\ket{\star}\ket{\Psi_{\ell,P}}\\
        &\qquad=
        \ket{\star}\ket{\phi_\lambda}\otimes
        \sum_{Q\in\Path(\lambda)}
        \left(
            \sum_{\mu\in\N^+(\lambda)}
            \frac{d_\mu}{nd_\lambda}
            [\Rep_\mu(g)]^{Q\to\mu}_{P\to\mu}
        \right)\ket{Q}.
    \end{align*}
    Since $k>\ell$, the representatives $\tau_k$ and $\tau_\ell$ belong to distinct left cosets, so $g\notin H$.
    Each parenthesized coefficient therefore vanishes by \cref{lemma:key_lemma_for_beals}, giving
    \begin{equation*}
        \Pi_{\mathcal{A}_{k,\lambda}^\star}
        R(\tau_k)^\dagger\ket{\star}\ket{\Psi_{\ell,P}} = 0.
    \end{equation*}
    Together with the orthogonality to the $\tau_k$-labelled part, this shows that the state lies in $\mathcal{A}_{k,\lambda}^{\perp}$, where $UV_kU^\dagger$ acts as identity.
    Finally, applying $R(\tau_k)$ returns the state to $\ket{\star}\ket{\Psi_{\ell,P}}$, proving \cref{eq: Beals proof step k gt l}.
    The three steps establish that $W_n\cdots W_1$ has the required action on every input basis vector.
    By linearity, the same circuit implements $\IndTrans$ on arbitrary superpositions of these vectors, including superpositions of the control labels $\lambda$.
\end{proof}

\section{Mackey algorithm for the induction transform}\label{sec: Mackey induction transform}

The aim of this section is to present a new way of constructing $\U_{\Ind}$ using Mackey theory.
We give definitions of the smaller unitary operations in our construction, but refrain from implementation details as they differ case by case.

For subgroups $H,K\leq G$, the $(H,K)$-double coset of $g\in G$ is $HgK=\{hgk:h\in H,\ k\in K\}$. The double cosets partition $G$, and an $(H,K)$-double-coset transversal is a set containing one representative of each part.

\begin{lemma}[Mackey decomposition, \protect{\cite[ch.~7]{serre}}]\label{lemma: Mackey}
    Let $H$ and $K$ be subgroups of a finite group $G$, and let $\Omega\subseteq G$ contain exactly one representative of each $(H,K)$-double coset. For any $\C[K]$-module $V$, with corresponding matrix representation $\Rep:K\to\GL(V)$, there is an isomorphism of $\C[H]$-modules
    \begin{equation}
        \Res^G_H \Ind^G_K V \cong \bigoplus_{\omega \in \Omega} \Ind_{H_\omega}^H V^{\omega}.
    \end{equation}
    Here $H_\omega\defeq H\cap\omega K\omega^{-1}$, and $V^\omega$ is the $\C[H_\omega]$-module with underlying vector space $V$ and action
    \begin{equation}
        \Rep^\omega(h)\ket{P}
        \defeq \Rep(\omega^{-1}h\omega)\ket{P},
        \qquad h\in H_\omega.
    \end{equation}
\end{lemma}

For the remainder of this section we specialize to $K=H$. We fix a double-coset transversal $\Omega$ for $H\backslash G/H$, put $H_\omega=H\cap\omega H\omega^{-1}$, choose a left transversal $\T_\omega$ for $H_\omega$ in $H$, and use the resulting left transversal
\begin{equation}\label{eq:Mackey transversals}
    \T = \bigsqcup_{\omega\in\Omega} \T_\omega \omega
\end{equation}
for $H$ in $G$. For the rest of the section we also assume that the branching decompositions $\Ind_H^G V_\lambda$ and $\Ind_{H_\omega}^H V_\kappa$, and the irreducible decompositions of the twisted $H_\omega$-modules $V_\lambda^\omega$, are multiplicity-free. We also fix Gelfand--Tsetlin bases for these decompositions.

\subsection{Circuit for \texorpdfstring{$\U_{\Ind, \lambda}$}{Uind, lambda}}
\label{subsec: Mackey circuit}

As a map of $\C[G]$-modules there is an isomorphism $\Ind_H^G V_\lambda \xrightarrow{\sim}{} \bigoplus_{\mu\in \N^+(\lambda)} V_{\mu}$. When restricted to a map of $\C[H]$-modules, we can apply Mackey decomposition to obtain a commutative diagram of isomorphic $\C[H]$-modules given below.
For $\widetilde\lambda\in\widehat H$, define the two multiplicity-label sets
\begin{align}
    \mathcal M_{\lambda,\widetilde\lambda}
    &\defeq \N^+(\lambda)\cap\N^+(\widetilde\lambda),\\
    \mathcal K_{\lambda,\widetilde\lambda}
    &\defeq
    \left\{(\omega,\kappa):\omega\in\Omega,
    \ \kappa\in\N(\lambda^\omega)\cap
    \N^-_{H_\omega\leq H}(\widetilde\lambda)\right\},
    \label{eq: Mackey multiplicity label sets}
\end{align}
where $\N^+$ and $\N^-$ denote the successors and predecessors, respectively (see \cref{def: successors and predecessors}) and $\N$ is the irreducible support (see \cref{def: irreducible support}).

\begin{figure}[H]
    \centering
    \begin{tikzcd}[column sep=3cm, row sep=1.2cm, nodes={inner sep=3pt}]
    \setwiretype{n}
        \Res^{G}_{H}\Ind^{G}_{H} V_\lambda
        \dar{\text{Mackey}}
        \rar{\U_{\Ind, \lambda}}
        &
        \displaystyle
        \bigoplus_{\widetilde{\lambda} \in \widehat{H}}
        \C[\mathcal M_{\lambda,\widetilde\lambda}]
        \otimes V_{\widetilde{\lambda}}
        \\
    \setwiretype{n}
        \displaystyle
        \bigoplus_{\omega \in \Omega} \Ind^{H}_{H_\omega} V_\lambda^{\omega}
        \dar{\oplus_{\omega} \U_{\tw}^\omega}
        &
        \displaystyle
        \bigoplus_{\widetilde{\lambda} \in \widehat{H}}
        \C[\mathcal K_{\lambda,\widetilde\lambda}]
        \otimes V_{\widetilde{\lambda}}
        \uar{A_\lambda}
        \\
    \setwiretype{n}
        \displaystyle
        \bigoplus_{\omega \in \Omega} \bigoplus_{\kappa \in \N(\lambda^\omega)} \Ind_{H_\omega}^H V_{\kappa}
        \rar{\oplus_{\omega}\oplus_{\kappa} \U_{\Ind, \kappa}^{\omega}}
        &
        \displaystyle
        \bigoplus_{\omega \in \Omega}
        \bigoplus_{\kappa \in \N(\lambda^\omega)}
        \bigoplus_{\widetilde{\lambda} \in \N^+(\kappa)}
        V_{\widetilde{\lambda}}
        \uar{}
    \end{tikzcd}
    \caption{A diagram of $\C[H]$-modules suggesting an implementation of $\U_{\Ind, \lambda}$}
    \label{diag: C[H]-modules}
\end{figure}
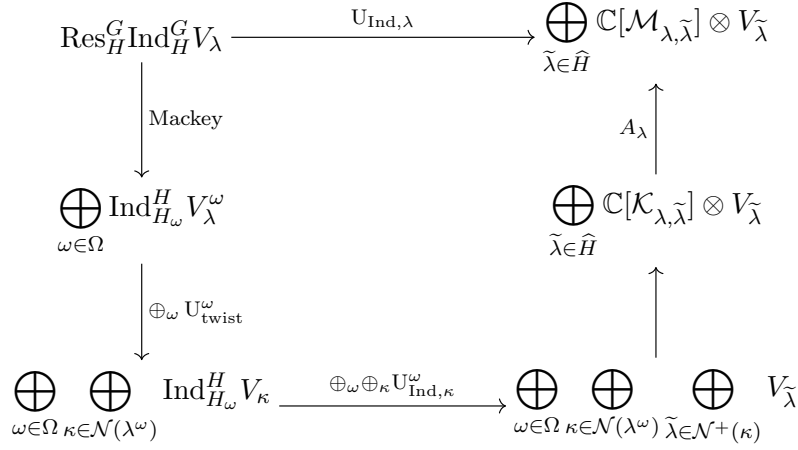

The top arrow is the map we want to construct. The same map (as $\C[H]$-modules) is obtained by following the other path in the diagram.
\begin{itemize}
    \item First we apply Mackey decomposition to obtain a direct sum over all double coset representatives of twisted representations.
    \item Since each $V_{\lambda}^\omega$ is a representation of $H_\omega$, it decomposes into irreducible parts via an intertwiner we call $\U_{\tw}^\omega$.
    \item After this ``untwisting'' we apply the induction operations on the smaller subgroups $H_{\omega}$ for each $\omega\in \Omega$.
    \item Rewriting the obtained module gives a decomposition over $\widetilde{\lambda}\in \widehat{H}$ with certain multiplicities.
    \item The middle-right module is equivalent to the target of $\U_{\Ind,\lambda}$, so there is an intertwiner $A_\lambda$ between them. By Schur's lemma it acts only on the multiplicity spaces
    \begin{equation}
        A_\lambda=
        \bigoplus_{\widetilde\lambda\in\widehat H}
        A_{\lambda,\widetilde\lambda}\otimes I_{V_{\widetilde\lambda}}.
    \end{equation}
\end{itemize}

We use the steps and accompanying diagram of \cref{diag: C[H]-modules} and turn them into a quantum algorithm. Each of the operators in the algorithm corresponds to a step in \cref{diag: C[H]-modules}. The operators need to be defined formally, which is done below. Note that we need to define the $A$-matrix in such a way that we get an equivalence of $\C[G]$-modules and not just over the subalgebra.

\begin{definition}\label{def: Mackey operators}
    We define the following operators.
    \begin{itemize}
    \item The encoding of a transversal element into its uniquely determined double-coset and inner-transversal labels is
    \begin{align}
         \Enc:\C[\T]
         &\longrightarrow
         \bigoplus_{\omega\in\Omega}
         \C[\T_\omega]\otimes\C\ket{\omega},\nonumber\\
        \ket{t}
        &\longmapsto \ket{t_\omega}\ket{\omega},
        \qquad t=t_\omega\omega.
    \end{align}

    \item For each $\omega \in \Omega$, we define $\U_{\tw}^\omega$ as the unitary operator that decomposes the twisted representation $V_\lambda^{\omega}$ into irreducible parts. It is the intertwiner between the $\C[H_\omega]$-modules $V_{\lambda^\omega} \xrightarrow{\sim} \oplus_{\kappa \in \N(\lambda^\omega)} V_{\kappa}$. If we denote by $R_{\rho}$ the action of representation $\rho$ on $V_{\rho}$, then $\U_{\tw}^\omega$ satisfies
    \begin{equation}
        \U_{\tw}^\omega~R_\lambda^\omega~{\U_{\tw}^\omega}^\dagger  =\sum_{\kappa\in \N(\lambda^\omega)} \ket{\kappa}\bra{\kappa} \otimes R_\kappa.
    \end{equation}
    The total untwisting operator is given by
    \begin{equation}\U_{\tw} = \sum_{\lambda \in \widehat{H}}\sum_{\omega\in \Omega} \ket{\lambda}\bra{\lambda}\otimes \U_{\tw}^\omega \otimes \ket{\omega}\bra{\omega}.
    \end{equation}
    \item Controlled on $(\omega,\kappa)$, apply the induction transform $\U_{\Ind,\kappa}^{\omega}$ of \cref{def: IndTrans} for the inclusion $H_\omega\leq H$.
    \item Lastly, the operator $A=\sum_{\lambda,\widetilde\lambda\in\widehat H} |\lambda\rangle\langle\lambda|\otimes A_{\lambda,\widetilde\lambda}\otimes |\widetilde\lambda\rangle\langle\widetilde\lambda|$ is defined on multiplicity spaces by
    \begin{align}
        \left[A_{\lambda,\widetilde\lambda}\right]^{\mu}_{\kappa,\omega}
        &=\sqrt{\frac{d_\mu[H:H_\omega]}
        {d_\kappa d_\lambda d_{\widetilde\lambda}[G:H]}}
        \sum_{\substack{P\to\lambda\\
                         \widetilde P\to\kappa\to\widetilde\lambda}}
        \left[R_\mu(\omega)\right]^
        {\widetilde P\to\kappa\to\widetilde\lambda\to\mu}_
        {P\to\lambda\to\mu}
        \left[(\U_{\tw}^{\omega})^\dagger\right]^P_{
        \widetilde P,\kappa}.
        \label{eq: corrected general A coefficient}
    \end{align}
    \end{itemize}
\end{definition}

\begin{algorithm}[H]
\caption{Mackey implementation of the induction transform}
\label{alg: Mackey induction transform}
On input $\ket{\lambda}\ket{P}\ket{t}$, where $\lambda\in\widehat H$, $P\in\Path(\lambda)$, and $t\in\T$, perform the following operations.
\begin{enumerate}
    \item Apply $\Enc$ to compute the unique factorization $t=t_\omega\omega$ by
    \begin{equation*}
        \ket{t}\longmapsto\ket{t_\omega}\ket{\omega}.
    \end{equation*}
    \item Controlled on $(\lambda,\omega)$, apply $\U_{\tw}^\omega$ to decompose the twisted restriction
    \begin{equation*}
        \ket{P}\longmapsto
        \sum_{\kappa,\widetilde P}
        [\U_{\tw}^\omega]^{\widetilde P,\kappa}_{P}
        \ket{\kappa}\ket{\widetilde P}.
    \end{equation*}
    \item Controlled on $(\omega,\kappa)$, apply the smaller induction transform $\U_{\Ind,\kappa}^{\omega}$ to $\ket{t_\omega}\ket{\widetilde P}$. Denote its output by $\ket{\widetilde\lambda}\ket{Q}$.
    \item Permute the registers as in \cref{fig: Mackey circuit} and, controlled on $(\lambda,\widetilde\lambda)$, apply $A_{\lambda,\widetilde\lambda}$ via
    \begin{equation*}
        \ket{\kappa}\ket{\omega}
        \longmapsto
        \sum_{\mu\in\mathcal M_{\lambda,\widetilde\lambda}}
        [A_{\lambda,\widetilde\lambda}]^{\mu}_{\kappa,\omega}
        \ket{\mu}.
    \end{equation*}
\end{enumerate}
The output labels $\ket{\mu}\ket{Q\to\widetilde\lambda\to\mu}$ encode the target subgroup adapted basis, while $\lambda$ is retained as a control label.
\end{algorithm}

The quantum circuit of this algorithm is given below. We prove that it indeed implements the induction transform.
\begin{figure}[H]
\centering
    \begin{quantikz}[wire types={q,n,q,n,n,q,n}]
    \lstick{\ket{\lambda}} & & & \ctrl{1} & & & & & \ctrl{1} & \rstick{\ket{\lambda}}
    \\
    & & & \gate[3]{\U_{\tw}^{\omega}} & \wire[l][1]["\ket{\kappa}"{above,pos=0.5}]{q} \setwiretype{q} & \ctrl{2} & & & \gate[3]{A_{\lambda, \widetilde{\lambda}}} & \setwiretype{n}
    \\
    \lstick{\ket{P}}& & & & \setwiretype{n} & & & & & \setwiretype{q} \rstick{\ket{\mu}}
    \\
    & & & & \wire[l][1]["\ket{\widetilde{P}}"{above,pos=0.5}]{q} \setwiretype{q} & \gate[2]{\U_{\Ind, \kappa}^{\omega}} & \wire[l][1]["\ket{\widetilde{\lambda}}"{above,pos=0.5}]{q} & \permute{2,4} & & \setwiretype{n}
    \\
    & \gate[3]{\Enc} & \wire[l][1]["\ket{t_\omega}"{above,pos=0.4}]{q}\setwiretype{q} & & & & \wire[l][1]["\ket{Q}"{above,pos=0.5}]{q} & & \ctrl{-1} & \rstick{\ket{\widetilde{\lambda}}}
    \\
    \lstick{\ket{t}}& & \setwiretype{n} & & & & & & &
    \\
    & & \wire[l][1]["\ket{\omega}"{above,pos=0.5}]{q}\setwiretype{q} & \ctrl{-3} & & \ctrl{-2} & & \permute{-2} & & \rstick{\ket{Q}}
    \end{quantikz}
    \caption{Circuit to implement $\U_{\Ind, \lambda}$ via Mackey decomposition.}
    \label{fig: Mackey circuit}
\end{figure}

\begin{theorem}\label{thm: Mackey induction transform circuit}
    \cref{alg: Mackey induction transform}, and equivalently the circuit in \cref{fig: Mackey circuit}, correctly implements $\U_{\Ind,\lambda}$. Moreover, every $A_{\lambda,\widetilde\lambda}$ is unitary on the corresponding multiplicity space.
\end{theorem}
\begin{proof}
	Let $N_{\lambda}$ denote the first three steps of \cref{alg: Mackey induction transform}, namely the circuit up to, but not including $A_\lambda$. By the standing assumptions, $N_\lambda$ and $\U_{\Ind,\lambda}$ are unitary. It therefore suffices to prove
	\begin{equation}
		A_\lambda=\U_{\Ind,\lambda}N_\lambda^\dagger,
		\label{eq: Mackey target identity}
	\end{equation}
	because right-multiplication by $N_\lambda$ then gives $A_\lambda N_\lambda=\U_{\Ind,\lambda}$. Equality in \cref{eq: Mackey target identity} also proves, rather than assumes, that $A_\lambda$ is unitary. Since $A_\lambda$ is block diagonal as in \cref{def: Mackey operators}, every block $A_{\lambda,\widetilde\lambda}$ is unitary on its multiplicity space.

    We now compute the matrix entries of the right-hand side of \cref{eq: Mackey target identity}. Fix $\omega\in\Omega$, $\kappa\in\widehat{H_\omega}$, $\widetilde\lambda\in\widehat H$, and $\mu\in \hatG$.  Let $t_\omega$ run through the chosen transversal $\T_\omega$ of $H_\omega$ in $H$.  For $P\in\Path(\lambda)$ and $Q\in\Path(\widetilde\lambda)$, the definition of the direct induction transform gives
    \begin{align*}
        &\langle\mu,\widetilde\lambda,Q|
        \U_{\Ind,\lambda}|t_\omega\omega,P\rangle \\
        &\qquad=
        \sqrt{\frac{d_\mu}{d_\lambda[G:H]}}\,
        [R_\mu(t_\omega\omega)]^{
            Q\to\widetilde\lambda\to\mu}_{
            P\to\lambda\to\mu}.
    \end{align*}
    On the other hand, untwisting followed by induction from $H_\omega$ to $H$ gives
    \begin{align*}
        &\langle\kappa,\omega,\widetilde\lambda,Q|
        N_\lambda|t_\omega\omega,P\rangle \\
        &\qquad=
        \sqrt{\frac{d_{\widetilde\lambda}}
        {d_\kappa[H:H_\omega]}}
        \sum_{\widetilde P\in\Path(\kappa)}
        [R_{\widetilde\lambda}(t_\omega)]^Q_{
            \widetilde P\to\kappa\to\widetilde\lambda}
        [\U_{\tw}^{\omega}]^{\widetilde P,\kappa}_P.
    \end{align*}

		As $A_{\lambda}$ does not depend on the lowest register in \cref{fig: Mackey circuit} ($V_{\widetilde{\lambda}}$ labeled $\ket{Q}$), we take a partial trace over it, giving
    \begin{equation}\label{eq: A partial trace}
        \left[A_{\lambda,\widetilde\lambda}\right]^{\mu}_{\kappa,\omega}
        =\frac1{d_{\widetilde\lambda}}
        \sum_Q
        \langle\mu,\widetilde\lambda,Q|
        \U_{\Ind,\lambda} N_\lambda^\dagger|
        \kappa,\omega,\widetilde\lambda,Q\rangle.
    \end{equation}
    Substituting these two matrix-element formulas into \cref{eq: A partial trace} yields
    \begin{align*}
        \left[A_{\lambda,\widetilde\lambda}\right]^{\mu}_{\kappa,\omega}
        &=
        \frac1{d_{\widetilde\lambda}}
        \sqrt{\frac{d_\mu}{d_\lambda[G:H]}}
        \sqrt{\frac{d_{\widetilde\lambda}}
        {d_\kappa[H:H_\omega]}} \\
        &\quad{}\cdot
        \sum_{\substack{Q,\ t_\omega\in\T_\omega\\
                        P\in\Path(\lambda),\
                        \widetilde P\in\Path(\kappa)}}
        [R_\mu(t_\omega\omega)]^{
            Q\to\widetilde\lambda\to\mu}_{
            P\to\lambda\to\mu}
        \overline{
        [R_{\widetilde\lambda}(t_\omega)]^Q_{
            \widetilde P\to\kappa\to\widetilde\lambda}}
        [(\U_{\tw}^{\omega})^\dagger]^P_{
            \widetilde P,\kappa},
    \end{align*}
   where the sum is over four variables. Since the bases are subgroup adapted, $R_\mu(t_\omega)$ acts as $R_{\widetilde\lambda}(t_\omega)$ on $V_{\widetilde\lambda}$ so that
    \begin{align*}
        [R_\mu(t_\omega\omega)]^{
            Q\to\widetilde\lambda\to\mu}_{
            P\to\lambda\to\mu}
        &=
        \sum_{R\in\Path(\widetilde\lambda)}
        [R_{\widetilde\lambda}(t_\omega)]^Q_R
        [R_\mu(\omega)]^{
            R\to\widetilde\lambda\to\mu}_{
            P\to\lambda\to\mu}.
    \end{align*}
    For each $t_\omega$, unitarity of $R_{\widetilde\lambda}(t_\omega)$ gives
    \begin{align*}
        \sum_Q
        [R_{\widetilde\lambda}(t_\omega)]^Q_R
        \overline{
        [R_{\widetilde\lambda}(t_\omega)]^Q_{
            \widetilde P\to\kappa\to\widetilde\lambda}}
        =
        \delta_{R,\,
        \widetilde P\to\kappa\to\widetilde\lambda}.
    \end{align*}
    Hence the sums over $Q$ and $R$ collapse.
The remaining summand no longer depends on $t_\omega$.
Thus that sum gives a factor of $|\T_\omega|=[H:H_\omega]$.
  The total scalar coefficient is therefore
    \begin{align*}
        &\frac{[H:H_\omega]}{d_{\widetilde\lambda}}
        \sqrt{\frac{d_\mu d_{\widetilde\lambda}}
        {d_\lambda d_\kappa[G:H][H:H_\omega]}}
        \\
        &\qquad=
        \sqrt{\frac{d_\mu[H:H_\omega]}
        {d_\kappa d_\lambda d_{\widetilde\lambda}[G:H]}}.
    \end{align*}
    The remaining path sum is
    \begin{equation*}
        \sum_{\substack{P\to\lambda\\
                        \widetilde P\to\kappa\to\widetilde\lambda}}
        [R_\mu(\omega)]^{
            \widetilde P\to\kappa\to\widetilde\lambda\to\mu}_{
            P\to\lambda\to\mu}
        [(\U_{\tw}^{\omega})^\dagger]^P_{
            \widetilde P,\kappa},
    \end{equation*}
    so the preceding expansion is exactly \cref{eq: corrected general A coefficient}.  Thus the operator specified in \cref{def: Mackey operators} equals $\U_{\Ind,\lambda} N_\lambda^\dagger$, and consequently $A_\lambda N_\lambda=\U_{\Ind,\lambda}$.  This is the circuit identity claimed in the theorem.
\end{proof}

\subsection{When \texorpdfstring{$\omega$}{omega} normalizes
\texorpdfstring{$H_\omega$}{H omega}}
\label{subsec: omega normalizes}

The circuit and operators in the previous subsection are intended to result in an easier QFT implementation. This also means we need to understand the  untwisting operator $\U_{\tw}$, which for now has only been defined by its intertwining relation. The actual matrix entries are very much based on the group $G$ and its subgroups. In this subsection we explain what happens when $\omega$ normalizes $H_{\omega}$.

Fix $\omega\in\Omega$ and suppose that $\omega^{-1}H_\omega\omega=H_\omega$.  Conjugation by $\omega$ is then an automorphism of $H_\omega$.  For $\kappa\in\widehat{H_\omega}$, define the $\omega$-twisted representation on $H_{\omega}$ by
\begin{equation}
    R_\kappa^\omega(h):=R_\kappa(\omega^{-1}h\omega),
    \qquad h\in H_\omega,
    \label{eq: normalized subgroup irrep twist}
\end{equation}
and denote it by $\kappa^{\omega}$. Thus $V_{\kappa^\omega}\simeq V_\kappa^\omega$ as $H_\omega$-modules.

\begin{lemma}
\label{lem: normalizer permutes Homega irreps rewritten}
    The map
    \begin{equation}
        \tau_\omega:\widehat{H_\omega}\longrightarrow
        \widehat{H_\omega},
        \qquad \kappa\longmapsto\kappa^\omega,
        \label{eq: normalizer permutation on dual}
    \end{equation}
    is a bijection.  Its inverse is $\tau_{\omega^{-1}}$, and $d_{\kappa^\omega}=d_\kappa$.
\end{lemma}

\begin{proof}
    Because $h\mapsto\omega^{-1}h\omega$ is an automorphism of $H_\omega$, \cref{eq: normalized subgroup irrep twist} is a representation.  If $W\subseteq V_\kappa$ is invariant under $R_\kappa^\omega(H_\omega)$, then it is invariant under $R_\kappa(H_\omega)$, since conjugation by $\omega^{-1}$ is surjective.
    Hence $R_\kappa^\omega$ is irreducible whenever $R_\kappa$ is.
    Twisting an intertwiner between two representations gives an intertwiner between their twists, so the assignment is well defined.

    Since $\omega$ also normalizes $H_\omega$, twisting in the opposite direction is defined, and for every $h\in H_\omega$,
    \begin{equation*}
        (R_\kappa^\omega)^{\omega^{-1}}(h)
        =R_\kappa\!\left(
          \omega^{-1}(\omega h\omega^{-1})\omega\right)
        =R_\kappa(h).
    \end{equation*}
    Thus $\tau_{\omega^{-1}}\tau_\omega$ is the identity; the reverse composition is identical.  Therefore $\tau_\omega$ is a bijection.
    Twisting does not change the carrier space, so it also preserves dimensions.
\end{proof}

The preceding permutation determines the irreducible decomposition of the twisted module directly.

\begin{lemma}
\label{lem: normalized twisted restriction decomposition rewritten}
    Let $\lambda\in\widehat H$ and write the restriction to $H_\omega$ as
    \begin{equation}
        \Res^H_{H_\omega}V_\lambda
        \simeq
        \bigoplus_{\kappa\in\widehat{H_\omega}}
        \C^{m_{\lambda,\kappa}}\otimes V_\kappa.
        \label{eq: Homega restriction with multiplicities}
    \end{equation}
    Then, as an $H_\omega$-module,
    \begin{align}
        V_\lambda^\omega
        &\simeq
        \bigoplus_{\kappa\in\widehat{H_\omega}}
        \C^{m_{\lambda,\kappa}}\otimes V_{\kappa^\omega}
        \label{eq: twisted restriction decomposition}
    \end{align}
    In particular, if the restriction is multiplicity-free then
    \begin{equation}
        V_\lambda^\omega
        \simeq
        \bigoplus_{\kappa\in\N^-_{H_\omega}(\lambda)}
        V_{\kappa^\omega}.
        \label{eq: multiplicity free twisted restriction decomposition}
    \end{equation}
\end{lemma}

\begin{proof}
    The $H_\omega$-action on $V_\lambda^\omega$ is obtained from the restricted action on $V_\lambda$ by precomposing with the automorphism $h\mapsto\omega^{-1}h\omega$.  Twisting commutes with direct sums and acts trivially on multiplicity spaces.  Applying it to \cref{eq: Homega restriction with multiplicities} gives \cref{eq: twisted restriction decomposition}.  The multiplicity-free formula is the corresponding special case.
\end{proof}

\paragraph{Interpretation of the untwisting operator.}
Assuming we have a Gelfand--Tsetlin basis, \cref{eq: Homega restriction with multiplicities} holds.
For every $\kappa$, let
\begin{equation*}
    U_{\omega,\kappa}:V_\kappa\longrightarrow V_{\kappa^\omega}
\end{equation*}
be a unitary intertwiner satisfying
\begin{equation}
    U_{\omega,\kappa}R_\kappa^\omega(h)
    U_{\omega,\kappa}^\dagger
    =R_{\kappa^\omega}(h),
    \qquad h\in H_\omega.
    \label{eq: normalizer block intertwiner}
\end{equation}
The untwisting operator from \cref{def: Mackey operators} may then be chosen as
\begin{equation}
    \U_{\tw}^\omega:
    \ket{\kappa}\ket{v}
    \longmapsto
    \ket{\kappa^\omega}\,
    U_{\omega,\kappa}\ket{v}.
    \label{eq: normalizer untwisting operator}
\end{equation}
Thus the only change in the irreducible decomposition is the permutation $\kappa\mapsto\kappa^\omega$.  If the representative for $\kappa^\omega$ is chosen to be the twisted model $R_{\kappa^\omega}=R_\kappa^\omega$, then $U_{\omega,\kappa}=I$ and \cref{eq: normalizer untwisting operator} is a permutation of the $\kappa$-label register.

\section{Application 1: QFT for \texorpdfstring{$\GL_2(\mathbb{F}_q)$}{GL2(Fq)}}\label{sec: GL2}

An application of the induction transform using Mackey theory is the QFT for the group of invertible $2 \times 2$ matrices over a finite field. Let $q=p^r$ be a prime power and consider the group $G:=\GL_2(\F_q)$. The main result is the following theorem.

\begin{theorem}[QFT for $\GL_2(\F_q)$]\label{thm: GL2 QFT}
	For every $\varepsilon>0$, there exists a uniform quantum circuit of size $\operatorname{poly}\left(\log q, \log \frac{1}{\varepsilon}\right)$ that implements the QFT of $\GL_2(\F_q)$ to operator-norm error at most $\varepsilon$.
\end{theorem}
The \hyperref[proof: GL2 QFT]{proof} is given at the end of \cref{sec: GL2 QFT circuit}.
We prove this theorem by induction along the chain
\begin{equation}\label{eq: subgroup chain}
	\{\I\} \leq T \leq B \leq \GL_2(\F_q),
\end{equation}
where $T$ consists of diagonal matrices and $B$ of upper triangular ones.
We perform the induction from $T$ to $B$ directly and use \cref{alg: Mackey induction transform} for the second induction.
The Bruhat decomposition in this case has two double cosets, and the intersection group of the non-identity element is $T$.
Thus, the smaller induction step in the Mackey algorithm is the induction from $T$ to $B$ which is already in our subgroup chain.

First, in \cref{sec: rep theory of GL2} we describe the representation theory of the subgroup chain \eqref{eq: subgroup chain}. Then in \cref{subsec: sub_adap bases} we give explicit formulas for each representation in the Gelfand--Tsetlin basis, and in \cref{sec: twiddle factors} compute matrix coefficients for transversal actions.
Finally, in \cref{sec: GL2 QFT circuit} we construct the $\GL_2(\F_q)$ QFT quantum circuit using the Mackey approach which we introduced in \cref{sec: Mackey induction transform}. The hardest part in this circuit is the $A$-matrix, which contains a $q\times q$-block that is constructed in \cref{subsec: polylog GL2 A}.

\begin{remark}
	See \cref{app: finite fields} for an introduction to finite fields. We often use an isomorphism $X_q:= \widehat{\F_q^*} \cong \F_q^*$ that allows us to interpret character labels as invertible field elements. This is done throughout this section. Note that the isomorphism is non-canonical and depends on the choice of generator of $\F_q^*$.
\end{remark}

\subsection{Representation theory of \texorpdfstring{$\GL_2(\F_q)$}{GL2(Fq)} and subgroups}\label{sec: rep theory of GL2}

We give an overview of the representation theory of $G$ and its Borel subgroup $B$.
This section mostly follows Ref.~\cite{Piatetski1983complex}. Other sources on the representation theory of $\GL_2(\F_q)$ are Refs.~\cite{prasad2007representations, LiNote} (also containing explicit constructions) and \cite[chapter 2]{bushnell2006Langlands} (which is a more concise overview).

Define the following subgroups and elements of $G$:
\begin{align}
	B &:= \{\psm{
		a&b\\0&d
	} \in G\}\quad &&(\text{Borel subgroup}), \label{eq: B}\\
	U &:=  \{\psm{
		1&b\\0&1
	} \in G\} &&(\text{unipotent/normal subgroup of $B$}),\\
	T&:= \{\psm{
		a&0\\0&d
	} \in G\} &&(\text{split maximal torus in } G),\\
	Z&:= \{\psm{
		a&0\\0&a
	} \in G\} &&(\text{the center of $G$}),\\
	P&:= \{\psm{a&b\\0&1}\in G\}.
\end{align}
Also define the following two group elements:
\begin{equation*}
	w:= \psm{
	0&1\\1&0
} \quad (\text{Weyl representative, }w^2=\I), \qquad
	w':= \psm{
	0&1\\-1&0
} .
\end{equation*}
There is a semidirect product decomposition $B = T \ltimes U$. Moreover, $G$ satisfies a Bruhat decomposition
\begin{equation}\label{eq: Bruhat decomposition}
	G  = B~W~B, \quad W=\{\I, w\},
\end{equation}
where $W$ is known as the \emph{Weyl group}.

\subsubsection{Irreducible representations of $B$ and $G$}\label{subsec: B G irreducibles}

The Borel subgroup $B < G$ has the following representations.

\paragraph{One-dimensional}
Let $\alpha, \beta \in X_q$ label two characters $\chi_{\alpha}, \chi_{\beta}$ of $\F_q^*$. We combine them to create a character on $T$ via $\chi_{\alpha,\beta}(\psm{a&0\\0&d})=\chi_\alpha(a)\chi_\beta(d)$.
We extend this character trivially to $U$ to obtain the one-dimensional representation of $B$
\begin{equation}\label{eq: one-dim B}
	\chi_{\alpha, \beta}\psm{a&b\\0&d}
    =\chi_{\alpha}(a)\chi_{\beta}(d).
\end{equation}

\paragraph{Higher-dimensional}
Consider the representation $\rho$ of $P=\{\psm{a&b\\0&1}: a\in \F_q^*, b\in \F_q\}$ acting on  $\C[\F_q^*]$ via
\begin{equation}\label{eq: rho on P}
	\rho\psm{a&b\\0&1} \: \ket{x} = \psi\left(\frac{b}{a}x\right)\ket{\frac{1}{a}x}.
\end{equation}
We extend $\rho$ to $B$ by picking a multiplicative character $\chi_\gamma$ and defining $\rho_{\gamma}\psm{d&0\\0&d}\ket{x} = \chi_{\gamma} (d)\ket{x}$. The total action on $B = Z  \times P$ is then
\begin{equation}\label{eq: rho_gamma action}
	\rho_{\gamma}\psm{a&b\\0&d}\ket{x} = \rho_{\gamma}\left(\psm{d&0\\0&d}\psm{a/d & b/d\\0&1} \right)\ket{x}  =  \chi_{\gamma}(d)\psi\left(\frac{b}{a}x\right) \ket{\frac{d}{a}x}.
\end{equation}

\begin{remark}
	The representation $\rho$ of $P$ can be obtained in the following way. View $\psi \in \widehat{\F_q^+}$ as a character on $U = \{\psm{1&b\\0&1}: b\in \F_q\}$. Induce this character to $P$ to obtain a $(q-1)$-dimensional representation $\rho'$. A transversal set of $P/U$ is given by $\{\psm{x&0\\0&1}: x\in \F_q^*\}$.
	The induced representation is then given by
	\begin{equation}
		\rho'\psm{a&b\\0&1}\ket{x} = \psi\left(\frac{b}{ax}\right)\ket{ax}.
	\end{equation}
	Intertwining with the map $(\ket{x}\mapsto \ket{1/x})$ gives an equivalence between $\rho'$ and $\rho$.
\end{remark}
The irreducible representations of $G$ are classified as follows.

\paragraph{Principal series}
The representations of $B$ can be extended to $G$ by \emph{parabolic induction}. For $\alpha,\beta\in X_q$ we induce $\chi_{\alpha,\beta}$ from $B$ to $G$ to obtain the \emph{principal series} representation. The construction is as follows. We choose a transversal set of $G/B$ to be
\begin{equation}\label{eq: transversal set G/B}
	\{t_x: x\in \F_q\}\cup \{t_{\infty}\}:=\{\psm{x&1\\1&0}: x\in \F_q\} \cup \{\psm{1&0\\0&1}\}.
\end{equation}
The set of transversals can be identified with $\P^1(\F_q)$, a projective copy of $\F_q$. In particular, they satisfy the following: for any $g=\psm{a&b\\c&d} \in G$,
\begin{align}
	g \cdot t_x &=
	\begin{cases}
		t_{gx}\cdot\psm{cx+d &c\\0 & \frac{\det(g)}{cx+d}} &\text{ if } cx+d\neq0\\
		t_{\infty} \cdot \psm{ax+b &a \\0&c} &\text{ if } cx+d=0,
	\end{cases}\label{eq: g action projective 1}\\
	g \cdot t_{\infty} &=
	\begin{cases}
		t_{a/c}\cdot \psm{c&d\\0&\frac{-\det(g)}{c}} &\text{ if } c\neq 0\\
		t_{\infty} \cdot \psm{a&b\\0&d} &\text{ if } c=0.
	\end{cases}\label{eq: g action projective 2}
\end{align}
By $gx\in \P^1(\F_q)$ we mean the standard action of $\GL_2(\F_q)$ on the projective line, given by $\psm{a&b\\c&d}x = \frac{ax+b}{cx+d}$. In particular, it holds that $g\cdot t_{x} = t_{gx} \cdot b_x^g$ for all $x\in \P^1(\F_q)$, where $b_x^g\in B$. Thus, the principal series representation acting on the space $\C[\P^1(\F_q)]$ is given by
\begin{equation} \label{eq: full chi_ab action}
	I_{\alpha, \beta} \psm{a&b\\c&d}\ket{x} = \chi_{\alpha, \beta}(b_x^g) \ket{gx},
\end{equation}
where the term $b_x^g$ is given by \cref{eq: g action projective 1,eq: g action projective 2}.
Conveniently, the principal series is irreducible, except when $\alpha =\beta$.
\begin{lemma}
	The principal series $I_{\alpha, \beta}$ is irreducible whenever $\alpha\neq\beta$. Otherwise, it splits into a one-dimensional and a Steinberg representation. Moreover, $I_{\alpha,\beta}$ and $I_{\gamma, \delta}$ are isomorphic if and only if $[\alpha = \gamma,~\beta=\delta]$ or $[\alpha=\delta,~\beta = \gamma]$.
\end{lemma}

\paragraph{Steinberg and determinant representations}
Whenever $\alpha=\beta$ in the principal series representation, it splits into two irreducible parts of dimensions $1$ and $q$. They are called the \emph{(twisted) determinant and Steinberg representation} respectively.

Following \cref{eq: full chi_ab action}, the action of $I_{\alpha, \alpha}$ on the standard basis of $\C[\P^1(\F_q)]$ is given by
\begin{align}
	I_{\alpha, \alpha} (g) \ket{x} &= \begin{cases}
		\chi_{\alpha}(\det(g))\ket{gx} &\text{ if } gx\neq \infty\\
		\chi_{\alpha}(-1)\chi_{\alpha}(\det(g))\ket{\infty} & \text{ if }gx = \infty
	\end{cases}\quad \text{ for } x\in \F_q, \\
	I_{\alpha, \alpha} (g) \ket{\infty} &= \begin{cases}
		\chi_{\alpha}(-1)\chi_{\alpha}(\det(g))\ket{g\infty} &\text{ if } g\infty\neq \infty\\
		\chi_{\alpha}(\det(g))\ket{\infty} & \text{ if }g\infty = \infty.
	\end{cases}
\end{align}
It follows that the vector
\begin{equation}\label{eq: v}
	\ket{v}:= \frac{1}{\sqrt{q+1}}\left(\chi_{\alpha}(-1) \ket{\infty} + \sum_{x\in \F_q}\ket{x}\right)
\end{equation}
spans a one-dimensional subspace. Indeed, distinguishing the cases $g\infty=\infty$ and $g\infty\neq\infty$, the preceding action formulas and $\chi_{\alpha}(-1)^2=1$ show that its action is given by
\begin{equation}\label{eq: determinant action on v}
	I_{\alpha, \alpha} (g) \ket{v} = \chi_{\alpha}(\det(g)) \ket{v}.
\end{equation}
This we call the (twisted) determinant representation and we denote it by $\det_{\alpha}$. The orthogonal complement is called the Steinberg representation and is denoted by $\St_{\alpha}$.

\paragraph{Cuspidal representations}
Any irreducible representation that is not part of an induction from the Borel subgroup is called \emph{cuspidal}. A construction is presented in Ref.~\cite[section 13]{Piatetski1983complex}. Using the Bruhat decomposition \eqref{eq: Bruhat decomposition} we only have to specify the cuspidal on $B$ and the element $w'\in G$, with the remark that $w = w'\psm{-1&0\\0&1}$.

The cuspidal representations are denoted by $\pi_{\theta}$ and indexed by Frobenius orbits $\{\theta,\theta^q\}$ of regular characters $\theta\in X_{q^2}$, where $\theta^q\neq\theta$ and $\pi_\theta\simeq\pi_{\theta^q}$.
Let $\gamma$ satisfy $\theta|_{\F_q^*}=\chi_\gamma$.  The cuspidal representation $\pi_\theta$ acts on $\C[\F_q^*]$ by
\begin{equation}\label{eq: cuspidal action}
	\pi_{\theta}\psm{a&b\\0&d} \ket{x}
    =\chi_{\gamma}(d) \psi\left(\frac{b}{a}x\right) \ket{\frac{d}{a}x},
\end{equation}
and
\begin{equation}\label{eq: cuspidal w prime}
	\pi_{\theta}(w')\ket{x}
    =\theta(x^{-1}) \sum_{y\in \F_q^*}j_\theta(yx)\ket{y},
    \qquad
    j_\theta(u)=-\frac{1}{q}\sum_{\substack{v\in \F_{q^2}^*\\ \No(v)=u}}
    \psi(\tr(v))\theta(v),
\end{equation}
where $\tr$ and $\No$ denote the finite field trace and norm, see \cref{eq: tr,eq: N} in \cref{app: finite fields}.

The irreducible representations of $G$ are summarized in \cref{tab:GL2-irreps}.
\begin{table}[h]
\centering
\small
\setlength{\tabcolsep}{6pt}
\renewcommand{\arraystretch}{1.15}
\begin{tabular}{@{}llrr@{}}
\toprule
\textbf{Family} & \textbf{Label} & \textbf{Dimension} & \textbf{Number} \\
\midrule
Determinant & $\chi_\alpha\circ\det$ & $1$ & $q-1$ \\
Steinberg & $\St_\alpha$ & $q$ & $q-1$ \\
Principal series & $I_{\alpha,\beta}$ & $q+1$ & $(q-1)(q-2)/2$ \\
Cuspidal & $\pi_\theta$ & $q-1$ & $q(q-1)/2$ \\
\bottomrule
\end{tabular}
\caption{Irreducible representations of $\GL_2(\F_q)$.}
\label{tab:GL2-irreps}
\end{table}
The sum of their squared dimensions
\begin{equation*}
    (q-1)\cdot 1^2 + (q-1)\cdot q^2 + \frac{(q-1)(q-2)}{2} \cdot (q+1)^2 + \frac{q(q-1)}{2}\cdot (q-1)^2 = (q^2-1)(q^2-q) = |G|
\end{equation*}
tells us that these are all the irreducible representations of $G$.

\subsection{Gelfand--Tsetlin bases for \texorpdfstring{$\GL_2(\F_q)$}{GL2(Fq)}}\label{subsec: sub_adap bases}
The irreducible representations of $B$ and $G$ are listed in \cref{subsec: B G irreducibles}. For completeness we note the ones of $T$ as well.
\begin{lemma}
    All the irreducible representations of $T$ are one-dimensional and given by
    \begin{equation}
        \chi_{\alpha, \beta}\psm{a&0\\0&d} = \chi_{\alpha}(a)\chi_{\beta}(d)
    \end{equation}
     with $\alpha,\beta\in X_q$, creating a total of $(q-1)^2$ characters.
\end{lemma}

\begin{lemma}\label{lem: decomposition of restrictions}
    The following restrictions of irreducible representations hold:
    \begin{align}
        \Res_T^B\chi_{\alpha,\beta}
        &=\chi_{\alpha,\beta},\\
        \Res_T^B\rho_\gamma
        &=\bigoplus_{\delta\in X_q}
          \chi_{\delta,\delta^{-1}\gamma},\\
        \Res_B^G\det_\alpha
        &=\chi_{\alpha,\alpha},\\
        \Res_B^G\St_\alpha
        &=\chi_{\alpha,\alpha}\oplus\rho_{\alpha^2},\\
        \Res_B^G I_{\alpha,\beta}
        &=\chi_{\alpha,\beta}\oplus\chi_{\beta,\alpha}
          \oplus\rho_{\alpha\beta},\qquad \alpha\neq\beta,\\
        \Res_B^G\pi_\theta
        &=\rho_\gamma,
        \qquad \theta|_{\F_q^*}=\chi_\gamma.
    \end{align}
    In particular, restriction along the chain $T < B < \GL_2(\F_q)$ is multiplicity-free.
\end{lemma}
\begin{proof}
    All these relations can be checked by taking inner products of characters and/or using Mackey decomposition in \cref{lemma: Mackey}.
\end{proof}

Now that we know the restrictions are multiplicity free, we need to find subgroup-adapted bases for each irreducible representation of $G$.

\subsubsection{From $T$ to $B$}\label{sec: T2B}

We need to study the restriction of $\rho_{\gamma}$ to $T$. By \cref{eq: rho_gamma action}, it follows that $\Res_T^B \rho_{\gamma} \psm{a&0\\0&d}$ equals a permutation matrix $\{x \mapsto \frac{d}{a}x\}$ and a global scaling by $\chi_{\gamma} (d)$. To diagonalize this matrix, consider the multiplicative Fourier basis given by
\begin{equation}\label{eq: mult fourier basis}
    \ket{w_{\delta}} := \frac{1}{\sqrt{q-1}} \sum_{x\in \F_q^*} \chi_{\delta}(x) \ket{x} \text{ for all } \delta\in \F_q^*.
\end{equation}
The action of $\rho_{\gamma}$ restricted to $T$ is then diagonalized and given by
\begin{align*}
    \rho_{\gamma} \psm{a&0\\0&d} \ket{w_\delta} &= \frac{1}{\sqrt{q-1}} \sum_{x\in \F_q^*} \chi_{\gamma}(d) \chi_\delta(x) \ket{\frac{d}{a}x} = \chi_{\gamma}(d)\frac{1}{\sqrt{q-1}}\sum_{y\in \F_q^*} \chi_\delta\left(\frac{a}{d}y\right) \ket{y}\\
    &=\chi_{\gamma}(d)\chi_{\delta}(\frac{a}{d}) \ket{w_\delta} = \chi_{\delta}(a)\chi_{\delta^{-1}\gamma}(d) \ket{w_\delta}.
\end{align*}

\subsubsection{From $B$ to $G$}

\paragraph{Principal series} We start by studying the principal series representation $I_{\alpha, \beta}$.  Recall by \cref{lem: decomposition of restrictions} that the restriction to $B$ needs to give two one-dimensional subspaces and one space of dimension $q-1$. The basis used in \cref{eq: full chi_ab action} gives a permutation of the basis elements by $x\mapsto \frac{ax+b}{d}$, which is sparse but not subgroup adapted. The permutation, as in \cref{sec: T2B}, directs us toward the Fourier basis.
\begin{itemize}
    \item Note that $B$ is the stabilizer of $\infty$ under the standard action. Therefore $\ket{\infty}$ spans an invariant subspace such that
    \begin{equation*}
        I_{\alpha, \beta} (g) \ket{\infty } = \chi_{\alpha, \beta}(g)\ket{\infty} \quad \text{ for all } g\in B.
    \end{equation*}
    \item Consider the additive (inverse) Fourier basis of $\F_q$ given by
    \begin{equation*}
        \ket{v_{t}}:= \frac{1}{\sqrt{q}}\sum_{x\in \F_q} \overline{\psi(tx)} \ket{x} \quad \text{ for all } t\in \F_q.
    \end{equation*}
    The action of $B$ on this basis is given by
    \begin{align}\label{eq: principal series on basis}
        I_{\alpha, \beta} \psm{a&b\\0&d} \ket{v_t}
        &= \frac{1}{\sqrt{q}} \sum_{x\in \F_q} \chi_{\alpha}(d) \chi_{\beta}(a) \overline{\psi(tx)} \ket{\frac{ax+b}d} \nonumber\\
        &= \chi_{\beta}(a) \chi_{\alpha}(d) \frac{1}{\sqrt{q}} \sum_{y\in \F_q} \overline{\psi(t \frac{dy-b}{a})} \ket{y}\nonumber\\
        &= \chi_{\beta}(a) \chi_{\alpha}(d) \psi(t \frac{b}{a}) \ket{v_{\frac{d}{a}t}}.
    \end{align}
    Thus $\ket{v_0}$ gives us another one-dimensional invariant subspace.
    \item The representation given by \cref{eq: principal series on basis} is equivalent to the one given in \cref{eq: rho_gamma action} after the basis change $\ket{v_t'} := \chi_{\beta^{-1}}(t)\ket{v_t}$ for all $t\in \F_q^*$. Indeed, the action on this basis is given by
    \begin{align*}
        I_{\alpha, \beta} \psm{a&b\\0&d} \ket{v_t'} &= \chi_{\beta^{-1}}(t) \chi_{\beta}(a) \chi_\alpha(d) \psi(t \frac{b}{a}) \ket{v_{\frac{d}{a}t}}\\
        &= \chi_{\alpha\beta}(d) \psi(t \frac{b}{a}) \chi_{\beta^{-1}}(t\frac{d}{a})\ket{v_{\frac{d}{a}t}}\\
        &= \chi_{\alpha\beta}(d)\psi(t \frac{b}{a}) \ket{v_{\frac{d}{a}t}'}.
    \end{align*}
\end{itemize}
Thus the Gelfand--Tsetlin basis for $I_{\alpha, \beta}$ is
\begin{align}\label{eq: GT principal}
    \{\ket{\infty}, \ket{v_0}, \ket{w_\delta}: \delta\in \F_q^*\}, \text{ with } \ket{w_\delta} &= \frac{1}{\sqrt{q-1}}\sum_{t\in \F_q^*} \chi_\delta(t)\chi_{\beta^{-1}}(t) \;\ket{v_t},\\
    \ket{v_t} &= \frac{1}{\sqrt{q}}\sum_{x\in \F_q} \overline{\psi(tx)}\;\ket{x}.\nonumber\end{align}

\paragraph{Steinberg and determinant representations}
We find a Gelfand--Tsetlin basis for both of these representations.
Recall from \cref{eq: v} the vector
\begin{equation}
    \ket{v}:= \frac{1}{\sqrt{q+1}}\left(\chi_{\alpha}(-1) \ket{\infty} + \sum_{x\in \F_q}\ket{x}\right) = \frac{1}{\sqrt{q+1}}\left(\chi_\alpha(-1)\ket{\infty} + \sqrt{q}\,\ket{v_0}\right)
\end{equation}
spans a one-dimensional subspace for the determinant representation, which is trivially subgroup adapted. The space of the Steinberg representation can then be spanned by
\begin{equation}
    \ket{v^\perp}:=\frac{1}{\sqrt{q+1}}\left(\sqrt{q}\,\chi_\alpha(-1)\ket{\infty} - \ket{v_0}\right) \text{ and } \ket{w_{\delta}}, \; \delta \in \F_q^*.
\end{equation}
Thus the basis for $(\chi_\alpha\circ\det)\oplus\St_\alpha$ is almost the same as for $I_{\alpha, \beta}$, but with a different combination of $\ket{\infty}$ and $\ket{v_0}$.

\paragraph{Cuspidal representations} The cuspidal representation $\pi_{\theta}$ is irreducible when restricted to $B$. The action on $\C[\F_q^*]$ is given by \cref{eq: cuspidal action}, which corresponds precisely to \cref{eq: rho_gamma action}. The Gelfand--Tsetlin basis is thus given by the multiplicative Fourier basis from \cref{eq: mult fourier basis}.

\subsection{Action of transversals}\label{sec: twiddle factors}
This subsection computes, in the Gelfand--Tsetlin bases fixed above, the representation matrices required by the two induction steps in the chain of subgroups.
For \(T < B\), these are the matrices of all chosen transversal representatives.
For \(B < G\), the Mackey circuit of \cref{alg: Mackey induction transform} only requires the matrices of the nontrivial double-coset representative \(w\), since every non-identity transversal element has the form \(uw\) with \(u\in B\). We calculate the explicit coefficients for each irrep action at these transversal elements. This subsection also makes use of Gauss sums, which are defined in \cref{subsec: Gauss sums}.

\subsubsection{From \texorpdfstring{$T$}{T} to \texorpdfstring{$B$}{B}}\label{subsubsec: from T to B}
We only have to study the transversal actions for $\rho_\gamma$ defined in \cref{eq: rho_gamma action}, as the one-dimensional ones are already given by \cref{eq: one-dim B}. The natural transversal set for $B/T$ is given by
\begin{equation*}
    u_b:= \psm{1&b\\0&1} \text{ for all } b \in \F_q.
\end{equation*}
The Gelfand--Tsetlin basis is given by \cref{eq: mult fourier basis}, on which the transversal set acts by
\begin{equation*}
    \rho_\gamma \psm{1&b\\0&1} \ket{w_\delta}
    =\frac{1}{\sqrt{q-1}}
      \sum_{x\in\F_q^*}\chi_\delta(x)\psi(bx)\ket{x}.
\end{equation*}
It follows that
\begin{align}\label{eq: twiddle factor B}
    \braket{w_\beta}{\rho_\gamma\psm{1&b\\0&1}| w_\alpha}
    &=\frac{1}{q-1}\sum_{x\in\F_q^*}
      \chi_{\alpha\beta^{-1}}(x)\psi(bx) \nonumber\\
    &=\frac{1}{q-1}G_q(\alpha\beta^{-1},b),
\end{align}
where $G_q$ is a Gauss sum, see \cref{def: Gauss sum} in \cref{subsec: Gauss sums}.

\subsubsection{From \texorpdfstring{$B$}{B} to \texorpdfstring{$G$}{G}}
The second set of transversal actions is a bit harder to compute. Using the transversal set of \cref{eq: transversal set G/B}, we can write each element as
\begin{equation*}
    t_x = \psm{1&x\\0&1} \cdot \psm{0&1\\1&0}.
\end{equation*}
Since the first part is an element of $B$, this will act as a block diagonal matrix on the Gelfand--Tsetlin basis, and the precise action is described in the previous section. In particular, we only have to study the representation action of the element $w=\psm{0&1\\1&0}$ on the Gelfand--Tsetlin basis. Since $w^2=\I$ it holds that $\braket{b'}{\rho(w)| b} = \overline{\braket{b}{\rho(w)| b'}}$ for any irreducible representation $\rho$ and normalized vectors $\ket{b},\ket{b'}$, so that we do not have to perform some calculations twice.

\paragraph{The principal series}\label{par: twiddle principal}
For the principal series, we can use \cref{eq: g action projective 1,eq: g action projective 2} to compute the transversal actions. On the initial basis, the action is given by
\begin{align*}
    &I_{\alpha, \beta} (w)\,\ket{0} = \ket{\infty}, \\
    &I_{\alpha, \beta} (w) \,\ket{\infty} = \ket{0}, \\
    &I_{\alpha, \beta} (w) \,\ket{x} = \chi_{\alpha}(x)\chi_{\beta}(-1/x)\ket{1/x} \quad \text{ for } x\in \F_q^*.
\end{align*}
We can then compute the inner product between basis elements to derive the explicit matrix. These derivations are done in \cref{app: twiddle calcs} and the resulting action for $I_{\alpha, \beta}$ with $\alpha\neq \beta$ is a sparse matrix given by the following.
\begin{itemize}
    \item The basis vectors $\{\ket{\infty}, \ket{v_0}, \ket{w_{\alpha}}, \ket{w_\beta}\}$ span an invariant subspace, where the action is given by
    \begin{align*}
        I_{\alpha,\beta}(w)\ket{\infty} &= \frac{1}{\sqrt{q}}\ket{v_0} + \frac{\sqrt{q-1}}{\sqrt{q}}\ket{w_{\beta}},\\
        I_{\alpha,\beta}(w) \ket{v_0} &=\frac{1}{\sqrt{q}}\ket{\infty} + \frac{\sqrt{q-1}}{q} \chi_{\beta}(-1)g_q(\alpha^{-1}\beta)\ket{w_\alpha},\\
        I_{\alpha,\beta}(w) \ket{w_\alpha} &=\frac{\sqrt{q-1}}{q} \chi_{\beta^{-1}}(-1)\overline{g_q(\alpha^{-1}\beta)} \ket{v_0} + \frac{1}{q} \chi_{\alpha^{-1}}(-1) g_q(1) g_q(\alpha\beta^{-1}) \ket{w_{\beta}}\\
        &=\chi_{\alpha^{-1}}(-1)g_q(\alpha\beta^{-1})\left(\frac{\sqrt{q-1}}{q} \ket{v_0} - \frac{1}{q} \ket{w_{\beta}}\right),\\
        I_{\alpha,\beta}(w) \ket{w_\beta} &= \frac{\sqrt{q-1}}{\sqrt{q}} \ket{\infty} - \frac{1}{q} \chi_{\beta^{-1}}(-1)g_q(\alpha^{-1}\beta) \ket{w_\alpha},
    \end{align*}
    where $g_q(\alpha) := G_q(\alpha,1)$ is a Gauss sum, see \cref{def: Gauss sum} in \cref{subsec: Gauss sums}.
    \item For the rest of the basis, consisting of $\{\ket{w_\delta}: \delta\notin\{\alpha,\beta\}\}$, the action by $I_{\alpha,\beta}$ is a phase-permutation matrix given by
    \begin{align*}
        I_{\alpha, \beta} (w) \ket{w_\delta}
        =\frac{1}{q}\chi_{\delta^{-1}}(-1)
        g_q(\delta\alpha^{-1})g_q(\delta\beta^{-1})
        \ket{w_{\alpha\beta\delta^{-1}}}.
    \end{align*}
\end{itemize}
Equivalently, for arbitrary $\delta,\epsilon\in X_q$,
\begin{equation}
 \langle w_\epsilon|I_{\alpha,\beta}(w)|w_\delta\rangle
 =\frac{1}{q}\chi_{\delta^{-1}}(-1)
   g_q(\delta\alpha^{-1})g_q(\delta\beta^{-1})
   \mathbf 1_{\{\epsilon=\alpha\beta\delta^{-1}\}}.
\end{equation}

\paragraph{Steinberg and determinant}\label{par: twiddle Steinberg}
The Steinberg and determinant representations are a little bit easier than the principal series.
 It holds that
\begin{align*}
    &I_{\alpha, \alpha} (w)\,\ket{0} = \ket{\infty}, \\
    &I_{\alpha, \alpha} (w) \,\ket{\infty} = \ket{0}, \\
    &I_{\alpha, \alpha} (w) \,\ket{x} = \chi_{\alpha}(-1)\ket{1/x},
\end{align*}
and indeed using the vector $\ket{v}$ from \cref{eq: v} we find that
\begin{equation*}
    I_{\alpha, \alpha}(w) \ket{v} = I_{\alpha, \alpha}(w) \frac{1}{\sqrt{q+1}} \left(\chi_{\alpha}(-1) \ket{\infty} + \sum_{x\in \F_q}\ket{x}\right) = \chi_{\alpha}(-1)\ket{v}
\end{equation*}
since $\chi_{\alpha}(-1)^2 = \chi_{\alpha}((-1)^2) = 1.$
On the GT basis for the Steinberg representation the action can be computed from the inner products in \cref{app: twiddle calcs}. We obtain
\begin{align*}
    I_{\alpha, \alpha}(w)\ket{v^\perp} &=-\frac{1}{q}\chi_{\alpha}(-1)\ket{v^\perp}+ \chi_{\alpha}(-1)\frac{\sqrt{q^2-1}}{q}\ket{w_{\alpha}},\\
    I_{\alpha, \alpha}(w)\ket{w_{\alpha}} &=\frac{\sqrt{q^2-1}}{q}\chi_{\alpha}(-1)\ket{v^\perp} + \frac{1}{q}\chi_{\alpha}(-1)\ket{w_{\alpha}},\\
    \St_{\alpha}(w)\ket{w_{\delta}}
    &=\frac{1}{q}\chi_{\delta^{-1}}(-1)
      g_q(\delta\alpha^{-1})^2
      \ket{w_{\alpha^2\delta^{-1}}}
      \quad \text{for all } \delta\neq\alpha.
\end{align*}
For arbitrary $\delta,\epsilon\in X_q$, the coefficient between the $\rho_{\alpha^2}$ paths is
\begin{equation}\label{eq: corrected Steinberg matrix coefficient}
 \langle w_\delta|\St_\alpha(w)|w_\epsilon\rangle
 =\frac{\chi_{\epsilon^{-1}}(-1)}q
   g_q(\epsilon\alpha^{-1})^2
   \mathbf 1_{\{\delta=\alpha^2\epsilon^{-1}\}}.
\end{equation}

\paragraph{Cuspidal}
Next, consider a cuspidal representation $\pi_\theta$ with $\theta|_{\F_q^*}=\chi_\gamma$.  Its Gelfand--Tsetlin basis is the multiplicative Fourier basis of $\C[\F_q^*]$.  Since $w=w'\psm{-1&0\\0&1}$, equations \cref{eq: cuspidal action,eq: cuspidal w prime} give
\begin{align}
    \braket{w_\epsilon}{\pi_\theta(w)|w_\delta}
    &=-\frac{\chi_\delta(-1)}{q}
      g_{q^2}\!\left(\theta\cdot(\chi_{\epsilon^{-1}}\circ\No)\right)
      \mathbf 1_{\{\epsilon=\delta^{-1}\gamma\}}.
    \label{eq: corrected cuspidal twiddle}
\end{align}
Thus $\pi_\theta(w)$ permutes the GT basis by $\ket{w_\delta}\mapsto\ket{w_{\delta^{-1}\gamma}}$, with the displayed Gauss-sum phase. The minus sign comes from the normalization of $j_\theta$ in \cref{eq: cuspidal w prime}.

\subsection{The QFT circuit}\label{sec: GL2 QFT circuit}

In this subsection we give the quantum circuits for the QFTs of both $B$ and $G$. We begin with a preparatory subsection needed for Gauss sum phases.  Then we give a circuit construction for $B$, after which we treat all the gates needed in the Mackey algorithm for $G$.

\subsubsection{Fourier transforms of \texorpdfstring{$\F_q$}{Fq} and Gauss sums}

To construct a quantum circuit for the QFT of $G$, we must use the additive and multiplicative descriptions of a field element interchangeably.  Fix a generator $g\in\F_q^*$.  We write $\ket{a}_+$ when $a\in\F_q$ is stored in a polynomial-basis field encoding, and write $\ket{a}_\times$ when a nonzero element $a=g^d$ is stored by its exponent $d\in\Z/(q-1)\Z$.
A regime flag distinguishes these encodings, but we suppress it from the notation.
 The state $\ket{0}_+$ is used as a zero element of the field and is never passed to a discrete-logarithm routine.

On nonzero basis states, define the representation-changing maps by
\begin{align}
    \exp_g\ket{g^d}_\times&=\ket{g^d}_+,
    \label{eq: exp operator}\\
    \log_g\ket{a}_+&=\ket{a}_\times.
    \label{eq: log operator}
\end{align}
They fix the zero and are extended to mutually inverse unitaries on the tagged work register.

\begin{lemma}\label{lem: add and mult of Fq}
    For every $\varepsilon>0$, there are uniform quantum circuits of size $\operatorname{poly}(\log q,\log(1/\varepsilon))$ which, with all ancillary registers returned to zero, implement to operator-norm error at most $\varepsilon$ the transforms
    \begin{align}
        \Fou_+\ket{x}_+
        &=
        \frac1{\sqrt q}\sum_{y\in\F_q}\psi_x(y)\ket{y}_+,
        \label{eq: efficient additive Fq QFT}\\
        \Fou_\times\ket{a}_\times
        &=
        \frac1{\sqrt{q-1}}\sum_{\alpha\in X_q}
        \chi_\alpha(a)\ket{\alpha}_\times,
        \qquad a\in\F_q^*,
        \label{eq: efficient multiplicative Fq QFT}
    \end{align}
    as well as the conversion maps \cref{eq: exp operator} and \cref{eq: log operator}.  The multiplicative transform is defined to fix $\ket{0}_+$.
\end{lemma}
\begin{proof}
    The additive transform is the QFT of the finite abelian group $(\F_q,+)$, with the trace pairing identifying $\F_q$ with its character group as in \cref{app: finite fields}, and is efficiently implementable by Ref.~\cite[Section~5]{Kitaev1995}.  The generator $g$ identifies $\F_q^*$ with the cyclic group $\Z/(q-1)\Z$, so the multiplicative transform is an efficient cyclic QFT as well.

    The exponentiation can be implemented efficiently \cite{beauregard2003quantum}, while an amplified version of Shor's discrete-logarithm algorithm \cite{Shor1994, MoscaZalka2004} implements $\log_g$ to error at most $\varepsilon$ with clean ancillas in polynomial time.
The regime flag that we suppress from notation handles the zero element of the field without using the discrete logarithm.
\end{proof}

The reason to switch between additive and multiplicative registers has to do with the Gauss sums appearing in \cref{sec: twiddle factors}.
Recall from \cref{def: Gauss sum} that a Gauss sum contains both an additive and multiplicative character.
A quantum algorithm for Gauss sum estimation is given in Ref.~\cite{vDS02} and was later reformulated in Ref.~\cite{Bruin}. This latter version lends itself well to our purpose, and we present it here in an altered version. The algorithm makes use of an ancilla and requires the following operator.
\begin{lemma}\label{lem: controlled character phase}
	The phase operator given by
    \begin{equation}
		U_{cp}\ket{\alpha}\ket{x}
        =\chi_\alpha(x)\ket{\alpha}\ket{x},
        \qquad \alpha\in X_q,\ x\in\F_q^*,
	\end{equation}
    and defined to be the identity when $\alpha=0$, can, for every $\varepsilon>0$, be implemented to operator-norm error at most $\varepsilon$ with size $\operatorname{poly}(\log q,\log(1/\varepsilon))$.
\end{lemma}
\begin{proof}

    Put $N=q-1$ and $\omega_N=e^{2\pi i/N}$.  In the exponent encodings, let $k,d\in\Z/N\Z$ be the stored labels of $\chi_\alpha = \chi^k$ and $x=g^d$, respectively, with the convention $\chi_\alpha(g^d)=\omega_N^{kd}$.  We apply the phase kickback trick.  Prepare a register
    \begin{equation*}
        \ket{\widetilde 1}
        =\frac1{\sqrt N}\sum_{s\in\Z/N\Z}\omega_N^{-s}\ket{s}
    \end{equation*}
    via inverse additive Fourier transform. Compute $t=kd\bmod N$ in an ancilla and add $t$ modulo $N$ to the $\ket{\widetilde{1}}$ register.  Since
    \begin{equation*}
        \operatorname{Add}_t\ket{\widetilde 1}
        =\omega_N^t\ket{\widetilde 1},
    \end{equation*}
    this multiplies the input by $\omega_N^{kd}=\chi_\alpha(x)$ without changing either input register.
    Uncomputing $t$ and unpreparing $\ket{\widetilde 1}$ returns all ancillas to zero. Note that the regime flag bypasses the entire construction on the zero element.
\end{proof}

\begin{theorem}[Gauss phase]\label{thm: coherent Gauss phase}
	\sloppy For every $\varepsilon>0$ there is a circuit over $\C[\F_q]$ of size $\operatorname{poly}(\log q,\log(1/\varepsilon))$ that is $\varepsilon$-close in operator norm to
	\begin{equation}\label{eq: coherent Gauss phase}
		\Gamma_q\ket{\alpha}=\begin{cases}
			\ket{0}, & \alpha=0,\\
			-\ket{1}, & \alpha=1,\\
			\dfrac{g_q(\alpha)}{\sqrt{q}}\ket{\alpha},
            & \alpha\in X_q\setminus\{1\}.
		\end{cases}
	\end{equation}
\end{theorem}

\begin{proof}
    First detect the zero using the regime flag.  On that branch perform no operation, so $\Gamma_q\ket{0}=\ket{0}$.  On the multiplicative branch, apply the phase $-1$ conditional on $\alpha=1$.
    For $\alpha\in X_q\setminus\{1\}$, use a work register in the multiplicative identity state and apply the following circuit.
    \begin{figure}[H]
\centering
\begin{quantikz}[column sep=0.32cm]
  \lstick{$\ket{\alpha}$}
  & \qw
  & \gate[2]{U_{cp}}
  & \qw
  & \qw
  & \qw
  & \gate[2]{U_{cp}}
  & \qw
  & \rstick{$\dfrac{g_q(\alpha)}{\sqrt q}\ket{\alpha}$}\qw\\
  \lstick{$\ket{1}_\times$}
  & \gate{\Fou_\times}
  &
  & \gate{\exp_g}
  & \gate{\Fou_+}
  & \gate{\log_g}
  &
  & \gate{\Fou_\times^\dagger}
  & \rstick{$\ket{1}_\times$}\qw\\
\end{quantikz}
\caption{Circuit implementing the Gauss phase on the nontrivial multiplicative-character branch.
The work register is returned to $\ket{1}_\times$.}
\label{fig: GL2 coherent Gauss phase circuit}
\end{figure}
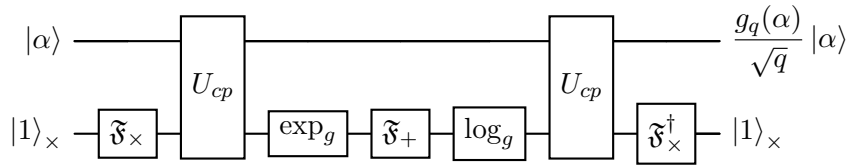

    Each $U_{cp}$ acts jointly on the displayed registers and only changes the phase.
    After the first controlled phase, exponentiation, and the additive Fourier transform, the state is
    \begin{equation*}
      \frac1{\sqrt{q(q-1)}}
      \sum_{y\in\F_q}
      \left(\sum_{x\in\F_q^*}\chi_\alpha(x)\psi(xy)\right)|y\rangle_+.
    \end{equation*}
    The coefficient at $y=0$ vanishes.  For $y\neq0$ the substitution $z=xy$ gives
    \begin{equation*}
      \sum_{x\in\F_q^*}\chi_\alpha(x)\psi(xy)
      =\chi_{\alpha^{-1}}(y)g_q(\alpha).
    \end{equation*}
    Hence the state can be rewritten to
    \begin{equation*}
      \frac{g_q(\alpha)}{\sqrt q}\,
      \frac1{\sqrt{q-1}}
      \sum_{y\in\F_q^*}\chi_{\alpha^{-1}}(y)|y\rangle_+.
    \end{equation*}
    The logarithm and second controlled phase cancel $\chi_{\alpha^{-1}}(y)$, and $\Fou_\times^\dagger$ returns the work register to its initial state.  Since $|g_q(\alpha)|=\sqrt q$, the remaining factor is a phase.  Every gate in this sequence is controlled on the nonzero regime flag, so the zero sentinel is unchanged and is never passed to $\log_g$ or $\exp_g$.  The error bound follows from \cref{lem: add and mult of Fq} and the fact we use a constant number of operators.
\end{proof}

\subsubsection{Circuit for the QFT of \texorpdfstring{$B$}{B}}
A circuit for the QFT of $B$ uses three input registers for $a,d\in \F_q^*$ and $b\in \F_q$. Together they encode the group element $g(a,d,b) = \psm{1&b\\0&1}\psm{a&0\\0&d}$.  For labeling the representations of $B$ we use $\C[\F_q]^{\otimes3}$, where $\ket{0, \alpha,\beta}$ labels the linear representation $\chi_{\alpha,\beta}$ and $\ket{\gamma, 0,0}$ labels $\rho_\gamma$. The circuit consists of the following steps.
\begin{enumerate}
	\item Start with an input state $\ket{a,d}\ket{b}$ and clean register $\ket{0,0}\ket{0,0}$. Assume that $\ket{b} = \ket{b}_+$ is encoded additively. The total starting state is then
	\begin{equation*}
		\ket{a,d}\ket{b,0,0}\ket{0,0},
	\end{equation*}
	which will be transformed into the Fourier basis with basis states of the form $\ket{\chi^{\vee}}\ket{\lambda}\ket{\tilde{\chi}}$.
	\item Perform the QFT over $T$ on the first two registers.
	\item Apply $\Fou_+$ over the $b$-register, followed by $\log_g$ and $\Fou_{\times}$. Ignoring the clean register in notation, we obtain the state
	\begin{align*}
		&\frac{1}{(q-1)\sqrt{q}}\sum_{\alpha, \beta\in X_q} \chi_{\alpha}(a)\chi_{\beta}(d)\left(\ket{\alpha,\beta}\ket{0}_+ + \frac{1}{\sqrt{q-1}}\sum_{\substack{t\in \F_q^*\\ \eta\in \F_q^*}}\psi(bt)\chi_{\eta}(t) \ket{\alpha, \beta}\ket{\eta}_{\times}\right)\\
		=&\frac{1}{(q-1)\sqrt{q}}\sum_{\alpha, \beta\in X_q} \chi_{\alpha}(a)\chi_{\beta}(d)\left(\ket{\alpha,\beta}\ket{0}_+ + \frac{1}{\sqrt{q-1}}\sum_{\substack{t\in \F_q^*\\ \delta\in \F_q^*}}\psi(bt)\chi_{\alpha\delta^{-1}}(t) \ket{\alpha, \beta}\ket{\alpha\delta^{-1}}_{\times}\right).
	\end{align*}

	\item Lastly, we need to use arithmetic over finite fields to reorder all the information into the correct registers. Specifically, we apply the following transformations
	\begin{align}
		\label{eq: Ub transv first} \UBtransv: &\ket{\alpha, \beta}\ket{0,0,0} \ket{0,0}\mapsto \ket{(\alpha, \beta)^\vee}\ket{0,\alpha,\beta} \ket{\alpha,\beta}\\
		&\label{eq: Ub transv sec}\ket{\alpha,\beta}\ket{ \alpha\delta^{-1},0,0}\ket{0,0} \mapsto \ket{(\alpha,\beta)^\vee}\ket{\alpha\beta, 0, 0}\ket{\delta, \delta^{-1}\alpha\beta}
	\end{align}
\end{enumerate}
The quantum circuit looks as follows.
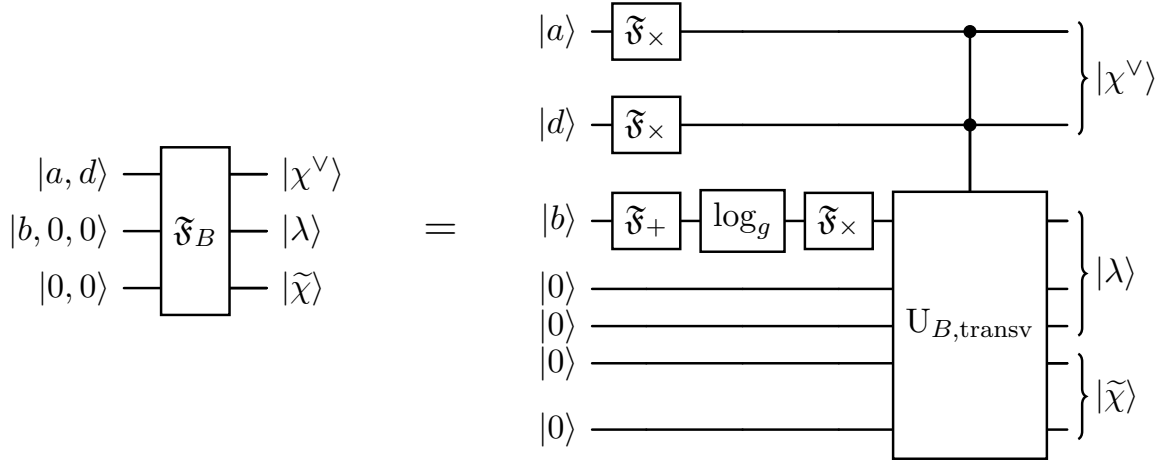
\begin{figure}[H]
\centering
\resizebox{\textwidth}{!}{\begin{tabular}{c@{\qquad}c@{\qquad}c}
\begin{quantikz}[row sep=0.34cm,column sep=0.42cm]
\lstick{$\ket{a,d}$}
 & \gate[3]{\Fou_B}
 & \rstick{$\ket{\chi^{{\vee}}}$}\qw
\\
\lstick{$\ket{b,0,0}$}
 &
 & \rstick{$\ket{\lambda}$}\qw
\\
\lstick{$\ket{0,0}$}
 &
 & \rstick{$\ket{\widetilde\chi}$}\qw
\end{quantikz}
&
\Large $=$
&
\begin{quantikz}[row sep=0.42cm,column sep=0.23cm]
\lstick{$\ket{a}$}
 & \gate{\Fou_\times}
 &
 &
 & \ctrl{2}
 & \rstick[2]{$\ket{\chi^{{\vee}}}$}\qw
\\
\lstick{$\ket{d}$}
 & \gate{\Fou_\times}
 &
 &
 & \ctrl{1}
 &
\\
\lstick{$\ket{b}$}
 & \gate{\Fou_+}
 & \gate{\log_g}
 & \gate{\Fou_\times}
 & \gate[5]{\UBtransv}
 & \rstick[3]{$\ket{\lambda}$}\qw
\\
\lstick{$\ket{0}$}
 &
 &
 &
 &
 &
\\
\lstick{$\ket{0}$}
 &
 &
 &
 &
 &
\\
\lstick{$\ket{0}$}
 &
 &
 &
 &
 & \rstick[2]{$\ket{\widetilde\chi}$}\qw
\\
\lstick{$\ket{0}$}
 &
 &
 &
 &
 &
\end{quantikz}
\end{tabular}}
\caption{The Fourier transform $\Fou_B$ with the input and Fourier-basis output registers used in the text, together with its decomposition into finite-field Fourier transforms and the reversible relabeling $\UBtransv$.}
\label{fig: GL2 B QFT decomposition}
\end{figure}

\begin{remark}
	The first two registers undergo the Fourier transform of $T \cong (\F_q^*)^2$.
\end{remark}

\begin{lemma}
	The operator $\UBtransv$ can be efficiently implemented.
\end{lemma}
\begin{proof}
	Recall that the register for $\F_q$ actually contains three parts. One for both its additive and multiplicative structure, and a regime flag.
	When the flag is $0$, it is in the additive structure and can be ignored by $\UBtransv$.
The operation we then need to perform is to duplicate $\ket{\alpha, \beta}$ into a clean register as in \cref{eq: Ub transv first}.

	When the flag is 1, we have performed $\Fou_\times$ and are in the multiplicative part of the register. We can therefore view $\UBtransv$ as an operator on $\C[\F_q^*]^{\otimes 5}$, leaving the additive part of $\C[\F_q]$ alone and also ignoring the other two parts of the middle registers. Write $\alpha = g^a, \beta = g^b, \delta = g^d$ and using the identification $\F_q^* \cong \Z/(q-1)\Z$ we need to construct an operator $\UBtransv'$ that maps
	\begin{equation*}\ket{a, b} \ket{a-d,0,0}\ket{0,0}
    \mapsto \ket{(a,b)^\vee}\ket{a+b,0,0}\ket{d,a+b-d}.\end{equation*}
	This is efficiently implementable by Ref.~\cite{beauregard2003quantum}.
\end{proof}

\begin{theorem}
\label{thm: efficient B QFT}
	For every $\varepsilon>0$, there is a quantum circuit of size $\operatorname{poly}(\log q,\log(1/\varepsilon))$ that implements the Fourier transform $\Fou_B$ of
	\begin{equation}
		B=\left\{\psm{a&b\\0&d}:a,d\in\F_q^*,\ b\in\F_q\right\}
	\end{equation}
	in the representation bases fixed above, to operator-norm error at most $\varepsilon$.
\end{theorem}
\begin{proof}
	The calculation above and \cref{fig: GL2 B QFT decomposition} express $\Fou_B$ as two multiplicative Fourier transforms implementing $\Fou_T$, followed by one additive Fourier transform, the conversion $\log_g$, a third multiplicative Fourier transform, and the relabeling $\UBtransv$.  By \cref{lem: add and mult of Fq}, each Fourier transform and $\log_g$ has a uniform circuit of size $\operatorname{poly}(\log q,\log(1/\varepsilon))$.  The preceding lemma provides an efficient implementation of $\UBtransv$.
\end{proof}

\subsubsection{The untwisting map}

For the Mackey circuit associated with $B < G$, the Bruhat decomposition gives the two double-coset representatives $e$ and $w$, with
\begin{equation*}
	H_e=B,\qquad H_w=B\cap wBw^{-1}=T.
\end{equation*}
The identity representative gives $\U_{\tw}^e=I$.  For the non-trivial cell, $w^{-1}Tw=T$ with
\begin{equation*}
	w^{-1}\psm{a&0\\0&d}w=\psm{d&0\\0&a},
\end{equation*}
so we are in the normalizer situation of \cref{subsec: omega normalizes}.
Twisting therefore does not require a new irreducible decomposition.
It only permutes the characters of $T$ by
	\begin{equation*}
	\chi_{\alpha,\beta}^w=\chi_{\beta,\alpha}.
	\end{equation*}

In particular, using $\Res_T^B\rho_\gamma=\bigoplus_{\delta\in X_q} \chi_{\delta,\delta^{-1}\gamma}$ from \cref{lem: decomposition of restrictions}, the $\delta$-summand is sent to the $\delta^{-1}\gamma$-summand.  With the one-dimensional block intertwiners chosen to be the identity, the untwisting operator is the permutation
\begin{equation}
	\U_{\tw}^w\ket{w_\delta}
	=\ket{w_{\delta^{-1}\gamma}}.
	\label{eq: GL2 untwisting convention}
\end{equation}
Equivalently, setting $\delta=\alpha^{-1}\gamma$ gives $\U_{\tw}^w\ket{w_{\alpha^{-1}\gamma}} =\ket{w_{\alpha}}$.
Thus $\U_{\tw}$ is implemented exactly by a cell-controlled swap of the two character coordinates, or by the reversible character operation $\delta\mapsto\delta^{-1}\gamma$ on a $\rho_\gamma$-path.  Its circuit size is polynomial in $\log q$.

\subsubsection{Circuit for the \texorpdfstring{$A$}{A}-matrix}

The hardest part of the circuit to implement is the $A$-matrix.  We use the double-coset and untwisting conventions from the previous subsection.
Fix a total order $\prec$ on $X_q$ and represent a principal series by $I_{\alpha,\beta}$ with $\alpha\prec\beta$.
For ease of notation we also define the Gauss sum phase
\begin{equation}\label{eq: mixed A phase}
    \zeta_{\alpha,\beta}
    :=\chi_\beta(-1)\frac{g_q(\alpha^{-1}\beta)}{\sqrt q}.
\end{equation}
Recall from \cref{def: Mackey operators} that the $A$-matrix is block diagonal, depending on the inputs $\lambda$ and $\widetilde{\lambda}$ and given by
\begin{equation*}
A=\sum_{\lambda,\widetilde\lambda\in\widehat B}
\ket{\lambda}\!\bra{\lambda}\otimes
A_{\lambda,\widetilde\lambda}\otimes
\ket{\widetilde\lambda}\!\bra{\widetilde\lambda}.
\end{equation*}
For most pairs $(\lambda, \widetilde\lambda)$ the multiplicity space is one- or two-dimensional.
Each display below specifies one block $A_{\lambda, \widetilde\lambda}$. Its source basis is labeled $\ket{\kappa, \omega}$, where $\kappa$ is an intermediate character and $\omega$ a double coset representative. Note that the same $\ket{\kappa, \omega}$ can appear in several blocks. For the linear representations the blocks are
\begin{align}
 &A_{\chi_{\alpha,\beta}, \chi_{\alpha, \beta}}|\chi_{\alpha,\beta},e\rangle
    \longmapsto |I_{\alpha,\beta}\rangle,
    &&\alpha\prec\beta,\\
 &A_{\chi_{\alpha,\beta}, \chi_{\beta, \alpha}}|\chi_{\beta,\alpha},w\rangle
    \longmapsto |I_{\alpha,\beta}\rangle,
    &&\alpha\prec\beta,\\
 &A_{\chi_{\alpha,\alpha}, \chi_{\alpha, \alpha}}|\chi_{\alpha,\alpha},e\rangle
    \longmapsto
      \sqrt{\frac q{q+1}}|\St_\alpha\rangle
      +\frac1{\sqrt{q+1}}|\det_\alpha\rangle,
   \label{eq: Steinberg-det 1} \\
 &A_{\chi_{\alpha,\alpha}, \chi_{\alpha, \alpha}}|\chi_{\alpha,\alpha},w\rangle
    \longmapsto
      \chi_\alpha(-1)\left(
      -\frac1{\sqrt{q+1}}|\St_\alpha\rangle
      +\sqrt{\frac q{q+1}}|\det_\alpha\rangle\right).
      \label{eq: Steinberg-det 2}
\end{align}
For $\alpha\prec\beta$, the blocks that mix one- and multidimensional representations are
\begin{align}
 A_{\chi_{\alpha,\beta},\rho_{\alpha\beta}}:
 |\chi_{\beta,\alpha},w\rangle
    &\longmapsto |I_{\alpha,\beta}\rangle,\\
 A_{\chi_{\beta,\alpha},\rho_{\alpha\beta}}:
 |\chi_{\alpha,\beta},w\rangle
    &\longmapsto \zeta_{\alpha,\beta}|I_{\alpha,\beta}\rangle,\\
 A_{\rho_{\alpha\beta},\chi_{\alpha,\beta}}:
 |\chi_{\alpha,\beta},w\rangle
    &\longmapsto |I_{\alpha,\beta}\rangle,\\
 A_{\rho_{\alpha\beta},\chi_{\beta,\alpha}}:
 |\chi_{\beta,\alpha},w\rangle
    &\longmapsto \overline{\zeta_{\alpha,\beta}}
       |I_{\alpha,\beta}\rangle.
\end{align}
The cases where $\alpha=\beta$ are
\begin{align}
 A_{\chi_{\alpha,\alpha},\rho_{\alpha^2}}:
 |\chi_{\alpha,\alpha},w\rangle
    &\longmapsto\chi_\alpha(-1)|\St_\alpha\rangle,\\
 A_{\rho_{\alpha^2},\chi_{\alpha,\alpha}}:
 |\chi_{\alpha,\alpha},w\rangle
    &\longmapsto\chi_\alpha(-1)|\St_\alpha\rangle.
\end{align}
It remains to describe $A_{\rho_\gamma,\rho_\gamma}$.  Its rows are indexed by
\begin{align}
 \mathcal C_\gamma
   &=\{\{\theta,\theta^q\}:\theta^q\neq\theta,
       \ \theta|_{\F_q^*}=\chi_\gamma\},\\
 \mathcal P_\gamma
   &=\{\{\delta,\eta\}:\delta\eta=\gamma,\ \delta\neq\eta\},\\
 \mathcal S_\gamma
   &=\{\delta:\delta^2=\gamma\},
\end{align}
corresponding respectively to paths via $\pi_\theta$, $I_{\delta,\eta}$, and $\St_\delta$.  If $s_\gamma=|\mathcal S_\gamma|$, then
\begin{equation*}
 |\mathcal C_\gamma|=\frac{q+1-s_\gamma}{2},\qquad
 |\mathcal P_\gamma|=\frac{q-1-s_\gamma}{2},\qquad
 |\mathcal S_\gamma|=s_\gamma,
\end{equation*}
so these sets contain exactly $q$ labels in total.  Write the input states as
\begin{equation}
\label{eq: A_gamma col indexing}
\ket{c_\infty}:=\ket{\rho_\gamma,e},
 \qquad
 \ket{c_\alpha}:=\ket{\chi_{\alpha,\alpha^{-1}\gamma},w},
 \quad \alpha\in X_q.
\end{equation}
Using the generator $g$ fixed above, identify $X_q$ with $\F_q^*$ by matching exponent labels and regard this source basis as one $\F_q$-register via $\ket{c_\infty}\leftrightarrow\ket{0}$ and $\ket{c_\alpha}\leftrightarrow\ket{\alpha}$.  The first identification is canonical and the nonzero identification uses the fixed generator.
The columns of $A_{\rho_\gamma, \rho_\gamma}$ are then given by its action on the input states.
The first column (corresponding to the $\ket{0}$-input) is
\begin{align}
 A_{\rho_\gamma,\rho_\gamma}\ket{c_\infty}
  ={}&\frac1{\sqrt{q+1}}
       \sum_{\{\theta,\theta^q\}\in\mathcal C_\gamma}|\pi_\theta\rangle
   +\frac1{\sqrt{q-1}}
       \sum_{\{\delta,\eta\}\in\mathcal P_\gamma}|I_{\delta,\eta}\rangle
       \nonumber\\
   &+\sqrt{\frac q{(q-1)(q+1)}}
       \sum_{\delta\in\mathcal S_\gamma}|\St_\delta\rangle.
 \label{eq: corrected A identity column}
\end{align}
For the other columns, put $b=\alpha^{-1}\gamma$. Then the $A$-matrix acts as
\begin{align}
 A_{\rho_\gamma,\rho_\gamma}\ket{c_\alpha}
 ={}&-\frac{\chi_b(-1)}{\sqrt{q(q^2-1)}}
       \sum_{\{\theta,\theta^q\}\in\mathcal C_\gamma}
       g_{q^2}\!\left(\theta\cdot(\chi_{\alpha^{-1}}\circ\No)\right)
       |\pi_\theta\rangle
       \nonumber\\
 &+\frac{\chi_b(-1)}{(q-1)\sqrt q}
       \sum_{\{\delta,\eta\}\in\mathcal P_\gamma}
       g_q(\alpha^{-1}\delta)g_q(\alpha^{-1}\eta)
       |I_{\delta,\eta}\rangle
       \nonumber\\
 &+\frac{\chi_b(-1)}{(q-1)\sqrt{q+1}}
       \sum_{\delta\in\mathcal S_\gamma}
       g_q(\alpha^{-1}\delta)^2|\St_\delta\rangle.
 \label{eq: corrected A nontrivial column}
\end{align}

The first sum is independent of the chosen representative of the orbit $\{\theta,\theta^q\}$ because the norm character is invariant under Frobenius, and the substitution $z\mapsto z^q$ leaves the $\F_{q^2}$ Gauss sum unchanged.
These coefficients follow from \cref{eq: corrected general A coefficient} and the transversal actions computed above. Note that we don't have to check unitarity as this also follows from \cref{thm: Mackey induction transform circuit}.

All blocks other than $A_{\rho_\gamma,\rho_\gamma}$ have dimension at most two and are implementable in $\operatorname{poly}(\log q)$ gates using reversible character arithmetic, one controlled rotation, and the Gauss phase of \cref{thm: coherent Gauss phase}.
The hardest part is given by the $q\times q$-unitary from \cref{eq: corrected A identity column,eq: corrected A nontrivial column}.  Treating this unitary as an unstructured matrix and synthesizing it by Givens rotations would cost $O(q^2\operatorname{poly}(\log q,\log(1/\varepsilon)))$ gates for each fixed $\gamma$ \cite[Section~4.5]{NielsenChuang}.  A more efficient way to construct these bigger blocks is treated in the next subsection. As a result we obtain the following.

\begin{theorem}[Efficient implementation of the total $A$-matrix]
\label{thm: efficient GL2 total A matrix}
	For every $\varepsilon>0$, the total $A$-matrix in the Mackey circuit for $B < \GL_2(\F_q)$ has an implementation to operator-norm error at most $\varepsilon$ using $\operatorname{poly}(\log q,\log(1/\varepsilon))$ gates.
\end{theorem}
\begin{proof}
	Recall from \cref{def: Mackey operators} that the total operator is the controlled direct sum
	\begin{equation*}
		A=\sum_{\lambda,\widetilde\lambda\in\widehat B}
		\ket{\lambda}\!\bra{\lambda}\otimes
		A_{\lambda,\widetilde\lambda}\otimes
		\ket{\widetilde\lambda}\!\bra{\widetilde\lambda}.
	\end{equation*}
	The registers entering $A$ also contain the double-cell label $\omega$ and the intermediate character label $\kappa$.  From $(\lambda,\widetilde\lambda,\kappa,\omega)$ we reversibly compute a control indicating which of the block families displayed above applies.  The required tests consist only of distinguishing a one-dimensional label $\chi_{\alpha,\beta}$ from a label $\rho_\gamma$, comparing character labels, and computing their products and inverses. This gives a fixed size of the control register and the checks therefore use finite fields arithmetic in $\operatorname{poly}(\log q)$ gates.

	On the one- and two-dimensional families, the formulas above are implemented by  relabeling, the phases $\chi_\alpha(-1)$ and $\zeta_{\alpha,\beta}$, and the displayed two-dimensional determinant--Steinberg rotation from \cref{eq: Steinberg-det 1,eq: Steinberg-det 2}.  Character arithmetic is reversible, the Gauss-sum factor in $\zeta_{\alpha,\beta}$ follows from \cref{thm: coherent Gauss phase}, and the rotation can be synthesized with $\operatorname{poly}(\log q,\log(1/\varepsilon))$ gates.
	Assigning a sufficiently small constant fraction of $\varepsilon$ to each approximate unitary makes every such block $\varepsilon$-accurate.
	The only remaining family is $A_{\gamma}:=A_{\rho_\gamma,\rho_\gamma}$.  By \cref{thm: efficient GL2 hard A circuit} below, a single circuit implements this $q\times q$ block with $\gamma$ as a control in the label register.  Thus the construction does not iterate over the $q-1$ possible values of $\gamma$.

	Apply the circuits for these constantly many families controlled by the block tag and then uncompute the tag via finite field arithmetic.
\end{proof}

\phantomsection
\label{proof: GL2 QFT}
\begin{proof}[Proof of \Cref{thm: GL2 QFT}]
	Put $G=\GL_2(\F_q)$ and $H=B$.  The subgroup induction method of \cref{lemma: induction step of first induction} gives us that
	\begin{equation}
		\Fou_G
		=\IndTrans\bigl(I_{G/B}\otimes\Fou_B\bigr)\Enc_{G/B}.
		\label{eq: GL2 final QFT factorization}
	\end{equation}
	The subgroup transform $\Fou_B$ has the required uniform efficient implementation by \cref{thm: efficient B QFT}.

	We first check the remaining classical encodings.  For $g=\psm{a&b\\c&d}\in G$, the transversal in \cref{eq: transversal set G/B} gives
	\begin{equation*}
		g=\begin{cases}
			t_\infty g,&c=0,\\[1mm]
			t_{a/c}\psm{c&d\\0&-\det(g)/c},&c\ne0.
		\end{cases}
	\end{equation*}
	Thus, $\Enc_{G/B}$ is computed reversibly using a constant number of finite-field operations.  It has size $\operatorname{poly}(\log q)$ and cleans its work registers.

	It remains to implement $\IndTrans$, for which we use the Mackey circuit in \cref{alg: Mackey induction transform}.
	For $q=2$ we can use the constant-sized rotation in \cref{eq: GL2 q2 hard A block} and from now on we assume that $q>2$.
	For $B < G$, the double cosets have representatives $e$ and $w$, with $H_e=B$ and $H_w=T$.  By \cref{thm: Mackey induction transform circuit}, $\IndTrans$ is the product of the Mackey encoding, the untwisting maps, the smaller induction transforms, and the total $A$-matrix.  Each of these components is efficient for our choice of groups.
	\begin{itemize}
		\item The Mackey encoding uses $t_\infty=e$ and $t_x=u_xw$ for $x\in\F_q$, so it is a relabeling using $O(\log q)$ qubits.
		\item On the identity cell the untwisting and smaller induction transforms are trivial.  On the $w$-cell, the untwisting convention in \cref{eq: GL2 untwisting convention} is reversible character arithmetic, and the induction transform from $T$ to $B$ is part of the circuit in \cref{fig: GL2 B QFT decomposition} (the part after $\Fou_T$).
		\item The total $A$-matrix has an implementation of size $\operatorname{poly}(\log q,\log(1/\varepsilon))$ by \cref{thm: efficient GL2 total A matrix}.
	\end{itemize}
	The conclusion is that the Mackey circuit implements $\IndTrans$ with size $\operatorname{poly}(\log q,\log(1/\varepsilon))$, where the error can be split across the components.
\end{proof}

\subsection{The \texorpdfstring{$A_\gamma$}{A-gamma}-block}
\label{subsec: polylog GL2 A}

The $A_{\gamma}$ blocks are the only blocks in the $A$-matrix whose dimension grows with $q$.  For fixed $\gamma$, the input basis is $\ket{c_\infty},\ket{c_\alpha}$, $\alpha\in X_q$, and the target basis is indexed by $\mathcal C_\gamma,\mathcal P_\gamma,\mathcal S_\gamma$, as above.  We will show that there exists an efficient circuit producing the matrix columns in \cref{eq: corrected A identity column,eq: corrected A nontrivial column}. Recall that the columns of $A_\gamma$ are indexed by the input states in \cref{eq: A_gamma col indexing}, while its rows are indexed by the irreps of $G$.

\begin{remark}
When $q=2$ it must hold that $\gamma=1$.  The set $\mathcal C_1$ consists of the unique nontrivial Frobenius orbit $\{\theta,\theta^2\}\subset X_4$, while $\mathcal P_1=\varnothing$ and $\mathcal S_1=\{1\}$.  Moreover, $g_2(1)=-1$ and, with the convention of \cref{def: Gauss sum}, $g_4(\theta)=2$. Therefore \cref{eq: corrected A identity column,eq: corrected A nontrivial column} give, in the input basis $(\ket{c_\infty},\ket{c_1})$ and target basis $(\ket{\pi_\theta},\ket{\St_1})$,
\begin{equation}
	A_1=
	\begin{pmatrix}
		1/\sqrt3&-\sqrt{2/3}\\
		\sqrt{2/3}&1/\sqrt3
	\end{pmatrix}.
	\label{eq: GL2 q2 hard A block}
\end{equation}
This constant-size rotation can be synthesized directly.
Hence for the remainder of this subsection we may assume that $q>2$.
\end{remark}

Looking at \cref{eq: corrected A identity column,eq: corrected A nontrivial column}, for each input state we need to create a superposition over the same three sums, where the normalization is slightly different for the special input $\ket{c_{\infty}}$. Beyond creating this superposition, we also need to add local phases that depend on the input. The total circuit will look as follows and we briefly explain the role of each gate.
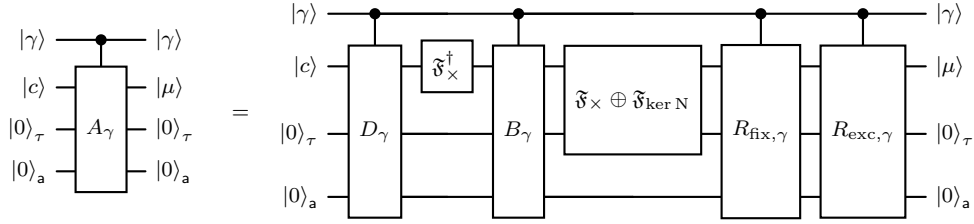
\begin{figure}[H]
\centering
{\scriptsize
\begin{quantikz}[row sep=0.28cm,column sep=0.25cm]
\lstick{$\ket\gamma$}
 & \ctrl{1}
 & \rstick{$\ket\gamma$}
\\
\lstick{$\ket c$}
 & \gate[3]{A_\gamma}
 & \rstick{$\ket\mu$}
\\
\lstick{$\ket0_\tau$}
 &
 & \rstick{$\ket0_\tau$}
\\
\lstick{$\ket0_{\mathsf a}$}
 &
 & \rstick{$\ket0_{\mathsf a}$}
\end{quantikz}
\quad=\quad
\begin{quantikz}[row sep=0.28cm,column sep=0.27cm]
\lstick{$\ket\gamma$}
 & \ctrl{1}
 & \qw
 & \ctrl{1}
 & \qw
 & \ctrl{1}
 & \ctrl{1}
 & \rstick{$\ket\gamma$}
\\
\lstick{$\ket c$}
 & \gate[3]{D_\gamma}
 & \gate{\Fou_\times^\dagger}
 & \gate[3]{B_\gamma}
 & \gate[2]{\Fou_\times\oplus\Fou_{\ker\No}}
 & \gate[3]{R_{\mathrm{fix},\gamma}}
 & \gate[3]{R_{\mathrm{exc},\gamma}}
 & \rstick{$\ket\mu$}
\\
\lstick{$\ket0_\tau$}
 &
 & \qw
 &
 &
 &
 &
 & \rstick{$\ket0_\tau$}
\\
\lstick{$\ket0_{\mathsf a}$}
 &
 & \qw
 &
 & \qw
 &
 &
 & \rstick{$\ket0_{\mathsf a}$}
\end{quantikz}
}
\caption{Circuit for $A_\gamma$ with its ancillary registers displayed.
The register $\tau$ stores the split/nonsplit branch between $B_\gamma$ and $R_{\mathrm{fix},\gamma}$.
The ancilla register $\mathsf a=(\mathsf a_D,\mathsf a_B,\mathsf a_{\mathrm{fix}}, \mathsf a_{\mathrm{exc}})$ contains the character and Gauss-phase work, the discriminant or trace tests and roots, the orbit-partner data, and the exceptional-column preparation work.
Each part of the circuit returns the part of $\mathsf a$ that it uses to zero, so the ancilla can be reused.}
\label{fig: GL2 hard A circuit}
\end{figure}

\begin{itemize}
	\item The gate $D_\gamma$ adds a Gauss-sum phase. This is a diagonal gate and does not work toward creating the correct superposition.
	\item A Fourier transform changes the input character label into a sum of field elements $t$.
	\item This sum over $t$'s we will divide into two routes by looking at the polynomial $X^2-X+t$. If the polynomial splits over $\F_q$, we call it the split case/branch and these will become the principal series. If the polynomial is irreducible, we call it non-split and they contribute to cuspidal representations. The gate $B_{\gamma}$ sends the $t$'s to their correct groups.
	\item After the branching, we apply a Fourier transform on both branches to help create the correct superposition and phases.
	\item On each branch, i.e. split and non-split, labels need to be combined into the correct representation. The idea of this step is that the principal series $I_{\alpha, \beta}$ and $I_{\beta, \alpha}$ are the same representation and thus need the same label. Similarly, $\pi_{\theta}$ and $\pi_{\theta^q}$ need the same label. The gate $R_{\mathrm{fix}, \gamma}$ fixes these labels.
	\item Lastly, the gate $R_{\mathrm{exc}, \gamma}$ deals with the exceptional (or edge) cases, i.e. when $\alpha = \beta$ for the principal series, to ensure these end up with the correct phases and normalization. Note that for even characteristic, the edge cases will be slightly different as no repeated root can occur.
\end{itemize}

The branching gate maps a $q$-dimensional input space into a larger labeled space.  We therefore use extra registers and extend its action to a unitary on the full register.  All equalities below are for inputs with these extra registers in the zero state.
In the next subsections we give formal definitions of all the gates, show that they form a circuit for $A_\gamma$, and prove that each gate is efficiently implementable.

\subsubsection{Definition of the gates}

For the rest of the subsection, when the input is $\ket{c_\alpha}$, write $b=\alpha^{-1}\gamma$ as in \cref{eq: corrected A nontrivial column}.

\begin{definition}[The phase gate]
	Define
	\begin{align}
		D_\gamma\ket{c_\infty}&=\ket{c_\infty},\\
		D_\gamma\ket{c_\alpha}
		&=\begin{cases}
			\displaystyle
			\zeta_{\alpha,b}\ket{c_\alpha}
			&\alpha\ne b,\\[2mm]
			\ket{c_\alpha}&\alpha=b,
		\end{cases}
		\label{eq: GL2 generic diagonal}
	\end{align}
	where $\zeta_{\alpha,\beta}$ is defined in \cref{eq: mixed A phase}.
\end{definition}

If $\alpha\ne b$, then $\alpha^{-1}b$ is nontrivial and the coefficient in \cref{eq: GL2 generic diagonal} has absolute value one.  Thus $D_\gamma$ is unitary.

\begin{definition}
The gate $\Fou_\times^\dagger$ is the inverse Fourier transform, i.e. the adjoint of \cref{eq: efficient multiplicative Fq QFT}. It fixes $\ket{c_\infty}=\ket0$.
\end{definition}

To define the branching gate, we first fix the character coordinates used on the non-split branch.  Let $h$ generate $\F_{q^2}^*$ such that $g=h^{q+1}$ is the generator of $\F_q^*$ fixed at the start of the section.
The map $z\mapsto z/z^q$ gives an identification $\F_{q^2}^*/\F_q^*\simeq\ker\No.$ If $\chi_\gamma(g)=e^{2\pi i k/(q-1)}$ for some $0\leq k\leq q-2$, choose $\theta_0$ by
\begin{equation*}
 \theta_0(h)=e^{2\pi i k/(q^2-1)}
\end{equation*}
so that $\theta_0|_{\F_q^*}=\chi_\gamma$.  Every character of $\F_{q^2}^*$ restricting to $\chi_\gamma$ is then uniquely of the form $\theta_0\nu$ for some $\nu\in\widehat{\ker\No}$ pulled back along $z\mapsto z/z^q$.  Using these coordinates, Frobenius acts by
\begin{equation}
 (\theta_0\nu)^q
 =\theta_0\bigl[(\theta_0^q\theta_0^{-1})\nu^{-1}\bigr].
 \label{eq: GL2 affine Frobenius coordinate}
\end{equation}

\begin{definition}[The quadratic branching gate] Define the branching map $B_{\gamma}: \C[\F_q] \to \C^2\otimes \C[\F_{q^2}^*]$, given by
\begin{equation} \label{eq: GL2 split router}
 B_\gamma\ket t
 =\frac{1}{\sqrt{2}}\cdot \begin{cases}\chi_\gamma(y)\ket{\mathrm{sp},x/y}+\chi_\gamma(x)\ket{\mathrm{sp},y/x} &
 \substack{\text{ if $ X^2-X+t$ is reducible over $\F_q$} \\ \text{with roots $(x, y)$ and $ t\neq 0$} \\ \text{ ($t\neq 1/4$ for $\operatorname{char}(\F_q)$ odd)}}\\
 -\theta_0(z)\ket{\mathrm{ns},z/z^q}-\theta_0(z^q)\ket{\mathrm{ns},z^q/z} & \substack{\text{ if $ X^2-X+t$ is irreducible over $\F_q$} \\ \text{with roots $(z, z^q)$}},
 \end{cases}
\end{equation}
where $\ket{\mathrm{sp}}$ and $\ket{\mathrm{ns}}$ indicate the split and non-split branch respectively. The boundary conditions (or edge cases) are defined by
\begin{equation}
	 B_\gamma\ket t
	=\frac{1}{\sqrt{2}}\cdot \begin{cases}
	\ket{\mathrm{sp},1}+\ket{\mathrm{ns},1} & \text{ if } t=0\\
	\chi_\gamma(1/2)\left(\ket{\mathrm{sp},1}-\ket{\mathrm{ns},1}\right) & \text{ if } t=1/4 \text{ and } \operatorname{char}(\F_q) \text{ is odd}.
	\end{cases}
\end{equation}
\end{definition}

\begin{definition}
On the two branches, define
\begin{align}
 \Fou_\times\ket{\mathrm{sp},r}
 &=\frac1{\sqrt{q-1}}\sum_{\delta\in X_q}
   \chi_\delta(r)\ket{\mathrm{sp},\delta},\nonumber\\
 \Fou_{\ker\No}\ket{\mathrm{ns},r}
 &=\frac1{\sqrt{q+1}}\sum_{\nu\in\widehat{\ker\No}}
   \nu(r)\ket{\mathrm{ns},\theta_0\nu}.
 \label{eq: GL2 torus Fourier gates}
\end{align}
\end{definition}

On both branches, states can be combined in pairs that give the same representation label. For split items we pair them as $(\alpha, \alpha^{-1}\gamma)$ and the non-split ones are paired by $(\theta, \theta^q)$. Note that there are edge cases when the two elements are the same, which we treat separately from the \emph{non-fixed pairs}.

\begin{definition}[Fixing labels]
For a non-fixed split pair $(\delta, \delta^{-1}\gamma)$ and a non-fixed Frobenius pair  $(\theta, \theta^q)$, set
\begin{align}
 R_{\mathrm{fix},\gamma}
 \frac{\ket{\mathrm{sp},\delta}
      +\ket{\mathrm{sp},\delta^{-1}\gamma}}{\sqrt2}
 &=\ket{I_{\delta,\gamma\delta^{-1}}}
 &&\delta^2\ne\gamma,\nonumber\\
 R_{\mathrm{fix},\gamma}
 \frac{\ket{\mathrm{ns},\theta}+\ket{\mathrm{ns},\theta^q}}{\sqrt2}
 &=\ket{\pi_\theta}
 &&\theta^q\ne\theta.
 \label{eq: GL2 nonfixed orbit folding}
\end{align}
The antisymmetric combinations are sent to unused orthogonal labels.  For each $\delta^2=\gamma$, set
\begin{align}
 R_{\mathrm{fix},\gamma}\ket{\mathrm{sp},\delta}
 &=\sqrt{\frac{q+1}{2q}}\ket{\St_\delta}
   +\sqrt{\frac{q-1}{2q}}\ket{\det_\delta},\nonumber\\
 R_{\mathrm{fix},\gamma}\ket{\mathrm{ns},\chi_\delta\circ\No}
 &=\sqrt{\frac{q-1}{2q}}\ket{\St_\delta}
   -\sqrt{\frac{q+1}{2q}}\ket{\det_\delta}.
 \label{eq: GL2 fixed torus rotation}
\end{align}
\end{definition}
The symmetric and antisymmetric combinations form an orthonormal basis, and the matrix in \cref{eq: GL2 fixed torus rotation} is orthogonal.
Hence $R_{\mathrm{fix},\gamma}$ is unitary.

\begin{definition}[The exceptional correction]
For every $\alpha^2=\gamma$, let $R_{\mathrm{exc},\gamma}$ act on the ordered pair
\begin{equation}
 \left(A_\gamma\ket{c_\alpha},\ket{\det_\alpha}\right)
\end{equation}
by
\begin{equation}
 \begin{pmatrix}
 -\chi_\alpha(-1)/\sqrt q & \sqrt{(q-1)/q}\\
 \sqrt{(q-1)/q} & \chi_\alpha(-1)/\sqrt q
 \end{pmatrix}.
 \label{eq: GL2 exceptional reflection}
\end{equation}
On the orthogonal complement it acts as the identity.
\end{definition}

\subsubsection{Correctness of the circuit}

This is a technical subsection that ends with \cref{eq: GL2 final hard A factorization} to show the circuit of \cref{fig: GL2 hard A circuit} indeed implements $A_\gamma$. The longest proofs have been put in \cref{app: proofs of A matrix}.

We first check that the branching formulas define an isometry.

\begin{lemma}
\label{lem: GL2 quadratic router isometry}
The images of the $q$ basis states under $B_\gamma$ are orthonormal.
\end{lemma}
\begin{proof}
The value of a non-boundary $t$ determines the roots of the  polynomial $X^2-X+t$, so different values of $t$ have disjoint support and each image in \cref{eq: GL2 split router} has norm one.  For odd $q$, the two boundary vectors in \cref{eq: GL2 split router} are orthonormal.  They are supported on $\Span\{\ket{\mathrm{sp},1},\ket{\mathrm{ns},1}\}$, which is not used by the other inputs.  For even $q$ there is one normalized boundary vector in this space.  Hence all $q$ images are orthonormal.
\end{proof}

For the middle part of the circuit, use the abbreviation
\begin{equation}
	\T_\gamma
	=R_{\mathrm{fix},\gamma}
	\bigl(\Fou_\times\oplus\Fou_{\ker\No}\bigr)
	B_\gamma\Fou_\times^\dagger.
	\label{eq: GL2 relative torus transform}
\end{equation}
It is denoted by $\T_{\gamma}$ because it branches and acts on the two different tori (split and non-split) of $\GL_2(\F_q)$. The main technical lemma of this section is the action of $\T_\gamma$ on the input states.

\begin{lemma}[Torus action]
\label{lem: GL2 relative torus coefficients}
On the special input symbol, the torus circuit acts as
\begin{align}
 \T_\gamma\ket{c_\infty}
 ={}&\frac1{\sqrt{q+1}}
   \sum_{\{\theta,\theta^q\}\in\mathcal C_\gamma}\ket{\pi_\theta}
  +\frac1{\sqrt{q-1}}
   \sum_{\{\delta,\eta\}\in\mathcal P_\gamma}
   \ket{I_{\delta,\eta}}\nonumber\\
 &+\sqrt{\frac q{(q-1)(q+1)}}
   \sum_{\delta\in\mathcal S_\gamma}\ket{\St_\delta}.
 \label{eq: GL2 relative identity column}
\end{align}
For input $\alpha\in X_q$ and output labels $I_{\delta,\eta}, \pi_{\theta}$, the coefficients are
\begin{align}
 \bra{I_{\delta,\eta}}\T_\gamma\ket{c_\alpha}
 &=\frac1{q-1}J(\alpha^{-1}\delta,\alpha^{-1}\eta),
 \label{eq: GL2 relative principal coefficient}\\
 \bra{\pi_\theta}\T_\gamma\ket{c_\alpha}
 &=-\frac1{\sqrt{(q-1)(q+1)}}
   \sum_{\tr(z)=1}\theta(z)\chi_{\alpha^{-1}}(\No z).
 \label{eq: GL2 relative cuspidal coefficient}
\end{align}
See \cref{lem: Gauss sum property} for the definition of the Jacobi sum $J$. If $\delta^2=\gamma$ and $\alpha\ne b$, then
\begin{equation}
 \bra{\St_\delta}\T_\gamma\ket{c_\alpha}
 =\frac{\sqrt q}{(q-1)\sqrt{q+1}}
   J(\alpha^{-1}\delta,\alpha^{-1}\delta),
 \qquad
 \bra{\det_\delta}\T_\gamma\ket{c_\alpha}=0.
 \label{eq: GL2 relative fixed coefficient}
\end{equation}
\end{lemma}

\begin{proof}
We refer to \cref{app: proofs of A matrix}.
\end{proof}

With the main technical lemma out of the way, we can show that adding a diagonal gate at the beginning gives the correct phases for most cases.
\begin{lemma}[Action on generic columns]
\label{lem: GL2 generic hard columns}
If $\alpha^2\ne \gamma$, then
\begin{equation}
 \T_\gamma D_\gamma\ket{c_\alpha}=A_\gamma\ket{c_\alpha}.
\end{equation}
The same equality holds for $\ket{c_\infty}$.
\end{lemma}
\begin{proof}
Recall that the rows of $A_\gamma$ are indexed by the irreps of $G$.
We check for each row that \cref{eq: corrected A nontrivial column} is satisfied.

For a principal row $I_{\delta, \eta}$ with $\delta \neq \eta$ and $\delta\eta=\gamma$, put $\xi=\alpha^{-1}\delta$ and $\zeta=\alpha^{-1}\eta$ so that $\xi\zeta=\alpha^{-1}b\ne1$.  We have
\begin{equation*}
 g_q(\xi)g_q(\zeta)
 =g_q(\alpha^{-1}b)J(\xi,\zeta),
\end{equation*}
 which follows from \cref{lem: Gauss sum property} if $\xi,\zeta\ne1$.  If one of them is trivial, say $\xi=1$, then $J(1,\zeta)=-1=g_q(1)$, so the same equation still holds.  Multiplying \cref{eq: GL2 relative principal coefficient} by the phase from $D_\gamma$ gives the principal coefficient in \cref{eq: corrected A nontrivial column}.

For a cuspidal row $\pi_{\theta}$, put
\begin{equation*}
 \Theta=\theta\cdot(\chi_{\alpha^{-1}}\circ\No).
\end{equation*}
Its restriction to $\F_q^*$ is $\chi_{\alpha^{-1}b}$.  Applying \cref{lem: Gauss sum trace fibers} gives
\begin{equation*}
 g_{q^2}(\Theta)
 =g_q(\alpha^{-1}b)\sum_{\tr(z)=1}\Theta(z),
\end{equation*}
where we used that the trace-zero contribution is zero because $\alpha^{-1}b\ne1$.
This identity together with \cref{eq: GL2 relative cuspidal coefficient} gives the cuspidal coefficient in \cref{eq: corrected A nontrivial column}.
For the Steinberg coefficient, consider a  fixed row indexed by $\St_{\delta}$ with $\delta^2 = \gamma$. Putting $\xi = \alpha^{-1}\delta$ we have $\xi^2 = \alpha^{-2}\gamma \neq 1$, so that we can apply \cref{eq: Gauss sum prop 3} together with \cref{eq: GL2 relative fixed coefficient}.
Finally, \cref{eq: GL2 relative identity column} is exactly \cref{eq: corrected A identity column}.
\end{proof}

The case $\alpha^2=\gamma$ is slightly different because $g_q(1)/\sqrt q$ is not normalized.
We use the following two lemmas to show that the correction operator $R_{\mathrm{exc},\gamma}$ takes care of these columns.

\begin{lemma}[Action on exceptional columns]
\label{lem: GL2 exceptional columns}
If $\alpha^2=\gamma$, then
\begin{equation}
 \T_\gamma\ket{c_\alpha}
 =-\frac{\chi_\alpha(-1)}{\sqrt q}A_\gamma\ket{c_\alpha}
  +\sqrt{\frac{q-1}{q}}\ket{\det_\alpha}.
 \label{eq: GL2 exceptional leakage}
\end{equation}
\end{lemma}
\begin{proof}
As in the proof of \cref{lem: GL2 generic hard columns}, we check the action on every row.
For a row indexed by a non-fixed principal series $I_{\delta, \eta}$ with $\delta \neq \eta$, put $\xi=\alpha^{-1}\delta$.  The other shifted character is $\xi^{-1} = \alpha^{-1}\eta$, and
\begin{equation*}
 J(\xi,\xi^{-1})=-\xi(-1),
 \qquad
 g_q(\xi)g_q(\xi^{-1})=q\xi(-1).
\end{equation*}
Thus the coefficient produced by $\T_\gamma$ for that principal series is $-\chi_\alpha(-1)/\sqrt q$ times the required coefficient.

For a cuspidal label $\pi_{\theta}$, put $\Theta=\theta(\chi_{\alpha^{-1}}\circ\No)$. Then $\Theta$ is non-trivial, while $\Theta|_{\F_q^*} = \chi_\gamma \chi_{\alpha^{-2}}$ is trivial on $\F_q^*$.
Choose $\zeta\ne0$ with $\tr(\zeta)=0$.  Splitting into trace fibers as in \cref{lem: Gauss sum trace fibers} gives
\begin{equation*}
 g_{q^2}(\Theta)=q\Theta(\zeta),
 \qquad
 \sum_{\tr(z)=1}\Theta(z)=-\Theta(\zeta),
\end{equation*}
by \cref{eq: Gauss sum trivial restriction fibers}, so the same factor appears.

For the fixed row $\delta=\alpha$ labeling $\St_{\delta}$, the split and non-split amplitudes before the fixed rotation are
\begin{equation*}
 \frac{q-2}{\sqrt2(q-1)},
 \qquad
 -\frac q{\sqrt{2(q-1)(q+1)}}.
\end{equation*}
The rotation sends them to
\begin{equation*}
 -\frac1{\sqrt q(q-1)\sqrt{q+1}}\ket{\St_\alpha}
 +\sqrt{\frac{q-1}{q}}\ket{\det_\alpha}.
\end{equation*}
This gives the stated Steinberg factor and determinant term.  If there is a second fixed label $\delta$ (for odd characteristic), then $\xi=\alpha^{-1}\delta$ is quadratic, $J(\xi,\xi)=-\xi(-1)$, and
\begin{equation*}
 \sum_{\tr(z)=1}(\xi\circ\No)(z)=\xi(-1).
\end{equation*}
The determinant coefficient is then zero and the Steinberg coefficient has the same factor.  This proves the lemma.
\end{proof}

\begin{lemma}[Exceptional column rotation]
\label{lem: GL2 exceptional correction}
The gate $R_{\mathrm{exc},\gamma}$ is unitary, fixes the identity and generic columns, and satisfies
\begin{equation}
 R_{\mathrm{exc},\gamma}\T_\gamma\ket{c_\alpha}
 =A_\gamma\ket{c_\alpha}
 \qquad(\alpha^2=\gamma).
\end{equation}
\end{lemma}
\begin{proof}
The columns $A_\gamma\ket{c_\alpha}$ are orthonormal, the determinant states are orthonormal, and the two families belong to different output spaces.  Hence the planes in the definition of $R_{\mathrm{exc},\gamma}$ are orthogonal.  The matrix in \cref{eq: GL2 exceptional reflection} is orthogonal, and it holds that
\begin{equation*}
 \begin{pmatrix}
 -\chi_\alpha(-1)/\sqrt q & \sqrt{(q-1)/q}\\
 \sqrt{(q-1)/q} & \chi_\alpha(-1)/\sqrt q
 \end{pmatrix}
 \begin{pmatrix}
 -\chi_\alpha(-1)/\sqrt q\\[1mm]
 \sqrt{(q-1)/q}
 \end{pmatrix}
 =\begin{pmatrix}1\\0\end{pmatrix}.
\end{equation*}
Combining this with \cref{eq: GL2 exceptional leakage}, we have indeed the rotation needed for the exceptional columns. The other $A_\gamma$-columns are orthogonal to these planes and have no determinant component, so they are fixed.
\end{proof}

Putting everything together, we get the correctness of the circuit.

\begin{theorem}
\label{thm: GL2 hard A correctness}
For every input state $\ket{c}\in\Span\{\ket{c_\infty},\ket{c_\alpha}:\alpha\in X_q\}$,
\begin{equation}
 R_{\mathrm{exc}, \gamma}~\T_\gamma~D_{\gamma} \bigl(\ket{c}\ket{0_{\mathrm{work}}}\bigr)
 =\bigl(A_\gamma\ket{c}\bigr)\ket{0_{\mathrm{work}}}.
 \label{eq: GL2 final hard A factorization}
\end{equation}
\end{theorem}

\subsubsection{Efficient implementation of the gates}

We show that each gate in \cref{fig: GL2 hard A circuit} has a circuit of size $\operatorname{poly}(\log q,\log(1/\varepsilon))$. Again the longer proofs are put in \cref{app: proofs of A matrix}.

\begin{lemma}[Implementation of $D_\gamma$]
The phase gate $D_\gamma$ can be implemented to operator-norm error at most $\varepsilon$ with clean ancillas and $\operatorname{poly}(\log q,\log(1/\varepsilon))$ gates.
\end{lemma}
\begin{proof}
On the zero label, do nothing.  On a nonzero label, compute from $(\gamma,\alpha)$ the characters $b=\alpha^{-1}\gamma$ and $\alpha^{-1}b$.  If $\alpha^{-1}b\ne1$, use \cref{thm: coherent Gauss phase} to apply the normalized Gauss phase, and use character arithmetic to apply $\chi_b(-1)$.  If $\alpha^{-1}b=1$ do nothing.  Uncomputing the two character labels returns the work registers to zero.
\end{proof}

\begin{lemma}[Implementation of the Fourier transforms]
The gates $\Fou_\times^\dagger$ and $\Fou_\times\oplus\Fou_{\ker\No}$ can be implemented to operator-norm error at most $\varepsilon$ with $\operatorname{poly}(\log q,\log(1/\varepsilon))$ gates.
\end{lemma}
\begin{proof}
The first gate is the adjoint of the multiplicative Fourier transform in \cref{lem: add and mult of Fq}.  The split transform is the same cyclic QFT.  The group $\ker\No$ is cyclic of order $q+1$ and is generated by $h^{q-1}$.  Thus the non-split transform is a cyclic QFT of order $q+1$.
We can convert between group elements and exponent labels by exponentiation and discrete logarithm in $\F_{q^2}^*$ as in \cref{lem: add and mult of Fq} with $q$ replaced by $q^2$.  Adding the lifted exponent of $\theta_0$ (which was fixed earlier) gives the label $\theta_0\nu$.
\end{proof}

\begin{lemma}[Implementation of $B_\gamma$]\label{lem: B implementation}
The unitary extension of $B_\gamma$ can be implemented to operator-norm error at most $\varepsilon$ with $\operatorname{poly}(\log q,\log(1/\varepsilon))$ gates.
\end{lemma}
\begin{proof}
We refer to \cref{app: proofs of A matrix}.
\end{proof}

\begin{lemma}[Implementation of $R_{\mathrm{fix},\gamma}$]
The labeling gate $R_{\mathrm{fix}, \gamma}$ can be implemented to operator-norm error at most $\varepsilon$ with $\operatorname{poly}(\log q,\log(1/\varepsilon))$ gates.
\end{lemma}

\begin{proof}
The circuit to implement this gate is shown in \cref{fig: GL2 fixed label gate ancillas}. We explain each part.

First we compute the orbit partner $\lambda'$ with $U_{\iota}$.
On the split branch, compute $\delta\mapsto\gamma\delta^{-1}$.  On the nonsplit branch, the character coordinate is sent to its Frobenius conjugate by \cref{eq: GL2 affine Frobenius coordinate}, which comes down to arithmetic on the exponents.

In either case, compare the two exponent labels to choose an orbit representative (bases in the ordering we fixed) and store the order in one bit $o$ with $U_{\operatorname{can}}$. This gate also computes a fixed-point flag $f = [\lambda = \lambda']$. On the split branch, the test $\delta^2=\gamma$ finds the fixed labels.  On the non-split branch, the fixed test is $\theta^q=\theta$. For such a label, cyclic exponent arithmetic recovers the unique $\delta$ with $\theta=\delta\circ\No$.

A Hadamard gate on the order bit $o$ sends the symmetric combination to the physical representation label and the antisymmetric combination to an unused label.

On the fixed pairs, apply the rotation in \cref{eq: GL2 fixed torus rotation} with $U_{\operatorname{out}}$, whose entries can be approximated with $\operatorname{poly}(\log q,\log(1/\varepsilon))$ gates. The same gate also relabels the fixed pairs for Steinberg and determinant representations, and uncomputes the work register.
\end{proof}
\begin{figure}[H]
\centering
{\scriptsize
\begin{quantikz}[row sep=0.25cm,column sep=0.38cm]
\lstick{$\ket\gamma$}
 & \ctrl{1}
 & \qw
 & \qw
 & \qw
 & \rstick{$\ket\gamma$}
\\
\lstick{$\ket{\tau,\lambda}$}
 & \gate[2]{U_\iota}
 & \gate[4]{U_{\mathrm{can}}}
 & \qw
 & \gate[5]{U_{\mathrm{out}}}
 & \rstick{$\ket\mu$}
\\
\lstick{$\ket0_{\lambda'}$}
 &
 &
 & \qw
 &
 & \rstick{$\ket0_{\lambda'}$}
\\
\lstick{$\ket0_f$}
 & \qw
 &
 & \qw
 & \ctrl{1}
 & \rstick{$\ket0_f$}
\\
\lstick{$\ket0_o$}
 & \qw
 &
 & \gate{H_{\,f=0}}
 &
 & \rstick{$\ket0_o$}
\\
\lstick{$\ket0_{\mathsf a_{\mathrm{fix}}}$}
 & \qw
 & \qw
 & \qw
 &
 & \rstick{$\ket0_{\mathsf a_{\mathrm{fix}}}$}
\end{quantikz}
}
\caption{Ancilla-level circuit for $R_{\mathrm{fix},\gamma}$.
Here $\tau\in\{\mathrm{sp},\mathrm{ns}\}$ is the branch and $\lambda$ is its character coordinate.
The register $\lambda'$ is for the conjugate element in the orbit.
Other registers are for the fixed flag, ordering and ancillas.}
\label{fig: GL2 fixed label gate ancillas}
\end{figure}
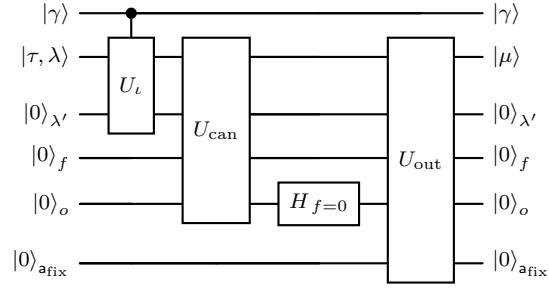

To implement the last gate, we use a lemma that tells us we can efficiently prepare the exceptional target columns.

\begin{lemma}
\label{lem: GL2 exceptional target preparation}
If $\alpha^2=\gamma$, then $A_\gamma\ket{c_\alpha}$ can be prepared to error at most $\varepsilon$ with $\operatorname{poly}(\log q,\log(1/\varepsilon))$ gates.
\end{lemma}
\begin{proof}
We refer to \cref{app: proofs of A matrix}.
\end{proof}

\begin{lemma}[Implementation of $R_{\mathrm{exc},\gamma}$]
The exceptional correction gate $R_{\mathrm{exc},\gamma}$ can be implemented to operator-norm error at most $\varepsilon$ with $\operatorname{poly}(\log q,\log(1/\varepsilon))$ gates.
\end{lemma}
\begin{proof}
We implement this gate via a symmetric prepare--reflect--unprepare construction as shown in \cref{fig: GL2 exceptional correction ancillas}.
First, we solve $\alpha^2=\gamma$ in the cyclic exponent labels.  For odd $q$, there are zero or two solutions (two in the exceptional case), and exactly one for even $q$. These are efficient computations when working with the exponents. Then, for each solution, the matrix in \cref{eq: GL2 exceptional reflection} is the reflection with normalized normal vector
\begin{equation*}
 -\sqrt{\frac{1+\chi_\alpha(-1)/\sqrt q}{2}}
  A_\gamma\ket{c_\alpha}
 +\sqrt{\frac{1-\chi_\alpha(-1)/\sqrt q}{2}}
  \ket{\det_\alpha}.
\end{equation*}
This state can be prepared using \cref{lem: GL2 exceptional target preparation} (which generates $A_{\gamma} \ket{c_{\alpha}}$) and one additional two-dimensional rotation.  If $P$ is its preparation circuit, the sequence $P^\dagger$, a phase flip on $\ket{0}$, and $P$ implements the required reflection.  Applying this for all solutions gives $R_{\mathrm{exc},\gamma}$.
\end{proof}
\begin{figure}[H]
\centering
{\scriptsize
\begin{quantikz}[row sep=0.25cm,column sep=0.36cm]
\lstick{$\ket\gamma$}
 & \ctrl{1}
 & \qw
 & \qw
 & \qw
 & \ctrl{1}
 & \rstick{$\ket\gamma$}
\\
\lstick{$\ket0_{\alpha_j}$}
 & \gate[2]{U_{\mathrm{sol}}^{(j)}}
 & \qw
 & \qw
 & \qw
 & \gate[2]{(U_{\mathrm{sol}}^{(j)})^\dagger}
 & \rstick{$\ket0_{\alpha_j}$}
\\
\lstick{$\ket0_f$}
 &
 & \ctrl{1}
 & \ctrl{1}
 & \ctrl{1}
 &
 & \rstick{$\ket0_f$}
\\
\lstick{$\ket\mu$}
 & \qw
 & \gate[2]{P^\dagger}
 & \gate{Z_{\ket0}}
 & \gate[2]{P}
 & \qw
 & \rstick{$\ket{\mu'}$}
\\
\lstick{$\ket0_{\mathsf a_{\mathrm{exc}}}$}
 & \qw
 &
 & \qw
 &
 & \qw
 & \rstick{$\ket0_{\mathsf a_{\mathrm{exc}}}$}
\end{quantikz}
}
\caption{One controlled reflection used in $R_{\mathrm{exc},\gamma}$.
For $j\in\{0,1\}$, $U_{\mathrm{sol}}^{(j)}$ computes the $j$th solution $\alpha_j$ of $\alpha_j^2=\gamma$, when it exists, and sets the validity flag $f$.
Conditional on $f=1$, the circuit $P$ prepares the normalized normal vector displayed in the proof and $Z_{\ket{0}}$ applies a phase flip so that $PZ_{\ket{0}}P^\dagger$ is the reflection in the correct plane.
Applying this circuit for $j=0,1$ implements $R_{\mathrm{exc},\gamma}$.
When a solution does not exist, the flag is zero and the reflection is not applied.}
\label{fig: GL2 exceptional correction ancillas}
\end{figure}
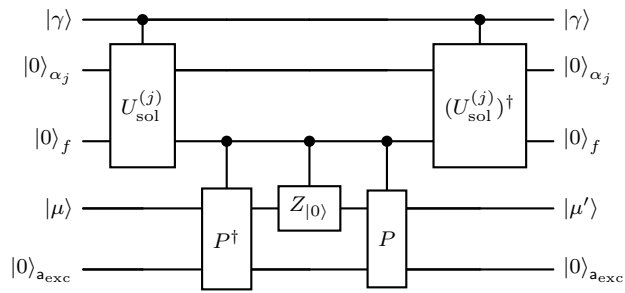

Putting everything together now proves the efficient implementation of the $A_\gamma$ matrix.
\begin{theorem}[Efficient hard \texorpdfstring{$A$}{A}-block]
\label{thm: efficient GL2 hard A circuit}
For every $\varepsilon>0$, the circuit for $A_\gamma$ in \cref{fig: GL2 hard A circuit} has a uniform implementation of size
\begin{equation}
 \operatorname{poly}\left(\log q,\log\frac1\varepsilon\right)
\end{equation}
whose action on the clean-work source space is within operator-norm error $\varepsilon$ of $A_\gamma$.  The implementation is coherent in $\gamma$.
\end{theorem}
\begin{proof}
Each operator in \cref{fig: GL2 hard A circuit} has the stated efficiency by the previous lemmas.
  There are only a constant number of factors, so we can divide the error among them.  Their product is indeed $A_\gamma$ by \cref{thm: GL2 hard A correctness}.
\end{proof}

\section{Inducing from normal subgroups and Clifford theory}\label{sec: clifford theory}

In this section, we present another new method to implement QFT via a subgroup.
While the subgroup construction of \cref{sec: first induction} and the resulting Beals and Mackey algorithms assume multiplicity-free restrictions,
here we consider a normal subgroup $N\unlhd G$ and allow restrictions to $N$ to have multiplicities.
We recall the classical Clifford-theoretic classification of the irreducible representations of $G$ in terms of those of $N$, and then use \emph{Mackey's little group method} to derive a QFT for $G$ from that of $N$,
provided our groups meet certain requirements and we are given black-box access to operations on $N$ and $G$. The relevant theory discussed can be found in Ref.~\cite[Chapter 2]{CecchFinitegroup}, but we provide proofs of most results for completeness.

The construction follows the same strategy as the preceding sections by performing the subgroup QFT first, with the output in this case expressed in the little group Fourier basis defined in \cref{def:little-group-basis}.
The extra normality constraint means that there will not be a precise reduction to an induction transform as in \cref{def: IndTrans}, but the QFT circuit of \cref{fig: little group qft} will nonetheless show many resemblances to \cref{fig: induction step of first induction}.

\subsection{\texorpdfstring{$G$-action on $\hatN$}{G-action on the dual of N} and inertia groups}

For $N\unlhd G$, we use the unitary dual $\hatN$, chosen representatives $\Rep_\la$, and bases fixed in \cref{subsec: representations and duals}.

For $g\in G$ and $\la\in\hatN$, define the \emph{transported representation} by
\begin{equation}\label{eq:transported-representation}
    \Rep_\la^g(n)
    :=
    \Rep_\la(g\inv ng),
    \qquad n\in N,
\end{equation}
and define the action of $G$ on $\hatN$ by
\begin{equation}\label{eq:G-action-on-irrep-labels}
    \act{g}\la
    :=
    [\Rep_\la^g]
\end{equation}
where $[\Rep_\la^g] \in \hatN$ denotes the label of the transported irrep.
As $N\unlhd G$, this is precisely the normalizer twist introduced in \cref{eq: normalized subgroup irrep twist}.  In the notation of \cref{sec: Mackey induction transform}, $\Rep_\la^g$ is $R_\kappa^\omega$, and $\act{g}\la$ is the label $\kappa^\omega$, with $\omega=g$ and $\kappa=\la$.
Thus $\act{g}\la$ is the equivalence-class label associated to the representation $\Rep_\la^g$, so that
\begin{equation*}
    \Rep_{\act{g}\la}\overset{N}{\cong}\Rep_\la^g.
\end{equation*}
More generally, if $(\Rep,V)$ is a $K$-module for $K\leq G$, then $V^g$ denotes the transported $gKg\inv$-module with the same underlying vector space and action
\begin{equation}\label{eq:transported-module}
    \Rep^g(x):=\Rep(g\inv xg),
    \qquad x\in gKg\inv.
\end{equation}
This is the twisted module denoted by $V^\omega$ in \cref{lemma: Mackey}.

\begin{definition}
For $\la\in\hatN$, we define:
\begin{itemize}
    \item $\orb_G(\la):=\{\act{g}\la: g\in G\} \subseteq \hatN$, the \emph{$G$-orbit} of $\la$ in $\hatN$,
    \item $\Inertia{\la}:=\Stab_G(\la):=\{g\in G: \act{g}\la=\la\} \leq G$, the \emph{inertia group} for $\la$ in $G$.
\end{itemize}
\end{definition}

\begin{definition}\label{def: representatives eq classes}
\
    \begin{enumerate}
    \item We denote by $\orbreps\subseteq \hatN$ a set of representative elements for distinct $G$-orbits in $\hatN$, so that
\begin{equation}
\hatN=\bigsqcup_{\sigma\in \orbreps} \orb_G(\sigma).
\end{equation}
We also define the representative function $\repO: \hatN\rightarrow \orbreps$, which sends $\la$ to the unique element in $\orbreps \cap \orb_G(\la)$.

\item Fix $\la\in \hatN$, and consider the following equivalence relation on $G$:
\begin{equation}
g_1\sim g_2
\iff
\act{g_1}\la=\act{g_2}\la.
\end{equation}
We denote by $\Transversal{\la}\subseteq G$ a set of representatives for this equivalence relation.
In particular, we write
\begin{equation}
\Transversal{\la}=\{\TransversalElt{\la,\mu}: \mu\in \orb_G(\la)\},
\end{equation}
where $\TransversalElt{\la,\mu}\in G$ is the unique element in $\Transversal{\la}$ such that
\begin{equation}
\act{\TransversalElt{\la,\mu}}\la=\mu.
\end{equation}
We normalize this choice by requiring
\begin{equation}
    \TransversalElt{\la,\la}=1.
\end{equation}
\end{enumerate}
\end{definition}

\begin{lemma}\label{lem: transv inertia}
    $\Transversal{\la}$ is a left transversal for $\Inertia{\la}$ in $G$.
\end{lemma}

\begin{proof}
By the orbit--stabilizer theorem, the map
\begin{equation*}
G/\Inertia{\la}\longrightarrow \orb_G(\la),
\qquad
g\Inertia{\la}\longmapsto \act{g}\la,
\end{equation*}
is a bijection. Since $\Transversal{\la}$ contains exactly one element $\TransversalElt{\la,\mu}$ for each $\mu\in\orb_G(\la)$, it contains exactly one representative of each left coset of $\Inertia{\la}$ in $G$.
\end{proof}

The relevance of the $G$-orbits on $\hatN$ to the representation theory of $G$ is given by the following result:

\begin{lemma}[Clifford's theorem]\label{lem: cliff thm}
Let $G$ be a finite group, $N \unlhd G$ a normal subgroup, and $(\Rep,V)$ an irreducible representation for $G$. Then, upon restriction to $N$, $V$ decomposes into a direct sum of irreducible $N$-modules that all belong to a single $G$-orbit in $\hatN$.
Furthermore, every irreducible representation in this orbit appears with the exact same multiplicity. That is, there exists a unique $\sigma \in \orbreps$ and an integer $m \ge 1$ such that
\begin{equation}
\Res_N^G(V) \cong
\bigoplus_{\la \in \orb_G(\sigma)} V_{\la}\otimes \C^{m}.
\end{equation}
\end{lemma}

\begin{proof}
$\Res_N^G(V)$ contains at least one irreducible representation for  $N$, $(\Rep_\la,V_{\la})$, for some $\la \in \hatN$.

For any $g \in G$, consider the translated subspace $\Rep(g)V_{\la}$. We first show that $\Rep(g)V_{\la}$ is $N$-invariant: for any $h\in N$,
\begin{equation*}
\Rep(h)\Rep(g)V_{\la}= \Rep(g)\Rep(g\inv h g)V_\la =\Rep(g)\Rep_\la(g\inv h g)V_{\la},
\end{equation*}
and the invariance follows from the normality of $N$ and the $N$-invariance of $V_{\la}$. Furthermore, this shows that the action of $h$ on $\Rep(g)V_{\la}$ is given by the transported representation $\Rep_\la^g$.

Now, consider the sum of all such translated subspaces:
\begin{equation*}
U := \sum_{g \in G} \Rep(g)V_{\la}\le V.
\end{equation*}
By construction, any element of $G$ acts on $U$ by permuting the terms in the sum and thus $U$ is a $G$-invariant subspace of $V$. Since $V$ is irreducible and $U$ is non-zero, we must have $U = V$. As a consequence, any $N$-irreducible submodule in $V$ must project nontrivially onto one of the $\Rep(g)V_{\la}$, and therefore by Schur's lemma we have that it must be of the form $\act{g}\la$ for some $g \in G$.

Finally, for any $g \in G$, the linear map $\ket{v} \mapsto \Rep(g) \ket{v}$ is an invertible transformation on the vector space $V$, which maps the $\la$-isotypic component of $N$ onto the $\act{g}\la$-isotypic component. Therefore, each isotypic component has the exact same dimension, meaning that they appear with the same multiplicity $m$.
\end{proof}

\begin{lemma}\label{lem: inertia successors}
Let $\la\in\hatN$ and $\eta\in\widehat{\Inertia{\la}}$. Then
\begin{equation}
\eta\in\N_{N\leq\Inertia{\la}}^{+}(\la)
\iff
\Res_N^{\Inertia{\la}}(V_\eta)
\cong V_\la\otimes \C^{m_\eta}
\end{equation}
for some $m_\eta\geq 1$.
\end{lemma}

\begin{proof}
By definition, the left-hand side means that $V_\la$ occurs in $\Res_N^{\Inertia{\la}}(V_\eta)$. Since $\Inertia{\la}$ stabilizes $\la$, the claim follows from \cref{lem: cliff thm}.
\end{proof}

Inertia groups play a central role in classifying irreducible representations of $G$, due to the following result:

\begin{lemma}[Clifford's correspondence]\label{lem: cliff corr}
    Let $\la\in \hatN$ and $\eta\in\N_{N\leq\Inertia{\la}}^{+}(\la)$. Then $\Ind_{\Inertia{\la}}^G(V_\eta)$ is irreducible as a $G$-module.

    Moreover, let $\mu\in \hatN$ and $\gamma\in\N_{N\leq\Inertia{\mu}}^{+}(\mu)$. Then
    \begin{equation}
    \Ind_{\Inertia{\la}}^G(V_\eta)
    \cong
    \Ind_{\Inertia{\mu}}^G(V_{\gamma})
    \end{equation}
    if and only if there exists $g\in G$ such that
    \begin{equation}
    \la=\act{g}\mu
    \end{equation}
    and
    \begin{equation}
    V_\eta\cong V_{\gamma}^g
    \end{equation}
    as $\Inertia{\la}=\Inertia{\act{g}\mu}$-modules.
\end{lemma}

\begin{proof}

Fix $\la\in\hatN$ and let $\eta,\gamma\in\N_{N\leq\Inertia{\la}}^{+}(\la)$.
By Schur's lemma, dimensions of intertwiner spaces count the multiplicities of common irreducible summands.
Thus, applying \cref{thm: Frobenius reciprocity} to irreducible summands and using \cref{lemma: Mackey} gives
\begin{align*}
\dim\Hom_G\!\left(
\Ind_{\Inertia{\la}}^G V_\eta,
\Ind_{\Inertia{\la}}^G V_\gamma
\right) &=
\dim\Hom_{\Inertia{\la}}\!\left(
V_\eta,
\Res_{\Inertia{\la}}^G\Ind_{\Inertia{\la}}^G V_\gamma
\right) \\
&=
\sum_{x\in\Inertia{\la}\backslash G/\Inertia{\la}}
\dim\Hom_{\Inertia{\la}}\!\left(
V_\eta,
\Ind_{\Inertia{\la}\cap x\Inertia{\la}x\inv}^{\Inertia{\la}}
V_\gamma^x
\right) \\
&=
\sum_{x\in\Inertia{\la}\backslash G/\Inertia{\la}}
\dim\Hom_{\Inertia{\la}\cap x\Inertia{\la}x\inv}\!\left(
V_\eta,V_\gamma^x
\right),
\end{align*}
where the restrictions to the intersection groups are implicit.
Since $N\unlhd G$ and $N\leq\Inertia{\la}$, we have $N\leq\Inertia{\la}\cap x\Inertia{\la}x\inv$.
By \cref{lem: inertia successors}, the restrictions of $V_\eta$ and $V_\gamma^x$ to $N$ are isotypic of types $\la$ and $\act{x}\la$, respectively.
If $\act{x}\la\neq\la$, every homomorphism counted in the $x$-summand is an $N$-homomorphism between two isotypic modules of inequivalent types, and is therefore zero.
A summand can thus be nonzero only when $x\in\Inertia{\la}$, which is the identity double coset.
Consequently,
\begin{equation*}
\dim\Hom_G\!\left(
\Ind_{\Inertia{\la}}^G V_\eta,
\Ind_{\Inertia{\la}}^G V_\gamma
\right)
=
\dim\Hom_{\Inertia{\la}}(V_\eta,V_\gamma).
\end{equation*}
By Schur's lemma, this dimension is $1$ if $\eta=\gamma$ and $0$ otherwise.
This proves that $\Ind_{\Inertia{\lambda}}^G V_\eta$ is irreducible, and that the induced modules for distinct $\eta$ and $\gamma$ are inequivalent.

Now let $\mu\in\hatN$ and $\gamma\in\N_{N\leq\Inertia{\mu}}^{+}(\mu)$.
If $\la$ and $\mu$ lie in different $G$-orbits, the restrictions to $N$ of the two induced modules are supported on the disjoint orbits $\orb_G(\la)$ and $\orb_G(\mu)$, so the modules are inequivalent.
If $\la=\act{g}\mu$, then $\Inertia{\la}=g\Inertia{\mu}g\inv$ and $\Ind_{\Inertia{\mu}}^G V_\gamma \cong\Ind_{\Inertia{\la}}^G V_\gamma^g$.
Therefore
\begin{align*}
\dim\Hom_G\!\left(
\Ind_{\Inertia{\la}}^G V_\eta,
\Ind_{\Inertia{\mu}}^G V_\gamma
\right) =
\dim\Hom_{\Inertia{\la}}(V_\eta,V_\gamma^g).
\end{align*}
Since the induced modules are irreducible, they are equivalent exactly when this dimension is nonzero, equivalently when $V_\eta\cong V_\gamma^g$.
\end{proof}

\begin{theorem}\label{thm: hatG from hatinertia}
    Let $N\unlhd G$. Then $\hatG$ is in bijection with the set
    \begin{equation}
    L_N
    :=
    \left\{
    (\sigma,\eta)
    \ \middle|\
    \sigma\in \orbreps
    \text{ and }
    \eta\in\N_{N\leq\Inertia{\sigma}}^{+}(\sigma)
    \right\}.
    \end{equation}
\end{theorem}

\begin{proof}
    Define
    \begin{equation*}
    \Phi:L_N\to\hatG,
    \qquad
    \Phi(\sigma,\eta)
    :=
    \Ind_{\Inertia{\sigma}}^G(V_\eta).
    \end{equation*}
   By \cref{lem: cliff corr}, $\Phi(\sigma,\eta)$ is irreducible and inequivalent to $\Phi(\mu,\gamma)$ unless $\mu= \sigma$ and $\eta=\gamma$.
Therefore $\Phi$ is well defined and injective.

    To prove surjectivity, let $V$ be an irreducible $G$-module. By \cref{lem: cliff thm} $\Res_N^G(V)$ is supported on a single $G$-orbit in $\hatN$. Let $\sigma\in \orbreps$ be the fixed representative of this orbit and let $\IsoComp{\sigma}\subseteq V$ be the $\sigma$-isotypic component of $\Res_N^G(V)$. This subspace is invariant under $\Inertia{\sigma}$.
    Choose an irreducible $\Inertia{\sigma}$-submodule $V_\eta\subseteq \IsoComp{\sigma}$. Then $\eta\in\N_{N\leq\Inertia{\sigma}}^{+}(\sigma)$.

    Since $V_\eta$ occurs in $\Res_{\Inertia{\sigma}}^G(V)$, \cref{thm: Frobenius reciprocity} implies that $V$ occurs in $\Ind_{\Inertia{\sigma}}^G(V_\eta)$.
    The induced module is irreducible by \cref{lem: cliff corr}, so it is equivalent to $V$.
Hence $V\cong \Phi(\sigma,\eta)$ and $\Phi$ is surjective.
\end{proof}

Given a group $G$ with a normal subgroup $N$ we have therefore derived an approach to construct any element in $\hatG$:
\begin{enumerate}
    \item pick an element $\la\in \hatN$,
    \item find $\sigma=\repO(\la)$
    \item find $\Inertia{\sigma} \leq G$,
    \item pick an element $\eta\in\N_{N\leq\Inertia{\sigma}}^{+}(\sigma)$,
    \item construct $\Ind_{\Inertia{\sigma}}^G(V_\eta)$ with transversal $\Transversal{\sigma}$.
\end{enumerate}
According to the description of $\hatG$ made in \cref{thm: hatG from hatinertia}, we will have constructed the irreducible representation labeled by $(\sigma, \eta)$. We have thus essentially reduced the problem of describing $\hatG$ to that of describing $\hatN$ and the successor sets $\N_{N\leq\Inertia{\sigma}}^{+}(\sigma)$ for each $\sigma\in\orbreps$. In the following subsection we will see how in certain cases, including many useful ones, these successor sets admit a particularly simple description.

\subsection{Mackey's little group method}

We begin by observing that the action of $G$ on $\hatN$ factors through the quotient group
\begin{equation*}
H:=G/N.
\end{equation*}
Indeed, for $h\in N$ and $\la\in\hatN$, we have
\begin{equation*}
\Rep_\la^h(n)
=\Rep_\la(h\inv n h)
=\Rep_\la(h)\inv \Rep_\la(n)\Rep_\la(h),
\end{equation*}
for all $n\in N$. Thus $\Rep_\la^h$ is unitarily equivalent to $\Rep_\la$, with intertwiner $\Rep_\la(h)$. Therefore $\act{h}\la=\la$ for every $h\in N$ and $\la\in\hatN$, so that
\begin{equation*}
N\unlhd \Inertia{\la}\qquad \forall\la\in\hatN.
\end{equation*}
For each $\la\in\hatN$ we may thus define the little group
\begin{equation*}
\Little{\la}:=\Inertia{\la}/N\leq H,
\end{equation*}
or equivalently,
\begin{equation*}
\Little{\la}=\Stab_H(\la).
\end{equation*}

Mackey's little group method describes the successor set $\N_{N\leq\Inertia{\la}}^{+}(\la)$ in terms of irreducible representations of the little group $\Little{\la}$, provided that $\la$ \emph{extends} to a representation of $\Inertia{\la}$.

In order to state it precisely we need to introduce the notions of extension and inflation of a representation.

\begin{definition}
Let $K$ be a group, $L\leq K$, and let
\begin{equation}
\Rep_\la:L\to \U(V_\la)
\end{equation}
be a representation. We say that $\Rep_\la$ extends to $K$ if there exists a representation
\begin{equation}
\ExtRep_\la:K\to \U(V_\la)
\end{equation}
such that
\begin{equation}
\Res_L^K(\ExtRep_\la)=\Rep_\la.
\end{equation}
By \cref{def: dual representation}, the dual of an extension is an extension of the dual action on $V_\la^\vee$:
\begin{equation}\label{eq:dual-extension-convention}
    (\ExtRep_\la)^\vee(k):=\ExtRep_\la(k\inv)^{\mathsf T},
    \qquad
    \Res_L^K\bigl((\ExtRep_\la)^\vee\bigr)=\Rep_\la^\vee.
\end{equation}
\end{definition}
\begin{definition}
Let $K$ be a finite group, $L\unlhd K$, and let
\begin{equation}
\Rep: K/L\to \U(V)
\end{equation}
be a representation. We define its inflation to $K$
\begin{equation}\InfRep: K\to \U(V)\quad \text{ by }\quad
\InfRep(k):= \Rep(kL),
\end{equation}
where $kL\in K/L$ is the coset (left and right) of $L$ in $K$ containing $k$.
Inflation commutes with taking duals: $(\InfRep)^\vee(k)=\Rep^\vee(kL)$.
\end{definition}

In our setting, the relevant case is when
\begin{equation*}
\Rep_\la:N\to \U(V_\la)\quad \text{ extends to }\quad
\ExtRep_\la:\Inertia{\la}\to \U(V_\la).
\end{equation*}
for all $\la\in\hatN$.
Such extensions need not exist for an arbitrary normal subgroup $N\unlhd G$.
We give examples satisfying the extension hypothesis at the end of this subsection.

\begin{lemma}[Gallagher's theorem, inertia-group form]\label{thm: Gallagher}
    Let $N\unlhd G$. Suppose that, for $\la\in \hatN$,
\begin{equation}
\Rep_\la: N\to \U(V_\la)\quad\text{ extends to }\quad \ExtRep_\la: \Inertia{\la}\to \U(V_\la).
\end{equation}
    Then there is a bijection \begin{equation}
    \N_{N\leq\Inertia{\la}}^{+}(\la)
    \leftrightarrow
    \widehat{\Little{\la}}.
    \end{equation}
    Namely, $\eta\in\N_{N\leq\Inertia{\la}}^{+}(\la)$ corresponds uniquely to $\eta'\in\widehat{\Little{\la}}$ such that
    \begin{equation}
    (\Rep_\eta, V_\eta)\cong(\ExtRep_\la\otimes \InfRep_{\eta'}, V_\la \otimes V_{\eta'}),
    \end{equation}
    where $(\Rep_{\eta'}, V_{\eta'})$ is an irreducible representation for $\Little{\la}$, and $(\InfRep_{\eta'}, V_{\eta'})$ denotes its inflation to $\Inertia{\la}$.
\end{lemma}
\begin{proof}
Apply Ref.~\cite[Theorem~2.18]{CecchFinitegroup} with $G$ replaced by $\Inertia{\la}$ and with the extension $\ExtRep_\la$ of $\Rep_\la$.
\end{proof}

Combining \cref{thm: hatG from hatinertia} and \cref{thm: Gallagher}, we finally get

\begin{theorem}[Little group method]\label{thm: little group}

Let $N\unlhd G$. Suppose that, for all $\la\in \hatN$,
\begin{equation}
\Rep_\la: N\to \U(V_\la)\quad\text{ extends to }\quad \ExtRep_\la: \Inertia{\la}\to \U(V_\la).
\end{equation}
Then there is a bijection
\begin{equation}
\hatG\leftrightarrow\{(\sigma,\eta'): \sigma\in\orbreps,\ \eta'\in \widehat{\Little{\sigma}}\}.
\end{equation}
Namely the $G$-irreducible representation labeled by $(\sigma, \eta')$, which we denote by $(\Rep_{(\sigma,\eta')},V_{(\sigma,\eta')})$, can be constructed as
\begin{equation}
\Rep_{(\sigma, \eta')}\cong \Ind_{\Inertia{\sigma}}^G(\ExtRep_\sigma\otimes\InfRep_{\eta'})\qquad V_{(\sigma, \eta')}\cong \Ind_{\Inertia{\sigma}}^G(V_\sigma\otimes V_{\eta'}).
\end{equation}
An explicit basis for $V_{(\sigma,\eta')}$ is
\begin{equation}
\{\ket{\TransversalElt{\sigma,\mu}}\otimes \ket{i_\sigma}\otimes\ket{j_{\eta'}}: \TransversalElt{\sigma,\mu}\in \Transversal{\sigma},\  i=1,2,\dots,d_\sigma,\ j=1,2,\dots,d_{\eta'}\}.
\end{equation}
\end{theorem}
\begin{proof}
    This follows from \cref{thm: hatG from hatinertia,thm: Gallagher,def: induced representation}.
\end{proof}

Several families of normal subgroups $N\unlhd G$ satisfy the extension
hypothesis of \cref{thm: little group}. These include semidirect products
$G=N\rtimes H$ with $N$ abelian
\cite[Theorem~2.54]{CecchFinitegroup}, extensions with $G/N$ cyclic
\cite[Proposition~2.23 and Example~4.1]{CecchFinitegroup}, and wreath
products $F\wr H=F^n\rtimes H$
\cite[Lemma~2.59 and Theorem~2.60]{CecchFinitegroup}. In
\cref{sec: wreath}, we consider in detail wreath products of the form
$F\wr S_n$.

\begin{remark}
The Borel subgroup $B<\GL_2(\F_q)$ considered in \cref{sec: GL2} can be expressed as a semidirect product
\begin{equation}
    B\cong U\rtimes T,
    \qquad U\cong(\F_q,+)
\end{equation}
with abelian normal subgroup $U$, and therefore satisfies the hypothesis of \cref{thm: little group}. See \cref{eq: B} for the definitions. The action of $T$ on $\widehat U$ has one trivial orbit and one nontrivial orbit. The trivial orbit gives the one-dimensional representations $\chi_{\alpha,\beta}$, while the nontrivial orbit has stabilizer $Z$ and gives the representations $\rho_\gamma$. Thus the little group method recovers the classification of $\widehat B$ used in the GL$_2$ application.
\end{remark}

\subsection{Deriving QFTs from the little group method}

Throughout this subsection we assume that $N\unlhd G$ and that, for every $\la\in\hatN$, the representation
\begin{equation*}
\Rep_\la:N\to \U(V_\la)
\end{equation*}
extends to a representation
\begin{equation*}
\ExtRep_\la:\Inertia{\la}\to \U(V_\la).
\end{equation*}
Thus \cref{thm: little group} applies.

\begin{definition}[The ``little group'' Fourier basis]\label{def:little-group-basis}
Combining the general Fourier decomposition in \cref{eq: group algebra Peter-Weyl decomposition} with the classification in \cref{thm: little group} gives
\begin{equation}\label{eq: petweyl little group}
\C[G]\cong
\bigoplus_{\sigma\in\orbreps,\ \eta'\in\widehat{\Little{\sigma}}}
V_{(\sigma,\eta')}\otimes V_{(\sigma,\eta')}^\vee.
\end{equation}

With the induced bases fixed in \cref{thm: little group}, we call the basis vectors
\begin{equation}
\ket{\sigma,\eta'}
\otimes
\ket{t_L,i_\sigma,j_{\eta'}}
\otimes
\ket{t_R,k_{\sigma}^\vee,l_{\eta'}^\vee}
\end{equation}
the ``little group'' basis of $\C[G]$, where
\begin{equation}
\sigma\in\orbreps,
\qquad
\eta'\in\widehat{\Little{\sigma}},
\qquad
t_L,t_R\in\Transversal{\sigma},
\end{equation}
and
\begin{equation}
i,k=1,\dots,\dim V_\sigma,
\qquad
j,l=1,\dots,\dim V_{\eta'}.
\end{equation}
\end{definition}

We now define the unitaries that we will use in the construction of a QFT circuit for $G$.

\begin{definition}[Orbit encoding]\label{def:orbit-data-encoding}
Define
\begin{equation}
    \mathcal O_N
    :=
    \bigsqcup_{\sigma\in\orbreps}
    \{\sigma\}\times\Transversal{\sigma}.
\end{equation}
The representative choices in \cref{def: representatives eq classes} give, for every $\la\in\hatN$, unique elements $\sigma=\repO(\la)$ and $t_L=\TransversalElt{\sigma,\la}\in\Transversal{\sigma}$ such that $\la=\act{t_L}\sigma$. The corresponding classical encoding is
\begin{equation}\label{eq:orbit-data-encoding}
    \Enc_{\hatN,\mathcal O_N}\ket{\la}
    =
    \ket{\repO(\la)}
    \otimes
    \ket{\TransversalElt{\repO(\la),\la}}
    =
    \ket{\sigma}\otimes\ket{t_L}.
\end{equation}
\end{definition}

\begin{definition}[Orbit transport gates]\label{def:orbit-transport-gates}
Let $\sigma\in\orbreps$ and $t\in\Transversal{\sigma}$, and set $\la=\act{t}\sigma$. Since $N\unlhd G$, every $t\in G$ normalizes $N$, so the normalizer twisting construction of \cref{subsec: omega normalizes} applies with $H_\omega=N$, $\omega=t$, and $\kappa=\sigma$. In the notation of \cref{eq: normalizer block intertwiner}, its unitary block is
\begin{equation}\label{eq:clifford-normalizer-block-intertwiner}
    U_{t,\sigma}:V_\sigma\longrightarrow V_{\act{t}\sigma}=V_\la,
    \qquad
    U_{t,\sigma}\Rep_\sigma^t(n)U_{t,\sigma}^\dagger
    =\Rep_\la(n),
    \qquad n\in N.
\end{equation}
We normalize $U_{1,\sigma}=\id_{V_\sigma}$. Equivalently, for every $n\in N$,
\begin{equation}\label{eq:orbit-frame-intertwining-relation}
    U_{t,\sigma}^\dagger\Rep_\la(n)U_{t,\sigma}
    =\Rep_\sigma^t(n)
    =\Rep_\sigma(t\inv nt).
\end{equation}

Define the controlled left transport gate
\begin{equation}\label{eq:U-orb-left-unitary}
    U_{\mathrm{orb}}^L:=
    \bigoplus_{\sigma\in\orbreps}
    \bigoplus_{t\in\Transversal{\sigma}}
    \ket{\sigma,t}\bra{\sigma, t}\otimes
    U_{t,\sigma}^\dagger
\end{equation}
and the controlled right transport gate
\begin{equation}\label{eq:U-orb-right-unitary}
    U_{\mathrm{orb}}^R:=
    \bigoplus_{\sigma\in\orbreps}
    \bigoplus_{t\in\Transversal{\sigma}}
    \ket{\sigma,t}\bra{\sigma,t}\otimes
    U_{t,\sigma}\tp
\end{equation}
\end{definition}

\begin{definition}[Factorization encoding]\label{def:factorization-encoding}
For $\sigma\in\orbreps$, $t_L\in\Transversal{\sigma}$, and $h\in\Transversal{N}$, let $a\in\Inertia{\sigma}$ and $t_R\in\Transversal{\sigma}$ be the unique elements satisfying
\begin{equation}
    h=t_Lat_R\inv.
\end{equation}
For fixed $\sigma$ and $t_L$, define the valid output set
\begin{equation}
    \mathcal F_{\sigma,t_L}
    :=
    \left\{
    (a,t_R)\in\Inertia{\sigma}\times\Transversal{\sigma}
    \ \middle|\
    t_Lat_R\inv\in\Transversal{N}
    \right\}.
\end{equation}
The controlled classical encoding is given by
\begin{equation}\label{eq:factorization-encoding-action}
    \Enc_{\Transversal{N},\mathcal F_{\sigma,t_L}}
    \bigl(\ket{\sigma}\otimes\ket{t_L}\otimes\ket{h}\bigr)
    =
    \ket{\sigma}\otimes\ket{t_L}\otimes\ket{a}\otimes\ket{t_R}.
\end{equation}
For fixed $\sigma$ and $t_L$, the map $h\mapsto(a,t_R)$ is a bijection onto $\mathcal F_{\sigma,t_L}$.
\end{definition}

\begin{observation}\label{obs:factor-coset-suffices}
Fix $\sigma\in\orbreps$ and $t_L,t_R\in\Transversal{\sigma}$. Then for any $h\in\Transversal{N}$, the unique element $a\in \Inertia{\sigma}$ such that
\begin{equation}
h=t_Lat_R\inv,
\end{equation}
is fully determined by its coset $aN\in\Little{\sigma}$.

Indeed, suppose we had
\begin{equation}
    h_1=t_La_1t_R\inv,
    \qquad
    h_2=t_La_2t_R\inv,
    \qquad \text{with}\quad
    a_1N=a_2N,\quad h_1,h_2\in\Transversal{N}.
    \end{equation}
    Then
\begin{equation}
    h_1N
    =
    t_La_1t_R\inv N
    =
    t_La_2t_R\inv N
    =
    h_2N.
\end{equation}
Since $\Transversal{N}$ contains one representative for each coset in $G/N$, it follows that $h_1=h_2$, and hence $a_1=a_2$.
\end{observation}

We can now properly define the general QFT algorithm. A circuit representation is shown in \cref{fig: little group qft}. The circuit resembles the general induction strategy of \cref{fig: theorem first induction full circuit}. In particular, both apply an encoding, perform the QFT over a subgroup, and then implement an induction step on the transversal elements controlled on the smaller QFT outcome. The addition in the little group method is some extra processing of the little group Fourier basis.

\begin{algorithm}[H]
\caption{Little group QFT for $N\unlhd G$}
\label{alg:little-group-qft}
\small
\begin{enumerate}
    \item Apply the right-coset version of $\Enc_{N,G}$ to the group register:
    \begin{equation*}
    \ket{g}
    \longmapsto
    \ket{n}\otimes\ket{h},
    \qquad g=nh,\qquad
    n\in N,\quad h\in\Transversal{N}.
    \end{equation*}

    \item Apply $\Fou_N$ to the $\C[N]$ register.
    Denote its output registers by $\ket{\la}\ket{i_\la}\ket{k_{\la}^\vee}$, where $\la\in\hatN$ and $i,k=1,\dots,d_\la$.

    \item Apply $\Enc_{\hatN,\mathcal O_N}$ to the irrep-label register. If $\la=\act{t_L}\sigma$, then
    \begin{equation*}
    \ket{\la}
    \longmapsto
    \ket{\sigma}\otimes\ket{t_L}.
    \end{equation*}

    \item Apply the two orbit transport gates, both controlled on $\sigma$ and $t_L$:
    \begin{align*}
        U_{\mathrm{orb}}^L:
        \quad
        \ket{i_\la}
        &\longmapsto
        U_{t_L,\sigma}^\dagger\ket{i_\la},
        \\
        U_{\mathrm{orb}}^R:
        \quad
        \ket{{k_{\la}^\vee}}
        &\longmapsto
        {U_{t_L,\sigma}\tp\ket{k_{\la}^\vee}.}
    \end{align*}
    Equivalently, the two applications give
    \begin{equation*}
    \ket{\sigma}\otimes\ket{t_L}
    \otimes\ket{i_\la}\otimes\ket{{k_{\la}^\vee}}
    \longmapsto
    \ket{\sigma}\otimes\ket{t_L}
    \otimes U_{t_L,\sigma}^\dagger\ket{i_\la}
    \otimes U_{t_L,\sigma}\tp\ket{{k_{\la}^\vee}}.
    \end{equation*}
    {The output registers use the bases $\{\ket{i_\sigma}\}$ and $\{\ket{k_\sigma^\vee}\}$, with $i,k=1,\dots,d_\sigma=d_\la$.}
    \item Apply $\Enc_{\Transversal{N},\mathcal F_{\sigma,t_L}}$ to the $h$-register, controlling on the $\sigma$- and $t_L$-registers:
    \begin{equation*}
    \ket{h}
    \longmapsto
    \ket{a}\otimes\ket{t_R},
    \qquad
    h=t_L a t_R\inv,
    \qquad
    a\in\Inertia{\sigma},\quad t_R\in\Transversal{\sigma}.
    \end{equation*}

    \item {Apply $(\ExtRep_{\sigma})^\vee(a\inv)$ to the right Fourier register, controlling on the $\sigma$- and $a$-registers.}

    \item Apply $\Enc_{N,\Inertia{\sigma}}$ to the $a$-register, controlling on the $\sigma$-register:
    \begin{equation*}
    \ket{a}
    \longmapsto
    \ket{h_\sigma}\otimes\ket{d},
    \qquad
    a=h_\sigma d\qquad h_\sigma\in \Inertia{\sigma}\cap \Transversal{N},\ d\in N.
    \end{equation*}
    By \cref{obs:factor-coset-suffices}, the $d$-register contains no independent information and can be uncomputed reversibly. Relabel the remaining representative register by its coset, $\ket{h_\sigma}\equiv\ket{aN}$, where $aN=h_\sigma N\in\Little{\sigma}$.

    \item Apply $\Fou_{\Little{\sigma}}$ to the $aN$-register, controlled on the $\sigma$-register, again using \cref{eq: Fourier transform}.
    {Denote its output registers by $\ket{\eta'}\ket{j_{\eta'}}\ket{l_{\eta'}^\vee}$, where $\eta'\in\widehat{\Little{\sigma}}$ and $j,l=1,\dots,d_{\eta'}$.}
\end{enumerate}
\end{algorithm}

\begin{figure}[H]
\centering
\resizebox{\textwidth}{!}{\begin{quantikz}[wire types={n,n,n,q,n,n,n,n}]
&\gate[8]{\Enc_{N,G}}&\gate[3]{\Fou_N}&
\wire[l][1]["\ket{i_\la}"{above,pos=0.2}]{q}\setwiretype{q}
&\gate{U_{\mathrm{orb}}^L}&&&&&&&\rstick{$\ket{i_\sigma}$}
\\
&&\wire[l][1]["\ket{n}"{above,pos=0.5}]{q}\setwiretype{q}
&\wire[l][1]["\ket{{k_{\la}^\vee}}"{above,pos=0.2}]{q}&
&\gate{U_{\mathrm{orb}}^R}&&\gate{{(\ExtRep_{\sigma})^\vee}(a\inv)}&&&&\rstick{$\ket{{k_{\sigma}^\vee}}$}
\\
&&&\gate[2]{\Enc_{\hatN,\mathcal O_N}}\wire[l][1]["\ket{\la}"{above}]{q}
&\ctrl{-2}\wire[l][1]["\ket{\sigma}"{above,pos=0.5}]{q}\setwiretype{q}
&\ctrl{-1}&\ctrl{3}&\ctrl{-1}&\ctrl{2}&&\ctrl{2}&\rstick{$\ket{\sigma}$}
\\
\lstick{$\ket{g}$}&&\setwiretype{n}&&
\ctrl{-3}\wire[l][1]["\ket{t_L}"{above,pos=0.5}]{q}\setwiretype{q}
&\ctrl{-2}&\ctrl{2}&&&&&\rstick{$\ket{t_L}$}
\\
&&&&&&&&\gate[2]{\Enc_{N,\Inertia{\sigma}}}
&\wire[l][1]["\ket{aN}"{above,pos=0.4}]{q}\setwiretype{q}
&\gate[3]{\Fou_{\Little{\sigma}}}&\rstick{$\ket{j_{\eta'}}$}
\\
&&&&&&\gate[3]{\Enc_{\Transversal{N},\mathcal F_{\sigma,t_L}}}
&\ctrl{-4}\wire[l][1]["\ket{a}"{above,pos=0.8}]{q}\setwiretype{q}
&&\ground{}\setwiretype{q}&\setwiretype{n}
&\setwiretype{q}\rstick{$\ket{{l_{\eta'}^\vee}}$}
\\
&&\wire[l][1]["\ket{h}"{above,pos=0.5}]{q}\setwiretype{q}
&&&&&\setwiretype{n}&&&&\rstick{$\ket{\eta'}$}\setwiretype{q}
\\
&&&&&&&\setwiretype{q}&&&&\rstick{$\ket{t_R}$}
\end{quantikz}}
\caption{The circuit for $\Fou_G$, the ``little group'' QFT for $N\unlhd G$. The $\Enc$ unitaries are, up to relabeling, classical permutation matrices.}
\label{fig: little group qft}
\end{figure}
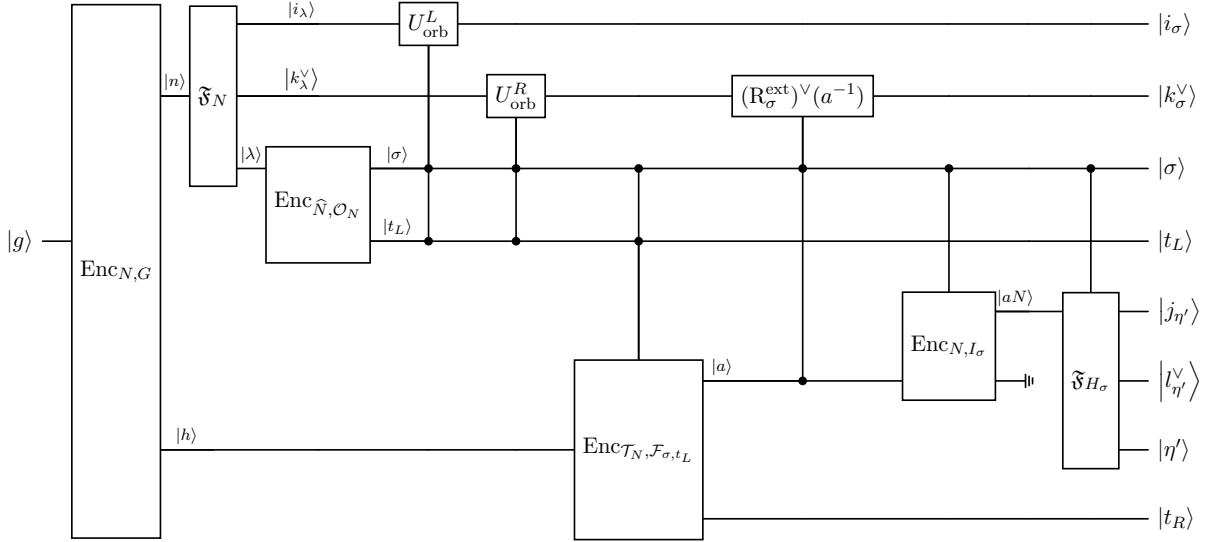

\begin{theorem}\label{thm:little-group-qft-correctness}
\cref{alg:little-group-qft} implements the Fourier transform from the group-element basis of $\C[G]$ to the ``little group'' Fourier basis given in \cref{def:little-group-basis}.
\end{theorem}

\begin{proof}
We need to show that the left and right regular actions of $G$ act correctly on the basis we create with the circuit and that the basis states are correctly normalized and produce the correct matrix coefficients.

The two Fourier transforms occurring in the circuit give the required normalization. Indeed, for a representation label $(\sigma, \eta') \in \hatG$,
\begin{equation*}
d_{(\sigma,\eta')}
=[G:\Inertia{\sigma}]d_\sigma d_{\eta'},
\qquad
|G|=[G:\Inertia{\sigma}]|N||\Little{\sigma}|,
\end{equation*}
and hence
\begin{equation*}
\sqrt{\frac{d_\sigma}{|N|}}
\sqrt{\frac{d_{\eta'}}{|\Little{\sigma}|}}
=
\sqrt{\frac{d_{(\sigma,\eta')}}{|G|}}.
\end{equation*}
Moreover, if the factorization steps give $g t_R=t_Lb$ with $b\in\Inertia{\sigma}$, then the two Fourier transforms and the basis changes on the left and right $N$ Fourier registers produce the matrix coefficient
\begin{equation*}
\bigl[\ExtRep_\sigma(b)\bigr]_{k}^i
\bigl[\Rep_{\eta'}(bN)\bigr]_{l}^j,
\end{equation*}
{which is exactly the $(t_L,i_\sigma,j_{\eta'};t_R,k_{\sigma}^\vee,l_{\eta'}^\vee)$ matrix coefficient of}
\begin{equation*}
\Rep_{(\sigma,\eta')}
=\Ind_{\Inertia{\sigma}}^G
\bigl(\ExtRep_\sigma\otimes\InfRep_{\eta'}\bigr).
\end{equation*}
It remains to verify that the left and right actions transform these registers according to their representation-theoretic meaning.  For this purpose, we transport the left and right regular actions through the circuit one step at a time, recording after each step the current basis and determining how the regular representations $L(x)$ and $R(x)$ act at each step.  For these calculations, we refer to the full proof in \cref{app: proofs of clifford}.

\end{proof}

\section{Application 2: QFT for \texorpdfstring{$F\wr S_n$}{F wr Sn}}\label{sec: wreath}

As an application of the little group method developed in \cref{sec: clifford theory}, we construct a QFT circuit for wreath products of the form $F\wr S_n$.
In particular, following Refs.~\cite[Sections~2.4 and~2.6]{CecchWreathProducts} and~\cite[Section~2.10]{CecchFinitegroup}, we describe all Clifford-theoretic data for wreath products $F\wr S_n$ and use them to implement a specialization of \cref{alg:little-group-qft}.
This leads to the following main result.

\begin{theorem}[QFT over a wreath product]
\label{thm:wreath-main}
Let $F$ be a finite group, let $n\geq 1$, and let $\varepsilon>0$. There is a
uniform quantum circuit implementing the Fourier transform over $F\wr S_n$,
in the basis of \cref{eq:wreath-target-fourier-basis}, with operator-norm
error at most $\varepsilon$, gate count
\begin{equation}
    C_{F\wr S_n}(\varepsilon)
    \leq
    nC_F\!\left(\frac{\varepsilon}{2n}\right)
    +
    \widetilde O\!\left(n^3+n^2L_F\right),
    \label{eq:wreath-main-gate-bound}
\end{equation}
and depth
\begin{equation}
   D_{F\wr S_n}(\varepsilon)
\leq
D_F\!\left(\frac{\varepsilon}{2n}\right)
+
\widetilde O\!\left(n^3+n^2L_F\right).
\end{equation}
Here
\[
    L_F
    :=
    \left\lceil
        \log_2
        \max\bigl\{|\widehat F|,d_{\max}(F)\bigr\}
    \right\rceil,
    \qquad
    d_{\max}(F)
    :=
    \max_{\rho\in\widehat F}d_\rho,
\]
and $C_F(\delta)$ denotes the gate count of a uniform implementation of
$\Fou_F$ with operator-norm error at most $\delta$. The notation
$\widetilde O$ suppresses polylogarithmic factors in $n$, $L_F$, and
$1/\varepsilon$.
\end{theorem}

The general construction of Ref.~\cite{GenQFT} gives an efficient QFT for $F\wr S_n$ under the assumption $|F|=\operatorname{poly}(n)$. The bound above instead applies to families of groups $F$ that admit efficient QFTs even when their order is not polynomial in $n$, thereby significantly extending the previous results.

As an example, for $F=S_m$, the definition gives $L_{S_m}\leq\lceil\log_2(m!)\rceil=O(m\log m)$, while Ref.~\cite{BruinsmaGrinkoOzols2026} gives the gate bound $C_{S_m}(\delta)=\widetilde O(m^3)$.
Hence
\begin{equation}\label{eq:iterated-symmetric-wreath-bound}
    C_{S_m\wr S_n}(\varepsilon)
    =
    \widetilde O\!\left(nm^3+n^3+n^2m\right),
\end{equation}
where the suppressed factors also depend polylogarithmically on $1/\varepsilon$. In particular, the complexity is polynomial in the two independent parameters $m$ and $n$.

We keep the notation of \cref{sec: clifford theory} whenever possible.
For the symmetric-group Fourier transform and irreducible modules we use the Young--Yamanouchi path labels and the relative path encoding of Ref.~\cite{BruinsmaGrinkoOzols2026}.

\subsection{Wreath products and their Clifford-theoretic data}\label{subsec:wreath-conventions}
\begin{definition}[Wreath product]\label{def:wreath-product}
Let $F$ be a finite group. The symmetric group $S_n$ acts on $F^n$ by permuting coordinates:
\begin{equation}\label{eq:wreath-action-on-base}
    \pi\cdot(f_1,\dots,f_n)
    :=
    \bigl(f_{\pi\inv(1)},\dots,f_{\pi\inv(n)}\bigr)
\end{equation}
for all $\pi \in S_n$ and $f_1,\dotsc,f_n \in F$.
The \emph{wreath product} $F\wr S_n$ is then
\begin{equation}\label{eq:general-wreath-product}
    F\wr S_n:=F^n\rtimes S_n,
\end{equation}
with multiplication and inversion given by
\begin{align}
    (\mathbf f,\pi)(\mathbf f',\pi')
    &=
    \bigl(\mathbf f(\pi\cdot\mathbf f'),\pi\pi'\bigr),
    \label{eq:general-wreath-multiplication}
    \\
    (\mathbf f,\pi)\inv
    &=
    \bigl(\pi\inv\cdot\mathbf f\inv,\pi\inv\bigr).
    \label{eq:general-wreath-inverse}
\end{align}
Multiplication and inversion in $F^n$ are coordinatewise.
\end{definition}

We denote $G:=F\wr S_n$, and study its Clifford theory via the normal subgroup
\begin{equation}\label{eq:wreath-base-group}
    N:=F^n\times\{\idSn\}\cong F^n,
\end{equation}
 which we call the \emph{base group}. We naturally have $G/N\cong S_n$, and choose as a transversal for $N$ in $G$ the subgroup
\begin{equation}\label{eq:wreath-transversal-base}
    \Transversal{N}
    :=
    \{(\idN,\pi):\pi\in S_n\}.
\end{equation}
Every element in $G$ then has the two factorizations
\begin{equation}\label{eq:wreath-base-factorization}
    (\mathbf f,\pi)
    =
    (\mathbf f,\idSn)(\idN,\pi)
    =
    (\idN,\pi)(\pi\inv\cdot\mathbf f,\idSn).
\end{equation}
For later use, we observe that conjugation of the base group by a pure permutation satisfies
\begin{equation}\label{eq:wreath-conjugation-pure-permutation}
    (\idN,\pi)\inv(\mathbf f,\idSn)(\idN,\pi)
    =
    (\pi\inv\cdot\mathbf f,\idSn).
\end{equation}
As $N\cong F^n$, its irreducible representations are tensor products of irreducible representations of $F$, so that $\hatN=\widehat F^{\,n}$.
For
\begin{equation*}
    \vla=(\lambda_1,\dots,\lambda_n)\in\widehat F^{\,n},
\end{equation*}
set
\begin{align}
    V_\vla
    &:=
    \bigotimes_{r=1}^nV_{\lambda_r},
    \label{eq:wreath-base-irrep}
    \\
    \Rep_\vla(\mathbf f)
    &:=
    \bigotimes_{r=1}^n\Rep_{\lambda_r}(f_r),
    \qquad
    \mathbf f=(f_1,\dots,f_n)\in F^n.
    \label{eq:wreath-base-representation}
\end{align}
We use the tensor-product bases
\begin{equation}\label{eq:wreath-base-index-bases}
    \ket{i_\vla}
    :=
    \bigotimes_{r=1}^n\ket{i_{\lambda_r}},
    \qquad
    \ket{{k_{\vla}^\vee}}
    :=
    \bigotimes_{r=1}^n\ket{{k_{\lambda_r}^\vee}}.
\end{equation}

Conjugation by $(\mathbf f,\idSn)\in F^n$ induces an inner automorphism on each factor of $\Rep_\vla$, and hence fixes the equivalence class label $\vla$. Thus the $G$-action on $\hatN$ factors through $S_n$. By \cref{eq:wreath-conjugation-pure-permutation}, this action is
\begin{equation}\label{eq:wreath-action-on-dual}
    \act{\pi}\vla
    =
    \bigl(
        \lambda_{\pi\inv(1)},\dots,
        \lambda_{\pi\inv(n)}
    \bigr).
\end{equation}

Recall from \cref{def: representatives eq classes} that $\orbreps$ is the chosen set of representatives for the $G$-orbits in $\hatN$, and that $\repO(\vla)$ is the representative of the orbit containing $\vla$. In the present setting, fix a total order on $\widehat F$ and the induced lexicographic order on $\widehat F^{\,n}$, and choose $\repO(\vla)$ to be the nonincreasingly sorted tuple in the orbit of $\vla$. Thus every $\vsig\in\orbreps$ has the form
\begin{align}\label{eq:wreath-multiplicity-profile}
    &\vsig
    =
    \bigl({\sigma_1}^{n_1},\dots,{\sigma_\ell}^{n_\ell}\bigr):=
    \bigl(\underbrace{{\sigma_1},\dots,{\sigma_1}}_{n_1},\underbrace{{\sigma_2},\dots,{\sigma_2}}_{n_2},\dots,\underbrace{{\sigma_\ell}\dots{\sigma_\ell}}_{n_\ell}\bigr)\in \widehat F\,^n,\\
    &\text{with}
    \qquad
    {\sigma_1}>{\sigma_2}>\cdots>{\sigma_\ell},
    \qquad
    \sum_{r=1}^{\ell}n_r=n.
\end{align}

\begin{definition}[Stable orbit transversal]\label{def:wreath-representatives}
Let $\vsig\in\orbreps$ and $\vla\in\orb_G(\vsig)$. For every $\rho\in\widehat F$, match the $s$-th occurrence of $\rho$ in $\vsig$ with the $s$-th occurrence of $\rho$ in $\vla$. Denote the resulting stable permutation by $\pi_{\vsig,\vla}$, so that
\begin{equation}\label{eq:wreath-stable-permutation-action}
    \act{\pi_{\vsig,\vla}}\vsig=\vla.
\end{equation}
Set
\begin{equation}\label{eq:wreath-transversal-element}
    \TransversalElt{\vsig,\vla}
    :=
    (\idN,\pi_{\vsig,\vla})
\end{equation}
and
\begin{equation}\label{eq:wreath-transversal-inertia}
    \Transversal{\vsig}
    :=
    \left\{
        \TransversalElt{\vsig,\vla}:
        \vla\in\orb_G(\vsig)
    \right\}.
\end{equation}
\end{definition}
By \cref{lem: transv inertia}, $\Transversal{\vsig}$ is a left transversal for $\Inertia{\vsig}$ in $G$.

\begin{definition}[Little group and inertia group]\label{def:wreath-little-inertia}
For $\vla\in\widehat F^{\,n}$ define the corresponding little group by
\begin{equation}\label{eq:wreath-permutation-stabilizer}
    \Little{\vla}
    :=
    \Stab_{S_n}(\vla).
\end{equation}
The full inertia group for $\vla$ is
\begin{equation}\label{eq:wreath-inertia-group}
    \Inertia{\vla}
    =
    N\rtimes\Little{\vla}.
\end{equation}
\end{definition}

\begin{observation}
\label{obs:wreath-young-subgroup}
For the sorted representative in \cref{eq:wreath-multiplicity-profile}, the stabilizer is the Young subgroup
\begin{equation}\label{eq:wreath-little-group-young}
    \Little{\vsig}
    =
    S_{n_1}\times\cdots\times S_{n_\ell},
\end{equation}
where each factor permutes the positions within one consecutive block of equal labels.
\end{observation}

\subsection{Applying the little group method to wreath products}
In order to apply the little group method to $G$, we need to show that each element $\Rep_\vla$ for $\vla \in \hatN$ extends to a representation of $\Inertia{\vla}$.
In order to do so, we begin by noticing that the orbit-transport unitaries of \cref{def:orbit-transport-gates} have a particularly simple form when specialized to wreath products. For $\vla\in\widehat F^{\,n}$ and $t=(\idN,\pi)$, define
\begin{equation}\label{eq:wreath-U-tensor-definition}
    U_{t,\vla}:V_\vla\longrightarrow V_{\act{\pi}\vla}
\end{equation}
on elementary tensors by
\begin{equation}\label{eq:wreath-U-tensor-action}
    U_{t,\vla}
    \left(v_1\otimes\cdots\otimes v_n\right)
    :=
    v_{\pi\inv(1)}\otimes\cdots\otimes v_{\pi\inv(n)}.
\end{equation}
Thus $U_{t,\vla}$ is the unitary that permutes the tensor factors according to $\pi$.
It is precisely an orbit-transport unitary, as defined in \cref{def:orbit-transport-gates}. In fact, if $s=(\idN,\tau)$, then the following properties hold:
\begin{align}
    U_{t,\act{\tau}\vla}U_{s,\vla}
    &=U_{ts,\vla},
    \label{eq:wreath-U-tensor-composition}
    \\
    U_{t,\vla}^\dagger
    &=U_{t\inv,\act{\pi}\vla},
    \label{eq:wreath-U-tensor-adjoint}
    \\
    U_{t,\vla}\Rep_\vla(\mathbf f)U_{t,\vla}^\dagger
    &=\Rep_{\act{\pi}\vla}(\pi\cdot\mathbf f).
    \label{eq:wreath-U-tensor-covariance}
\end{align}

When $\pi\in\Little{\vla}$, we abbreviate
\begin{equation}\label{eq:wreath-little-transport-abbreviation}
    U_{\pi,\vla}:=U_{(\idN,\pi),\vla}.
\end{equation}
For $(\mathbf f,\pi)\in\Inertia{\vla}$, define
\begin{equation}\label{eq:general-wreath-extension}
    \ExtRep_\vla(\mathbf f,\pi)
    :=
    \Rep_\vla(\mathbf f)U_{(\idN,\pi),\vla}.
\end{equation}
Here $\pi\in\Little{\vla}$, so $\act{\pi}\vla=\vla$ and the transport unitary is an endomorphism of $V_\vla$.

\begin{lemma}\label{lem:general-wreath-extension}
The map $\ExtRep_{\vla}:\Inertia{\vla}\rightarrow \U(V_\vla)$ in \cref{eq:general-wreath-extension} is a unitary representation of $\Inertia{\vla}$ that extends $\Rep_\vla$.
\end{lemma}

\begin{proof}
For $\pi,\pi'\in\Little{\vla}$, \cref{eq:wreath-U-tensor-composition,eq:wreath-U-tensor-covariance} give
\begin{equation*}
    \ExtRep_\vla(\mathbf f,\pi)\ExtRep_\vla(\mathbf f',\pi')
    =
    \Rep_\vla\bigl(\mathbf f(\pi\cdot\mathbf f')\bigr)
    U_{(\idN,\pi\pi'),\vla},
\end{equation*}
which agrees with \cref{eq:general-wreath-multiplication}. Restriction to $N$ is obtained by setting $\pi=\idSn$.
\end{proof}

Therefore every base group irreducible representation has an extension to its inertia group, and we can apply the little group method to $G$.

Fix now
\begin{equation*}
    \vsig=({\sigma_1}^{n_1},\dots,{\sigma_\ell}^{n_\ell})\in\orbreps.
\end{equation*}
By \cref{obs:wreath-young-subgroup}, $\Little{\vsig}=S_{n_1}\times\cdots\times S_{n_\ell}$.
Following the representation theory of symmetric groups, write
\begin{equation*}
    \widehat{S_m}=\{\nu:\nu\pt m\},
\end{equation*}
and, for $\nu\pt m$, let $\Rep_\nu:S_m\to\mathrm U(V_\nu)$ be the unitary irreducible representation on the corresponding simple $\C[S_m]$-module $V_\nu$.
Consequently, an irrep of $\Little{\vsig}$ is labeled by a tuple of partitions

\begin{equation}\label{eq:wreath-little-irrep-label}
    \vnu=(\nu_1,\dots,\nu_\ell),
    \qquad
    \nu_r\pt n_r\quad(r=1,\dots,\ell),
\end{equation}
and acts on
\begin{equation}\label{eq:wreath-little-irrep-space}
    V_\vnu
    :=
    \bigotimes_{r=1}^{\ell}V_{\nu_r}
\end{equation}
through
\begin{equation}\label{eq:wreath-little-irrep}
    \Rep_\vnu
    :=
    \bigotimes_{r=1}^{\ell}\Rep_{\nu_r},\text{ by } \quad
       \Rep_\vnu(\alpha_1,\dots,\alpha_\ell)
    =
    \bigotimes_{r=1}^{\ell}\Rep_{\nu_r}(\alpha_r).
\end{equation}
Following Ref.~\cite{BruinsmaGrinkoOzols2026}, for each $r$, we choose the Young--Yamanouchi path basis $\{\ket{Q_r}:Q_r\in\mathcal P(\nu_r)\}$ of $V_{\nu_r}$ and its dual basis $\{\ket{P_r^\vee}:P_r\in\mathcal P(\nu_r)\}$ of $V_{\nu_r}^\vee$.
Their tensor-product bases are denoted by
\begin{equation}\label{eq:wreath-path-index-identification}
    \ket{\mathbf Q}
    :=
    \bigotimes_{r=1}^{\ell}\ket{Q_r},
    \qquad
    \ket{\mathbf P^\vee}
    :=
    \bigotimes_{r=1}^{\ell}\ket{P_r^\vee}.
\end{equation}
We are now ready to provide the full classification of irreducible representations for $G$.
\begin{theorem}[Irreducible representations of $F\wr S_n$]
\label{thm:wreath-irreps}
Every irreducible representation of $F\wr S_n$ can be labeled by a unique pair $(\vsig,\vnu)$, where
\begin{equation}
    \vsig=({\sigma_1}^{n_1},\dots,{\sigma_\ell}^{n_\ell})\in\orbreps,
    \qquad
    \vnu=(\nu_1,\dots,\nu_\ell),
    \qquad
    \nu_r\pt n_r.
\end{equation}
A representative for the irreducible representation labeled by $(\vsig,\vnu)$ is
\begin{equation}\label{eq:wreath-irrep-classification}
    \Rep_{(\vsig,\vnu)}
    \cong
    \Ind_{\Inertia{\vsig}}^G
    \left(
        \ExtRep_\vsig\otimes\InfRep_\vnu
    \right).
\end{equation}
\end{theorem}

\begin{proof}
Apply \cref{thm: little group} using \cref{eq:wreath-inertia-group,eq:wreath-little-group-young,lem:general-wreath-extension}.
\end{proof}

The induced basis from \cref{thm: little group} gives the target wreath-product Fourier basis
\begin{equation}\label{eq:wreath-target-fourier-basis}
    \ket{\vsig,\vnu}
    \otimes
    \ket{t_L,i_\vsig,\mathbf Q}
    \otimes
    \ket{t_R,{k_{\vsig}^\vee},\mathbf P^\vee}.
\end{equation}
\Cref{tab:wreath-clifford-dictionary} gives a dictionary of the Clifford-theoretic data developed above.

\begin{table}[ht]
\centering
\small
\setlength{\tabcolsep}{6pt}
\renewcommand{\arraystretch}{1.15}
\begin{tabular}{@{}>{\raggedright\arraybackslash}p{0.28\linewidth}
                  >{\raggedright\arraybackslash}p{\dimexpr0.72\linewidth-2\tabcolsep\relax}@{}}
\toprule
\textbf{General notation in \cref{sec: clifford theory}} & \textbf{Wreath product notation}\\
\midrule
$G$, $N$, $G/N$
& $F^n\rtimes S_n$, $F^n$, $S_n$ \\
$\la\in\hatN$
& $\vla=(\lambda_1,\dots,\lambda_n)\in\widehat F^{\,n}$ \\
$V_\la$, $\Rep_\la$
& $V_\vla=\bigotimes_rV_{\lambda_r}$,
  \quad$\Rep_\vla=\bigotimes_r\Rep_{\lambda_r}$ \\
$i_\la$, ${k_{\la}^\vee}$
& $i_\vla=(i_{\lambda_1},\dots,i_{\lambda_n})$,
  \quad${k_{\vla}^\vee}=({k_{\lambda_1}^\vee},\dots,{k_{\lambda_n}^\vee})$ \\
$\act{g}\la$
& $\act{\pi}\vla=(\lambda_{\pi\inv(1)},\dots,\lambda_{\pi\inv(n)})$ \\
$\sigma\in\orbreps$
&

$\vsig=({\sigma_1}^{n_1},\sigma_2^{n_2},\dots,{\sigma_\ell}^{n_\ell})\in \widehat F^{n},\quad{\sigma_1}>{\sigma_2}>\cdots>{\sigma_\ell},\quad\sum_{r=1}^{\ell}n_r=n$ \\
$\Transversal{\sigma}$
& $\{(\idN,\pi_{\vsig,\vla}):\vla\in\orb_G(\vsig)\}$ \\
$t_L$, $t_R$
& $(\idN,\pi_L)$, $(\idN,\pi_R)$ \\
$\Inertia{\sigma}$
& $F^n\rtimes\Little{\vsig}$ \\
$\Little{\sigma}$
& $S_{n_1}\times\cdots\times S_{n_\ell}$ \\
$a$, $aN$
& $(\idN,\alpha)$, $\alpha\in\Little{\vsig}$ \\
$U_{t,\sigma}$, $U_{t,\sigma}^\dagger$, $U_{t,\sigma}\tp$
& $U_{(\idN,\pi),\vsig}
    \left(v_1\otimes\cdots\otimes v_n\right)
    :=
    v_{\pi\inv(1)}\otimes\cdots\otimes v_{\pi\inv(n)}$ \\
$\ExtRep_\sigma$
& $\ExtRep_\vsig(\mathbf f,\pi)
  =\Rep_\vsig(\mathbf f)U_{(\idN,\pi),\vsig}$ \\
$\eta'\in\widehat{\Little{\sigma}}$
& $\vnu=(\nu_1,\dots,\nu_\ell)\in \widehat{S_{n_1}}\times\cdots\times\widehat{S_{n_\ell}}$, \quad$\nu_r\pt n_r$ \\
$V_{(\sigma,\eta')}$
& $V_{(\vsig,\vnu)}$ \\
$j_{\eta'}$, ${l_{\eta'}^\vee}$
& $\mathbf Q=(Q_1,\dots,Q_\ell)$,
  \quad$\mathbf P^\vee=(P_1^\vee,\dots,P_\ell^\vee)$ \\
\bottomrule
\end{tabular}
\caption{Summary of the Clifford-theoretic data for $F\wr S_n$.}
\label{tab:wreath-clifford-dictionary}
\end{table}

\subsection{Analysis of the specialized wreath product circuit}\label{subsec:wreath-qft-circuit}

We now consider a concrete specialization of \cref{alg:little-group-qft} to the wreath product case.
We encode permutations as disjoint-cycle lists. Write
\begin{equation}\label{eq:wreath-disjoint-cycle-encoding}
    \pi=c_1c_2\cdots c_s,
    \qquad
    c_r=(c_{r,1},\dots,c_{r,m_r}),
\end{equation}
where fixed points are omitted, the smallest entry of each cycle is written first, and the cycles are ordered by increasing smallest entry. The identity is represented by the empty cycle list.

Inversion, multiplication, conversion to or from a value table and factorization into adjacent transpositions can each be implemented with $\widetilde O(n^2)$ one- and two-qubit gates.  We use this common upper bound for all elementary permutation operations and do not explicitly compute the overall cost, as it does not impact the final asymptotic bound.
An input group element is represented by
\begin{equation}\label{eq:wreath-input-encoding}
    \ket{\mathbf f}\otimes\ket{\pi}
    =
    \ket{f_1,\dots,f_n}\otimes\ket{\pi}.
\end{equation}
The complete gate specialization is summarized in \cref{tab:wreath-circuit-dictionary}.

\begin{table}[ht]
\centering
\small
\setlength{\tabcolsep}{6pt}
\renewcommand{\arraystretch}{1.15}
\begin{tabular}{@{}>{\raggedleft\arraybackslash}p{0.07\linewidth}
                  >{\raggedright\arraybackslash}p{0.30\linewidth}
                  >{\raggedright\arraybackslash}p{\dimexpr0.63\linewidth-4\tabcolsep\relax}@{}}
\toprule
\textbf{Gate} & \textbf{General operation in \cref{alg:little-group-qft}} & \textbf{Wreath product operation} \\
\midrule
1 & $\Enc_{N,G}$ & - \\
2 & $\Fou_N$ & $\Fou_F^{\otimes n}$ \\
3 & $\Enc_{\hatN,\mathcal O_N}$
& $\vla\mapsto\bigl(\vsig,t_L=(\idN,\pi_L)\bigr)$ \\
4& $U_{\mathrm{orb}}^L,\ U_{\mathrm{orb}}^R$
& $U_{t_L,\vsig}^\dagger,\quad U_{t_L,\vsig}\tp$ \\
5 & $\Enc_{\Transversal{N},\mathcal F_{\vsig,t_L}}$
& $\pi\mapsto\bigl(\alpha,t_R=(\idN,\pi_R)\bigr)$,
  \quad$\alpha=\pi_L\inv\pi\pi_R$ \\
6 & $(\ExtRep_{\vsig})^\vee((\idN,\alpha)\inv)$
& $U_{\alpha,\vsig}\tp$ \\
7 & $\Enc_{N,\Inertia{\sigma}}$ & - \\
8 & $\Fou_{\Little{\vsig}}$ & $\bigotimes_{r=1}^{\ell}\Fou_{S_{n_r}}$ \\
\bottomrule
\end{tabular}
\caption{Specialization of the little group QFT to $F\wr S_n$.
The encoding steps $1$ and $7$ of \cref{alg:little-group-qft} are trivial in the chosen wreath product encoding.}
\label{tab:wreath-circuit-dictionary}
\end{table}
For our later complexity analysis, we also develop an explicit construction for a coherently controlled QFT for the little group, which corresponds to step $8$ in \cref{alg:little-group-qft}.

\begin{lemma}[Coherent Young-subgroup Fourier transform]
\label{lem:wreath-coherent-young-qft}
For
\[
    \vsig=({\sigma_1}^{n_1},\sigma_2^{n_2},\dots,{\sigma_\ell}^{n_\ell})\in\orbreps,
    \qquad
    \Little{\vsig}=S_{n_1}\times\cdots\times S_{n_\ell},
\]
the transform
\[
    \Fou_{\Little{\vsig}}
    =
    \Fou_{S_{n_1}}\otimes\cdots\otimes\Fou_{S_{n_\ell}}
\]
can be implemented coherently, controlled on $\ket{\vsig}$, with
operator-norm error at most $\delta$ using $\widetilde O(n^3)$ one- and
two-qubit gates.
\end{lemma}
\begin{proof}
We apply the Beals construction of Ref.~\cite{BruinsmaGrinkoOzols2026} to each symmetric-group factor.
To do this coherently for varying factor sizes, we first arrange the registers in decreasing order of size, then perform induction from $S_{k-1}$ to $S_k$ on the factors of size at least $k$, for $k=2,\dots,n$.
At most $\lfloor n/k\rfloor$ factors need this operation, which gives the gate bound above.
Finally, we restore the original register order.

\begin{enumerate}
    \item \textbf{Prepare the registers.}
    Controlling on $\vsig$ and utilizing the disjoint-cycle encoding of the input permutation, store each factor size $n_r$ with its permutation register, allowing $\widetilde O(n)$ qubits for each one. Add padding registers to obtain $n$ total registers.
    Convert each input permutation from its disjoint-cycle encoding to its successive left-coset factors encoding along $S_1 < \cdots < S_{n_r}$, using the Beals transversals of Ref.~\cite{BruinsmaGrinkoOzols2026}.
    This costs at most $\widetilde O(n^2)$ gates per factor, hence $\widetilde O(n^3)$ gates in total.

    \item \textbf{Order the registers by size.}
    Sort the registers by decreasing factor size using a fixed sequence of $O(n\log^2 n)$ comparisons.
    Apply the resulting sorting permutation on the factors and store it in an additional register.
    Each exchange acts on $\widetilde O(n)$ qubits, so the total cost is $\widetilde O(n^2)$ gates.

    \item \textbf{Apply the induction transforms.}
    For $k=2,\dots,n$, apply the induction transform $\IndTrans[k]$ for $S_{k-1} < S_k$ to each of the first $\lfloor n/k\rfloor$ registers, conditioned on its stored factor size being at least $k$.
    Since $\sum_r n_r=n$, all factors of size at least $k$ lie among these positions.
    Each application, including the changes of basis required by the Beals construction, costs $\widetilde O(k^2)$ gates~\cite[Lemma~4.5 and proof of Theorem~1.1]{BruinsmaGrinkoOzols2026}.
    The total cost of these operations is
    \[
        \sum_{k=2}^{n}
        \left\lfloor\frac{n}{k}\right\rfloor\widetilde O(k^2)
        \leq\widetilde O\!\left(n\sum_{k=2}^{n}k\right)
        =\widetilde O(n^3).
    \]

    \item \textbf{Restore the factor order.}
    Reverse the sorting by applying the inverse permutation on the factors and clear the stored comparison results.
    Uncompute the stored factor sizes from $\vsig$.
    This costs $\widetilde O(n^2)$ gates.
    Each factor of size $n_r$ has received exactly the operations for $\Fou_{S_{n_r}}$, and all computations controlling these operations are reversible, so the circuit implements the required tensor product coherently controlled on $\ket{\vsig}$.

\end{enumerate}
Adding the costs gives $\widetilde O(n^3)$ one- and two-qubit gates.

As for the approximation accuracy, we approximate each controlled application of $\IndTrans[k]$, including its required changes of basis, to operator-norm error at most $\delta/n^2$.
There are fewer than $n^2$ such applications, so we get error at most $\delta$ for the complete controlled circuit.
The precision overhead is polylogarithmic in $n/\delta$ and is absorbed by the $\widetilde O$ notation.
\end{proof}

We can now give a proof of the main theorem of this section.
\begin{proof}[Proof of \cref{thm:wreath-main}]

Correctness follows from \cref{thm:little-group-qft-correctness}, since \cref{lem:general-wreath-extension} supplies the required extension and the wreath product algorithm specializes the steps of \cref{alg:little-group-qft}. It remains to prove the gate-count bound.

Let
\[
    d_{\max}(F):=\max_{\rho\in\widehat F}d_\rho,
    \qquad
    L_F
    :=
    \left\lceil
        \log_2\max\bigl\{|\widehat F|,d_{\max}(F)\bigr\}
    \right\rceil.
\]
Thus $O(L_F)$ bounds the number of qubits required for an $F$-irrep label or for
one of its matrix indices.  We first record the cost of applying a permutation
to the tensor factors.  Each factor is stored in $O(L_F+\log n)$ qubits, so
exchanging two factors costs $O(L_F+\log n)$ elementary gates.  Every
permutation of $n$ factors is a
product of at most $n(n-1)/2$ adjacent transpositions.  Consequently, a
coherently controlled tensor-factor permutation costs
\[
    O\!\left(n^2(L_F+\log n)\right)
    =
    \widetilde O\!\left(n^2L_F\right)
\]
elementary gates.

We use the numbering of \cref{tab:wreath-circuit-dictionary}, and analyze the cost of each step separately.\\
Step~2 applies $n$ copies of $\Fou_F$, each with error at most $\varepsilon/(2n)$, and contributes $nC_F(\varepsilon/(2n))$ gates.\\
Step~3 applies a reversible stable sort to the pairs $(\lambda_r,r)$.
 The
sorted labels form $\repO(\vla)=\vsig$, while the position indices
determine $\pi_{\repO(\vla),\vla}=\pi_L$. This contributes
$\widetilde O(n^2L_F)$ gates.\\
Step~4 implements the two orbit transports by applying $\pi_L$ to the two
$F^n$ Fourier registers. This contributes $\widetilde O(n^2L_F)$ gates.\\
Step~5 computes $\vla_R=\act{\pi\inv}\vla$, performs the second stable sort,
and writes the resulting $\pi_R$ as a cycle list. It then computes
$\alpha=\pi_L\inv\pi\pi_R$. This contributes
$\widetilde O(n^2(L_F+1))$ gates.\\
Step~6 implements $U_{\alpha,\vsig}\tp$ and contributes
$\widetilde O(n^2L_F)$ gates.\\
Step~8 applies the $S_{n_r}$-QFTs. By \cref{lem:wreath-coherent-young-qft} this step contributes
$\widetilde O(n^3)$.

Adding the six contributions gives

\begin{align*}
    C_{F\wr S_n}(\varepsilon)
    &\leq
    nC_F\!\left(\frac{\varepsilon}{2n}\right)
    +4\widetilde O\!\left(n^2L_F\right)
    +\widetilde O(n^3)
    \\
    &=
    nC_F\!\left(\frac{\varepsilon}{2n}\right)
    +\widetilde O\!\left(n^3+n^2L_F\right),
\end{align*}
which is \cref{eq:wreath-main-gate-bound}.
The same calculation, running the QFTs for the copies of $F$ on disjoint registers in parallel, and bounding the depth of the remaining circuit by its gate count, gives the depth bound
\[
D_{F\wr S_n}(\varepsilon)
\leq
D_F\!\left(\frac{\varepsilon}{2n}\right)
+
\widetilde O\!\left(n^3+n^2L_F\right).
\]
Choose now operator-norm error $\varepsilon/(2n)$ for each of the $n$
copies of $\Fou_F$ in Step~{2} and error $\delta=\varepsilon/2$ for the complete
coherently controlled Young-subgroup QFT in Step~8, using \cref{lem:wreath-coherent-young-qft}.  Replacing the factors in a tensor product and
applying the triangle inequality gives
\[
    \left\|
        \bigotimes_{r=1}^m\widetilde U_r
        -
        \bigotimes_{r=1}^m U_r
    \right\|_{\mathrm{op}}
    \leq
    \sum_{r=1}^m
    \|\widetilde U_r-U_r\|_{\mathrm{op}}.
\]
It follows that the $n$ copies of $\Fou_F$ in Step~{2} contribute error at most
$\varepsilon/2$.  For every fixed $\vsig$, \cref{lem:wreath-coherent-young-qft} bounds the error of the complete Step~8 transform by $\delta=\varepsilon/2$.  Since
the controlled Step~{8} gate is block diagonal in $\vsig$, its operator-norm
error obeys the same bound.  The remaining gates are exact, and the sorting, cycle-manipulation, and swap circuits are reversible, with all their work registers being uncomputed, so they introduce no additional approximation. Hence
\[
    \left\|
        \widetilde{\Fou}_{F\wr S_n}-\Fou_{F\wr S_n}
    \right\|_{\mathrm{op}}
    \leq \varepsilon.
\]
The resulting precision overhead is polylogarithmic in $n/\varepsilon$ by~\cref{lem:wreath-coherent-young-qft} and is absorbed by the $\widetilde O$ notation.
This proves the claim.
\end{proof}

We next analyze the qubit count required for implementing the circuit in \cref{thm:wreath-main}.

\begin{lemma}[Qubit count for the wreath product QFT]
\label{lem:wreath-qubit-count}
Let $Q_F(\delta)$ count the total number of qubits in the implementation of the circuit for $\Fou_F$ used in the proof of \cref{thm:wreath-main}, with error $\delta$, and let $ Q_{F\wr S_n}(\varepsilon)$ be the total number of qubits required to implement the full wreath product circuit with error $\varepsilon$. Then
\begin{equation}
    Q_{F\wr S_n}(\varepsilon) = \widetilde O\!\left(nQ_F\!\left(\frac{\varepsilon}{2n}\right)+n^2\right).
    \label{eq:wreath-qubit-bound}
\end{equation}
\end{lemma}

\begin{proof}
Step~2 uses $nQ_F(\varepsilon/(2n))$ qubits for the parallel $F$-QFTs, together with $O(n\log(n+1))$ qubits for the input permutation.
In Steps~3--6, the Fourier data, temporary label arrays, and permutation registers occupy $O(nL_F+n\log(n+1))$ qubits.
Stable sorting and tensor-factor permutations can use $O(n^2)$ comparison decisions, with the comparator workspace reused between operations.
Their temporary arrays and histories can be uncomputed from the retained outputs with a constant-factor gate overhead.
Elementary permutation operations require at most $\widetilde O(n^2)$ additional qubits, since they have $\widetilde O(n^2)$ one- and two-qubit gates.

In Step~8, \cref{lem:wreath-coherent-young-qft} allocates $n$ padded factor registers of $\widetilde O(n)$ qubits each, giving $\widetilde O(n^2)$ data qubits.
Factor sizes and sorting decisions use $\widetilde O(n)$ further qubits.
Each permutation conversion and each induction call has at most $\widetilde O(n^2)$ gates, hence needs at most $\widetilde O(n^2)$ additional work qubits after unused wires are omitted.
Perform these calls sequentially, clearing and reusing their work registers.
The total additional space is therefore $\widetilde O(nL_F+n^2)$, including the polylogarithmic precision overhead.
Since the input register must encode $|F|$ orthogonal group elements, $Q_F(\delta)\geq\lceil\log_2|F|\rceil\geq L_F$, which gives \cref{eq:wreath-qubit-bound}.
This scheduling preserves the gate count and the depth bound obtained from it in \cref{thm:wreath-main}.
\end{proof}

\section{Acknowledgments}
The authors would like to thank Peter Bruin, Matthias Christandl, and Eric Opdam for useful discussions. We are also thankful to the organizers of the joint scientific retreat of QMATH, QuSoft, and the LMU quantum information theory group, where part of this research collaboration started.

CB was supported by a QDNL CAT-1 Phase 3 grant.
PMP was supported by the European Research Council (ERC Grant Agreement No.~818761), VILLUM FONDEN via the QMATH Centre of Excellence (Grant No.~10059), the Novo Nordisk Foundation (grant NNF20OC0059939 `Quantum for Life'), a PhD scholarship on Quantum Algorithms from the Danish e-infrastructure Consortium (DeiC), and INdAM-GNFM.
JS was supported by a Dutch Research Council (NWO) grant 024.003.037 as part of the Quantum Software Consortium program, and a QDNL grant (NGF.1623.23.023).
DG and MO were supported by National Growth Fund grant (NGF.1623.23.025) ``Qudits in theory and experiment''.

\paragraph{AI statement}
Generative AI tools were used to obtain part of the results of this work, to find references and to edit the manuscript.
In particular, the construction of an efficient circuit for the $A_{\gamma}$-matrix for $\GL_2(\F_q)$ was found with the help of a generative AI system. It also helped in formulating the proofs in \cref{subsec: polylog GL2 A}.

\printbibliography[heading=bibintoc]

\appendix

\zcsetup{countertype={section=appendix}}
\section{Finite fields} \label{app: finite fields}
Let $q=p^r$ be a prime power and use $\F_q$ to denote the finite field of $q$ elements. Write $\F_q^* = \F_q \setminus\{0\}$ for its unit group. For any two fields $\F_q$ and $\F_{q^n}$ there exists a \emph{trace map}
\begin{equation}\label{eq: tr}
	\tr_{q^n}^q: \F_{q^n} \to \F_q: x \mapsto x+x^q+\dots + x^{q^{n-1}},
\end{equation}
and a \emph{norm map}
\begin{equation}\label{eq: N}
	\No_{q^n}^{q}: \F_{q^n} \to \F_q: x\mapsto x^1\cdot x^q\cdot \ldots \cdot x^{q^{n-1}}.
\end{equation}
We will often omit the sub- and superscript from notation. Both $\tr_{q^2}^q$ and $\No_{q^2}^q$ are surjective.

Consider a quadratic extension $\F_{q^2}$ of $\F_q$. As it is a degree 2 extension there is a unique conjugation on $\F_{q^2}$ given by $\bar{z} = z^q$, also known as the \emph{Frobenius action}. The trace and norm of this extension are given by $\tr(z)= z+ \bar{z}$ and $\No(z)= z\bar{z}$.
\begin{itemize}
	\item The multiplicative characters of $\F_q^*$ form the cyclic group $X_q:=\widehat{\F_q^*}$.  After fixing a generator $g$ of $\F_q^*$ we label them by $\chi_\alpha$, with $\alpha$ an exponent label so that the character is given by $\chi_{\alpha}(g^d)= e^{2\pi i \alpha d/(q-1)}$.  Similarly write $X_{q^2}:=\widehat{\F_{q^2}^*}$.
	\item The additive characters of $\F_q$ are indexed by $\psi_t(x)=e^{2\pi i\tr_q^p(tx)/p}$ for $t\in\F_q$.  Throughout we fix the non-trivial character $\psi=\psi_1$ and define $\psi_b(x)=\psi(bx)$ for any $b\in \F_q$.
\end{itemize}
\begin{definition}
	A character $\theta \in \widehat{\F_{q^2}^*}$ is \emph{decomposable} if it factors through $\F_q^*$, i.e. if there exists a $\mu\in \widehat{\F_q^*}$ such that $\theta(z) = \mu(\No(z))$ for all $z\in \F_{q^2}^*.$
\end{definition}
\begin{lemma}
	A character $\theta$ is decomposable if and only if $\theta(z) = \theta(\bar{z})$ for all $z \in \F_{q^2}^*$.
\end{lemma}
\begin{proof}
	If $\theta=\mu\circ\No$, then $\theta(z)=\mu(z\bar z)=\theta(\bar z)$.  Conversely, suppose that $\theta(z)=\theta(z^q)$ for all $z$.  Since $\F_{q^2}^*$ is cyclic, every element of $\ker(\No)$ has the form $u=z^{q-1}$.  Hence $\theta(u)=\theta(z^q)\theta(z)^{-1}=1$.  The norm is surjective, so $\mu(\No(z)):=\theta(z)$ is a well-defined multiplicative character of $\F_q^*$ and $\theta=\mu\circ\No$.
\end{proof}
In other words, a character is decomposable (also known as \emph{irregular}) if and only if it is fixed by the Frobenius action.

\subsection{Gauss sums}\label{subsec: Gauss sums}
\begin{definition}[Gauss sum]\label{def: Gauss sum}
	For $\alpha\in X_q$ and $b\in\F_q$, their \emph{Gauss sum} is given by
	\begin{equation}
		G_q(\alpha,b)=\sum_{x\in\F_q^*}\chi_\alpha(x)\psi(bx),
		\qquad g_q(\alpha):=G_q(\alpha,1).
	\end{equation}
	For a character $\theta\in X_{q^2}$, write
	\begin{equation}
		g_{q^2}(\theta)
		:=\sum_{z\in\F_{q^2}^*}\theta(z)\psi(\tr(z)).
	\end{equation}
\end{definition}

The definitions of $g_q(\alpha)$ and $g_{q^2}(\theta)$ are consistent, since the generating additive character over $\F_{q^2}$ can be written as $\psi_{q^2}(z) = \psi(\tr(z))$.
 We give some properties of Gauss sums.

\begin{lemma}\label{lem: Gauss sum property}
	For $\alpha,\beta\in X_q$ and $b\in\F_q^*$,
	\begin{align}
		G_q(\alpha,b)&=\chi_{\alpha^{-1}}(b)g_q(\alpha),\\
		\overline{g_q(\alpha)}&=\chi_\alpha(-1)g_q(\alpha^{-1}),\\
		g_q(\alpha)g_q(\beta)&=J(\alpha,\beta)g_q(\alpha\beta),
		&&\alpha\beta\neq1, \label{eq: Gauss sum prop 3}\\
		g_q(1)=-1, \quad |g_q(\alpha)|&=\sqrt q, &&\alpha\neq1,
	\end{align}
	where $J(\alpha,\beta)=\sum_{x\in\F_q^*\setminus\{1\}} \chi_\alpha(x)\chi_\beta(1-x)$ is known as the \emph{Jacobi sum}.
	Moreover, $|g_{q^2}(\theta)|=q$ for non-trivial $\theta$.
\end{lemma}

\begin{lemma}
\label{lem: Gauss sum trace fibers}
	Let $\theta\in X_{q^2}$, and write $\theta|_{\F_q^*}=\chi_\alpha$.  For $t\in\F_q$, set
	\begin{equation}
		S_t(\theta):=
		\sum_{\substack{z\in\F_{q^2}^*\\ \tr(z)=t}}\theta(z).
	\end{equation}
	Then
	\begin{equation}
		S_t(\theta)=\chi_\alpha(t)S_1(\theta)
		\quad(t\in\F_q^*),
		\qquad
		g_{q^2}(\theta)=S_0(\theta)+g_q(\alpha)S_1(\theta).
		\label{eq: Gauss sum trace fiber splitting}
	\end{equation}
	If $\zeta\in\F_{q^2}^*$ has $\tr(\zeta)=0$, then
	\begin{equation}
		S_0(\theta)=
		\begin{cases}
			0,&\alpha\neq1,\\
			(q-1)\theta(\zeta),&\alpha=1.
		\end{cases}
	\end{equation}
	Moreover, if $\theta\neq1$ and $\alpha=1$, then
	\begin{equation}
		S_1(\theta)=-\theta(\zeta),
		\qquad
		g_{q^2}(\theta)=q\theta(\zeta).
		\label{eq: Gauss sum trivial restriction fibers}
	\end{equation}
\end{lemma}
\begin{proof}
	For $t\neq0$, multiplication by $t$ maps the fiber above $1$ bijectively onto the fiber above $t$.  Hence $S_t(\theta)=\theta(t)S_1(\theta)=\chi_\alpha(t)S_1(\theta)$.
	Regrouping the defining sum for $g_{q^2}(\theta)$ according to the value of the trace now gives
	\begin{equation*}
		g_{q^2}(\theta)
		=S_0(\theta)+S_1(\theta)
		  \sum_{t\in\F_q^*}\chi_\alpha(t)\psi(t),
	\end{equation*}
	which is \cref{eq: Gauss sum trace fiber splitting}.

	The nonzero elements of the trace-zero fiber are precisely $\{t\zeta:t\in\F_q^*\}$.  Therefore
	\begin{equation*}
		S_0(\theta)=\theta(\zeta)
		\sum_{t\in\F_q^*}\chi_\alpha(t),
	\end{equation*}
	which proves the stated formula for $S_0(\theta)$.  Finally, suppose that $\theta\neq1$ and $\alpha=1$.  Then $\theta$ induces a nontrivial character of $\F_{q^2}^*/\F_q^*$.  The coset $\zeta\F_q^*$ is the unique trace-zero coset, while every other coset has a unique representative of trace $1$.  Summing the induced character over the quotient gives $\theta(\zeta)+S_1(\theta)=0$.  The last identity follows from \cref{eq: Gauss sum trace fiber splitting} and $g_q(1)=-1$.
\end{proof}

\begin{lemma}[Davenport--Hasse lifting equation]
\label{lem: degree two Davenport Hasse}
	For every $\alpha\in X_q$,
	\begin{equation}
		g_{q^2}(\chi_\alpha\circ\No)
		=-\bigl(g_q(\alpha)\bigr)^2.
		\label{eq: appendix Davenport Hasse}
	\end{equation}
\end{lemma}
\begin{proof}
	For $\alpha\neq1$, the Davenport--Hasse lifting theorem, with the lifted additive character $\psi_{q^n}=\psi\circ\tr_{q^n}^q$, states that
	\begin{equation*}
		\sum_{z\in\F_{q^n}^*}
		\chi_\alpha\bigl(\No_{q^n}^q(z)\bigr)
		\psi\bigl(\tr_{q^n}^q(z)\bigr)
		=(-1)^{n-1}\bigl(g_q(\alpha)\bigr)^n,
	\end{equation*}
	see \cite{DavenportHasse1935}.  Taking $n=2$ proves the claim.  If $\alpha=1$, then the left-hand side is $-1$, while the right-hand side is $-g_q(1)^2=-1$, so the identity holds in that case as well.
\end{proof}

\section{Additional proofs of \texorpdfstring{\cref{sec: GL2}}{the GL2 section}}

\subsection{Action of transversals}\label{app: twiddle calcs}

As the Gelfand--Tsetlin basis for the principal series representation uses two Fourier transforms, computing the irrep action at $w=\psm{0&1\\1&0}$ becomes quite tedious. We do it here and use the results in \cref{par: twiddle principal,par: twiddle Steinberg}.

Following \cref{eq: g action projective 1,eq: g action projective 2}, the action of $I_{\alpha, \beta}(w)$ on the Gelfand--Tsetlin basis is given by
\begin{align*}
	&I_{\alpha, \beta} (w)\,\ket{0} = \ket{\infty}, \\
	&I_{\alpha, \beta} (w) \,\ket{\infty} = \ket{0}, \\
	&I_{\alpha, \beta} (w) \,\ket{x} = \chi_{\alpha}(x)\chi_{\beta}(-1/x)\ket{1/x} \quad \text{ for } x\in \F_q^*\\
	&I_{\alpha, \beta} (w) \ket{v_t} = \frac{1}{\sqrt{q}}\left(\sum_{x\in \F_q^*} \overline{\psi(tx)} \chi_{\alpha}(x) \chi_{\beta}(-1/x) \ket{1/x} + \ket{\infty} \right)\\
	&I_{\alpha,\beta}(w)\ket{w_\delta} = \frac{1}{\sqrt{q(q-1)}} \sum_{t\in \F_q^*} \chi_{\delta\beta^{-1}}(t)\left(\sum_{x\in \F_q^*} \overline{\psi(tx)} \chi_{\alpha}(x) \chi_{\beta}(-1/x) \ket{1/x} + \ket{\infty}\right).\\
\end{align*}
We compute the following inner products
\begin{align*}
	\braket{\infty}{I_{\alpha, \beta} (w) | \infty} &=0\\
	\braket{v_0}{I_{\alpha, \beta}(w) | \infty} &= \frac{1}{\sqrt{q}}\\
	\braket{w_{\delta}}{I_{\alpha, \beta} (w) | \infty}
	&= \frac{1}{\sqrt{q(q-1)}}\sum_{t\in \F_q^*} \chi_{\delta^{-1}\beta}(t) = \sqrt{\frac{q-1}{q}}\; \delta_{\delta=\beta}\\
	\braket{v_{0}}{I_{\alpha, \beta} (w) | v_0} &= \frac{1}{q}\sum_{x,y\in \F_q^*} \chi_{\alpha}(x)\chi_{\beta}(-1/x) \braket{y}{1/x}= \frac{1}{q}\chi_{\beta}(-1)\sum_{x\in \F_q^*}\chi_{\alpha\beta^{-1}}(x)\\
    &= \frac{q-1}{q}\chi_{\beta}(-1)\delta_{\alpha=\beta}.\\
\end{align*}
They result in the following factors
\begin{align*}
	\braket{w_{\delta}}{I_{\alpha, \beta} (w) | v_0}
	&=\frac{1}{\sqrt{q-1}}\sum_{t\in \F_q^*} \chi_{\delta^{-1}\beta}(t) \braket{v_t}{I_{\alpha, \beta} (w) | v_0} \\
	&= \frac{1}{\sqrt{q-1}}\sum_{t\in \F_q^*} \chi_{\delta^{-1}\beta}(t)\frac{1}{q}\sum_{x,y\in \F_q^*} \psi(ty) \chi_{\alpha}(x)\chi_{\beta}(-1/x) \braket{y}{1/x} \\
	&=\frac{1}{q\sqrt{q-1}}\sum_{t\in \F_q^*}\chi_{\delta^{-1} \beta}(t)\chi_{\beta}(-1)\sum_{x\in \F_q^*}\psi(t/x)\chi_{\alpha\beta^{-1}}(x)\\
	&=\frac{1}{q\sqrt{q-1}} \chi_{\beta}(-1)\sum_{x\in \F_q^*}\chi_{\alpha \beta^{-1}}(x) G_q(\delta^{-1}\beta, 1/x)\\
	&=\frac{1}{q\sqrt{q-1}} \chi_{\beta}(-1)\sum_{x\in \F_q^*}\chi_{\alpha \beta^{-1}}(x) \chi_{\delta^{-1}\beta}(x)G_q(\delta^{-1}\beta, 1)\\
	&= \frac{\sqrt{q-1}}{q} \chi_{\beta}(-1) g_q(\alpha^{-1}\beta) \delta_{\delta=\alpha}\\
	\braket{w_{\delta}}{I_{\alpha, \beta} (w) | w_{\epsilon}}
	&=\frac{1}{q-1}\sum_{s, t\in \F_q^*} \chi_{\delta^{-1}\beta}(s)\chi_{\epsilon \beta^{-1}}(t)\braket{v_{s}}{I_{\alpha, \beta} (w) | v_t} \\
	&=  \frac{1}{q-1}\sum_{s, t\in \F_q^*} \chi_{\delta^{-1}\beta}(s)\chi_{\epsilon \beta^{-1}}(t)\frac{1}{q}\sum_{x,y\in \F_q^*} \psi(sy)\overline{\psi(tx)} \chi_{\alpha}(x)\chi_{\beta}(-1/x) \braket{y}{1/x}\\
	&=\frac{1}{q(q-1)}\sum_{s, t\in \F_q^*} \chi_{\delta^{-1}\beta}(s)\chi_{\epsilon \beta^{-1}}(t)\chi_{\beta}(-1)\sum_{x\in \F_q^*}\psi(s/x)\psi(-tx) \chi_{\alpha\beta^{-1}}(x)\\
	&=\frac{1}{q(q-1)}\chi_{\beta}(-1) \sum_{x\in \F_q^*} \chi_{\alpha\beta^{-1}}(x) \,G_q(\delta^{-1}\beta, 1/x)\, G_q(\epsilon\beta^{-1}, -x)\\
	&= \frac{1}{q(q-1)}\chi_{\beta}(-1) \, g_q(\delta^{-1}\beta) g_q(\epsilon\beta^{-1}) \sum_{x\in \F_q^*} \chi_{\alpha\beta^{-1}}(x) \chi_{\delta^{-1}\beta}(x) \chi_{\epsilon^{-1}\beta}(-x)\\
	&= \frac{1}{q}\chi_{\epsilon^{-1}}(-1) g_q(\delta^{-1}\beta)g_q(\epsilon \beta^{-1}) \delta_{\delta = \alpha\beta\epsilon^{-1}}.
\end{align*}
Whenever $\alpha=\beta$ the following holds:
\begin{align*}
	&I_{\alpha, \alpha} (w)\,\ket{0} = \ket{\infty}, \\
	&I_{\alpha, \alpha} (w) \,\ket{\infty} = \ket{0}, \\
	&I_{\alpha, \alpha} (w) \,\ket{x} = \chi_{\alpha}(-1)\ket{1/x},
\end{align*}
and
\begin{align*}
	I_{\alpha, \alpha}(w)\ket{v^\perp} &=I_{\alpha, \alpha}(w)\frac{1}{\sqrt{q+1}}(\sqrt{q}\,\chi_\alpha(-1) \ket{\infty} - \ket{v_0})\\
	&= \frac{1}{\sqrt{q+1}}\Bigg( \chi_{\alpha}(-1)\ket{v_0} + \chi_{\alpha}(-1)\sqrt{q-1}\ket{w_{\alpha}} \\
	&\qquad -\left(\frac{1}{\sqrt{q}}\ket{\infty} + \frac{q-1}{q}\chi_{\alpha}(-1)\ket{v_0} + \frac{\sqrt{q-1}}{q}\chi_{\alpha}(-1)g_q(1)\ket{w_{\alpha}}\right)\Bigg)\\
	&=\frac{1}{\sqrt{q+1}}\Biggl(-\frac{1}{q}\chi_{\alpha}(-1)(\sqrt{q}\,\chi_{\alpha}(-1)\ket{\infty} - \ket{v_0}) \\
	&\qquad\qquad + \chi_{\alpha}(-1)\left(\sqrt{q-1} + \frac{\sqrt{q-1}}{q}\right)\ket{w_\alpha}\Biggr)\\
	&=-\frac{1}{q}\chi_{\alpha}(-1)\ket{v^\perp}+ \chi_{\alpha}(-1)\frac{\sqrt{q^2-1}}{q}\ket{w_{\alpha}} \\
	I_{\alpha, \alpha}(w)\ket{w_{\alpha}} &= \frac{\sqrt{q-1}}{\sqrt{q}}\ket{\infty} + \frac{\sqrt{q-1}}{q}\chi_{\alpha}(-1)g_q(1)\ket{v_0} + \frac{1}{q}\chi_{\alpha}(-1)g_q(1)g_q(1)\ket{w_{\alpha}} \\
	&=\frac{\sqrt{q^2-1}}{q}\chi_{\alpha}(-1)\ket{v^\perp} + \frac{1}{q}\chi_{\alpha}(-1)\ket{w_{\alpha}}\\
	\St_{\alpha}(w)\ket{w_{\delta}}
	&=\frac{1}{q}\chi_{\delta^{-1}}(-1)
	g_q(\delta\alpha^{-1})^2
	\ket{w_{\alpha^2\delta^{-1}}}.
\end{align*}

\subsection{Proofs of \texorpdfstring{\cref{subsec: polylog GL2 A}}{the A-matrix efficiency subsection}}
\label{app: proofs of A matrix}

\begin{lemma}[Torus action, \protect{\cref{lem: GL2 relative torus coefficients}}]
	On the special input symbol, the torus circuit acts as
	\begin{align}
		\T_\gamma\ket{c_\infty}
		={}&\frac1{\sqrt{q+1}}
		\sum_{\{\theta,\theta^q\}\in\mathcal C_\gamma}\ket{\pi_\theta}
		+\frac1{\sqrt{q-1}}
		\sum_{\{\delta,\eta\}\in\mathcal P_\gamma}
		\ket{I_{\delta,\eta}}\nonumber\\
		&+\sqrt{\frac q{(q-1)(q+1)}}
		\sum_{\delta\in\mathcal S_\gamma}\ket{\St_\delta}.
		\label{eq: appendix GL2 relative identity column}
	\end{align}
	For input $\alpha\in X_q$ and output labels $I_{\delta,\eta}, \pi_{\theta}$, the coefficients are
	\begin{align}
		\bra{I_{\delta,\eta}}\T_\gamma\ket{c_\alpha}
		&=\frac1{q-1}J(\alpha^{-1}\delta,\alpha^{-1}\eta),
		\label{eq: appendix GL2 relative principal coefficient}\\
		\bra{\pi_\theta}\T_\gamma\ket{c_\alpha}
		&=-\frac1{\sqrt{(q-1)(q+1)}}
		\sum_{\tr(z)=1}\theta(z)\chi_{\alpha^{-1}}(\No z).
		\label{eq: appendix GL2 relative cuspidal coefficient}
	\end{align}
	If $\delta^2=\gamma$ and $\alpha\ne b$, then
	\begin{equation}
		\bra{\St_\delta}\T_\gamma\ket{c_\alpha}
		=\frac{\sqrt q}{(q-1)\sqrt{q+1}}
		J(\alpha^{-1}\delta,\alpha^{-1}\delta),
		\qquad
		\bra{\det_\delta}\T_\gamma\ket{c_\alpha}=0.
		\label{eq: appendix GL2 relative fixed coefficient}
	\end{equation}
	The coefficients on the unused antisymmetric labels are zero.
\end{lemma}

\begin{proof}
    The proof consists of a direct calculation of all the cases mentioned in the lemma.

	First consider the special input $\ket{c_\infty}$.  The inverse multiplicative Fourier transform fixes this state and identifies it with $\ket0$, so that
	\begin{equation*}
	B_\gamma\Fou_\times^\dagger\ket{c_\infty}
	=\frac{\ket{\mathrm{sp},1}+\ket{\mathrm{ns},1}}{\sqrt2}.
	\end{equation*}
	After the two  Fourier transforms, every split label has amplitude $1/\sqrt{2(q-1)}$ and every nonsplit label has amplitude $1/\sqrt{2(q+1)}$.  For a nonfixed split pair, the folding therefore gives an amplitude of
	\begin{equation*}
	\frac1{\sqrt2}
	\left(\frac1{\sqrt{2(q-1)}}+
	\frac1{\sqrt{2(q-1)}}\right)
	=\frac1{\sqrt{q-1}},
	\end{equation*}
	and the same calculation on a nonfixed nonsplit pair gives a factor of $1/\sqrt{q+1}$.  If $\delta^2=\gamma$, the coefficient of $\ket{\St_\delta}$ after the fixed rotation is
	\begin{equation*}
	\frac1{\sqrt{2(q-1)}}\sqrt{\frac{q+1}{2q}}
	+\frac1{\sqrt{2(q+1)}}\sqrt{\frac{q-1}{2q}}
	=\sqrt{\frac q{(q-1)(q+1)}},
	\end{equation*}
	while the two contributions to $\ket{\det_\delta}$ cancel. This proves \cref{eq: appendix GL2 relative identity column}.

	We next calculate the coefficients for $\ket{c_\alpha}$.  The Fourier transform gives
	\begin{equation}
		\Fou_\times^\dagger\ket{c_\alpha}
		=\frac1{\sqrt{q-1}}
		\sum_{t\in\F_q^*}\chi_{\alpha^{-1}}(t)\ket t
		\label{eq: GL2 inverse multiplicative input expansion}
	\end{equation}
	and we first consider the split part of the state. For split values $t$ having distinct roots $x,y\in\F_q^*$, it holds that $x+y=1$ and $xy=t$.  Applying $B_\gamma$ and then the split Fourier transform to all these terms at once gives
	\begin{align*}
		\frac1{\sqrt2(q-1)}
		\sum_{\substack{\{x,y\}\subset\F_q^*\\x+y=1,\;x\ne y}}
		\sum_{\kappa\in X_q}
		\chi_{\alpha^{-1}}(xy)
		\bigl(&\chi_\gamma(y)\chi_\kappa(x/y)
		+\chi_\gamma(x)\chi_\kappa(y/x)\bigr)
		\ket{\mathrm{sp},\kappa},
	\end{align*}
	where the first sum is over unordered root pairs.  The two terms in its summand satisfy
	\begin{align*}
		\chi_{\alpha^{-1}}(xy)\chi_\gamma(y)\chi_\kappa(x/y)
		&=\chi_{\alpha^{-1}\kappa}(x)
		\chi_{\alpha^{-1}\gamma\kappa^{-1}}(y),\\
		\chi_{\alpha^{-1}}(xy)\chi_\gamma(x)\chi_\kappa(y/x)
		&=\chi_{\alpha^{-1}\kappa}(y)
		\chi_{\alpha^{-1}\gamma\kappa^{-1}}(x).
	\end{align*}
	Because $y=1-x$, the two terms for the unordered pair $\{x,y\}$ are exactly the summands indexed by the two ordered roots $x$ and $1-x$.

	When $q$ is odd, the remaining split input is $t=1/4$, with repeated root $x=1/2$.  Its contribution after the split Fourier transform is the full state
	\begin{equation*}
	\frac{\chi_{\alpha^{-1}}(1/4)\chi_\gamma(1/2)}
	{\sqrt2(q-1)}
	\sum_{\kappa\in X_q}\ket{\mathrm{sp},\kappa}.
	\end{equation*}
	For every $\kappa\in X_q$, the coefficient appearing here satisfies
	\begin{equation*}
	\chi_{\alpha^{-1}}(1/4)\chi_\gamma(1/2)
	=\chi_{\alpha^{-1}\kappa}(1/2)
	\chi_{\alpha^{-1}\gamma\kappa^{-1}}(1/2),
	\end{equation*}
	so this is precisely the missing $x=1/2$ term.  In even characteristic there is no repeated root.  Combining the distinct-root terms with this boundary term gives, in every characteristic, the complete split state before $R_{\mathrm{fix},\gamma}$ as
	\begin{align}
		&\frac1{\sqrt2(q-1)}
		\sum_{\kappa\in X_q}
		\sum_{x\in\F_q\setminus\{0,1\}}
		\chi_{\alpha^{-1}\kappa}(x)
		\chi_{\alpha^{-1}\gamma\kappa^{-1}}(1-x)
		\ket{\mathrm{sp},\kappa}
		\nonumber\\
		&\qquad=
		\frac1{\sqrt2(q-1)}
		\sum_{\kappa\in X_q}
		J(\alpha^{-1}\kappa,
		\alpha^{-1}\gamma\kappa^{-1})
		\ket{\mathrm{sp},\kappa}.
		\label{eq: GL2 full split state before folding}
	\end{align}
	The values $x=0,1$ are absent because both correspond to $t=0$, which does not occur in \cref{eq: GL2 inverse multiplicative input expansion}.

	We now group the complete sum under the involution $\kappa\mapsto\gamma\kappa^{-1}$.  Since $J$ is symmetric, \cref{eq: GL2 full split state before folding} becomes
	\begin{align}
		\frac1{\sqrt2(q-1)}\bigg(&
		\sum_{\{\delta,\eta\}\in\mathcal P_\gamma}
		J(\alpha^{-1}\delta,\alpha^{-1}\eta)
		\bigl(\ket{\mathrm{sp},\delta}
		+\ket{\mathrm{sp},\eta}\bigr)
		\nonumber\\
		&+\sum_{\delta\in\mathcal S_\gamma}
		J(\alpha^{-1}\delta,\alpha^{-1}\delta)
		\ket{\mathrm{sp},\delta}\bigg).
		\label{eq: GL2 split state grouped into orbits}
	\end{align}
	Applying $R_{\mathrm{fix},\gamma}$ to the first sum gives
	\begin{equation*}
	\frac1{q-1}
	\sum_{\{\delta,\eta\}\in\mathcal P_\gamma}
	J(\alpha^{-1}\delta,\alpha^{-1}\eta)
	\ket{I_{\delta,\eta}}.
	\end{equation*}
	This proves \cref{eq: appendix GL2 relative principal coefficient} for all non-fixed split orbits simultaneously. We also obtain that the antisymmetric split amplitudes are zero.

	Next we calculate the nonsplit part of the state, which follows the same pattern as the split case.
	For every nonsplit value $t$, the roots of $X^2-X+t$ form an unordered Frobenius pair $\{z,z^q\}$ in $\F_{q^2}\setminus\F_q$, with
	\begin{equation*}
	z+z^q=1,
	\qquad
	\No(z)=t.
	\end{equation*}
	Applying $B_\gamma$ and then $\Fou_{\ker\No}$ to all irreducible terms in \cref{eq: GL2 inverse multiplicative input expansion} gives
	\begin{align*}
		-\frac1{\sqrt{2(q-1)(q+1)}}
		\sum_{\substack{\{z,z^q\}\subset\F_{q^2}\setminus\F_q\\
				\tr(z)=1}}
		\sum_{\nu\in\widehat{\ker\No}}
		\chi_{\alpha^{-1}}(\No z)
		\bigl(&\theta_0(z)\nu(z/z^q)
		+\theta_0(z^q)\nu(z^q/z)\bigr)
		\ket{\mathrm{ns},\theta_0\nu}.
	\end{align*}
	The first sum is over unordered Frobenius pairs.  By the definition of the nonsplit characters, if $\theta=\theta_0\nu$ then
	\begin{equation*}
	\theta_0(z)\nu(z/z^q)=\theta(z),
	\qquad
	\theta_0(z^q)\nu(z^q/z)=\theta(z^q).
	\end{equation*}
	Thus the two terms belonging to $\{z,z^q\}$ are exactly the summands indexed by its two ordered roots.

	When $q$ is odd, the nonsplit component of the boundary input $t=1/4$ contributes the complete state
	\begin{equation*}
	-\frac{\chi_{\alpha^{-1}}(1/4)\chi_\gamma(1/2)}
	{\sqrt{2(q-1)(q+1)}}
	\sum_{\nu\in\widehat{\ker\No}}
	\ket{\mathrm{ns},\theta_0\nu}.
	\end{equation*}
	Every character $\theta=\theta_0\nu$ restricts to $\chi_\gamma$ on $\F_q^*$.  Since $\No(1/2)=1/4$, its coefficient can therefore be written as
	\begin{equation*}
	\chi_{\alpha^{-1}}(1/4)\chi_\gamma(1/2)
	=\theta(1/2)\chi_{\alpha^{-1}}(\No(1/2)).
	\end{equation*}
	This is precisely the missing ordered-root term $z=1/2$.  There is no Frobenius-fixed element of trace one in even characteristic.  Consequently, in every characteristic the complete nonsplit state before $R_{\mathrm{fix},\gamma}$ is
	\begin{align}
		-\frac1{\sqrt{2(q-1)(q+1)}}
		\sum_{\substack{\theta\in X_{q^2}\\
				\theta|_{\F_q^*}=\chi_\gamma}}
		\left(
		\sum_{\tr(z)=1}
		\theta(z)\chi_{\alpha^{-1}}(\No z)
		\right)
		\ket{\mathrm{ns},\theta}.
		\label{eq: GL2 full nonsplit state before folding}
	\end{align}

	We group this sum under the Frobenius involution $\theta\mapsto\theta^q$.  The substitution $z\mapsto z^q$ shows that
	\begin{equation*}
	\sum_{\tr(z)=1}\theta^q(z)\chi_{\alpha^{-1}}(\No z)
	=\sum_{\tr(z)=1}\theta(z)\chi_{\alpha^{-1}}(\No z).
	\end{equation*}
	The nonfixed orbits are the elements of $\mathcal C_\gamma$.  The fixed characters are $\chi_\delta\circ\No$ for $\delta\in\mathcal S_\gamma$.  Therefore \cref{eq: GL2 full nonsplit state before folding} becomes
	\begin{align}
		-\frac1{\sqrt{2(q-1)(q+1)}}\bigg(&
		\sum_{\{\theta,\theta^q\}\in\mathcal C_\gamma}
		\left(\sum_{\tr(z)=1}
		\theta(z)\chi_{\alpha^{-1}}(\No z)\right)
		\bigl(\ket{\mathrm{ns},\theta}
		+\ket{\mathrm{ns},\theta^q}\bigr)
		\nonumber\\
		&+\sum_{\delta\in\mathcal S_\gamma}
		\left(\sum_{\tr(z)=1}
		(\chi_{\alpha^{-1}\delta}\circ\No)(z)\right)
		\ket{\mathrm{ns},\chi_\delta\circ\No}\bigg).
		\label{eq: GL2 nonsplit state grouped into orbits}
	\end{align}
	Applying $R_{\mathrm{fix},\gamma}$ to the first sum gives
	\begin{equation*}
	-\frac1{\sqrt{(q-1)(q+1)}}
	\sum_{\{\theta,\theta^q\}\in\mathcal C_\gamma}
	\left(\sum_{\tr(z)=1}
	\theta(z)\chi_{\alpha^{-1}}(\No z)\right)
	\ket{\pi_\theta}.
	\end{equation*}
	This proves \cref{eq: appendix GL2 relative cuspidal coefficient} for all nonfixed nonsplit orbits simultaneously.  Again it also shows directly that the antisymmetric nonsplit amplitudes are zero.

	Finally suppose that $\delta^2=\gamma$ and put $\xi=\alpha^{-1}\delta$.  Before the fixed rotation, the split and nonsplit amplitudes obtained from \cref{eq: GL2 split state grouped into orbits,eq: GL2 nonsplit state grouped into orbits} are
	\begin{equation*}
	\frac{J(\xi,\xi)}{\sqrt2(q-1)} \text{ and}
	\qquad
	-\frac1{\sqrt{2(q-1)(q+1)}}
	\sum_{\tr(z)=1}(\chi_\xi\circ\No)(z),
	\end{equation*}
	respectively with $\xi^2=\alpha^{-1}b\ne1$.
	By \cref{lem: Gauss sum property,lem: Gauss sum trace fibers,lem: degree two Davenport Hasse} it follows that
	\begin{equation*}
	\sum_{\tr(z)=1}(\chi_\xi\circ\No)(z)=-J(\xi,\xi).
	\end{equation*}
	Substitution into \cref{eq: GL2 fixed torus rotation} now gives
	\begin{align*}
		\bra{\St_\delta}\T_\gamma\ket{c_\alpha}
		&=\frac{\sqrt q}{(q-1)\sqrt{q+1}}J(\xi,\xi),\\
		\bra{\det_\delta}\T_\gamma\ket{c_\alpha}&=0.
	\end{align*}
	This is \cref{eq: appendix GL2 relative fixed coefficient}.  The sign in \cref{lem: degree two Davenport Hasse} uses the convention in \cref{def: Gauss sum}.
\end{proof}

\begin{lemma}[Implementation of $B_\gamma$, \protect{\cref{lem: B implementation}}]
	The unitary extension of $B_\gamma$ can be implemented to operator-norm error at most $\varepsilon$ with $\operatorname{poly}(\log q,\log(1/\varepsilon))$ gates.
\end{lemma}
\begin{proof}
    We claim that $B_\gamma$ is implemented by the circuit in \cref{fig: GL2 branching gate ancillas} together with some boundary computations.
    In short, it consists of root finding, a Hadamard gate, swapping, taking ratios of roots, adding phases and uncomputing. We explain all the steps, why they combine to the branching map and why they are efficiently implementable.

	We first treat the case in which $X^2-X+t$ has two distinct, nonzero roots and do the boundary cases, i.e., $t=0,1/4$ for odd $q$, $t=0$ for even $q$, later.
	Suppose first that $q$ is odd.  Compute the discriminant $\Delta=1-4t$. As $\Delta$ is non-zero, write $\Delta=g^e$ using the fixed generator $g$ of $\F_q^*$.  Discrete logarithm computes $e$ in a work register.
The lowest bit of $e$  gives its parity, determines whether $\Delta$ is a square and thus marks the split/nonsplit flag.

	On the split branch, write $e=2m$ and compute $s=g^m$.  Then $s^2=\Delta$, and a fixed ordering of the two roots is
	\begin{equation*}
	x=\frac{1+s}{2},
	\qquad
	y=\frac{1-s}{2}.
	\end{equation*}
	On the nonsplit branch, $e$ is odd.  Recall that $g=h^{q+1}$ for the fixed generator $h$ of $\F_{q^2}^*$.
Hence $ \Delta=h^{(q+1)e}$, and $s:=h^{(q+1)e/2}$ satisfies $s^2=\Delta$.
  Its Frobenius conjugate is $s^q=-s$, and
	\begin{equation*}
	z=\frac{1+s}{2},
	\qquad
	z^q=\frac{1-s}{2}
	\end{equation*}
	are the two roots in $\F_{q^2}\setminus\F_q$. This root finding process uses exponentiation, discrete logarithm, and finite field arithmetic.

	Now suppose that $q$ is even. Then finding a root is the same as solving $u^2+u=t$.  The map
	\begin{equation*}
	L:u\longmapsto u^2+u
	\end{equation*}
	is $\F_2$-linear on a fixed polynomial-basis representation of the field.
	Its image in $\F_q$ is precisely the absolute-trace-zero subspace.  Thus
	\begin{equation*}
	\operatorname{Tr}_{\F_q/\F_2}(t)=0
	\end{equation*}
	is exactly the split condition.  On this branch, a fixed linear right inverse of $L$ computes one root $x\in\F_q$, and the other root is $y=x+1$.  If the absolute trace is one, there is no root in $\F_q$.
	However,
	\begin{equation*}
	\operatorname{Tr}_{\F_{q^2}/\F_2}(t)
	=2\operatorname{Tr}_{\F_q/\F_2}(t)=0,
	\end{equation*}
	so a fixed linear section $S: \im L \to \F_{q^2}$ of $L$ on $\F_{q^2}$ computes a root $z\in\F_{q^2}$. $S$ being a section means $LS = \id_{\im L}$. This root cannot lie in $\F_q$ and its conjugate $z^q$ is the other root.  Fixed linear maps over $\F_2$ have circuits of size polynomial in $\log q$. In total we obtain an efficient implementation of $U_{\mathrm{roots}}$ in \cref{fig: GL2 branching gate ancillas}.

	In any characteristic, we now have a branch flag and a fixed ordering of two distinct roots.  Applying a Hadamard on an ancilla and conditionally swapping the two root registers gives
	\begin{equation*}
	\frac{\ket0\ket{x}\ket{y}+
		\ket1\ket{y}\ket{x}}{\sqrt2},
	\end{equation*}
	for the split branch. On the nonsplit branch it gives the same expression with $(x,y)$ replaced by $(z,z^q)$.

	For each ordered pair $(u,v)$, compute the ratio
	\begin{equation*}
	r=u/v
	\end{equation*}
	in an extra register.
	The roots are nonzero because $t\ne0$.  On the split branch $r\in\F_q^*$, while on the nonsplit branch
	\begin{equation*}
	r^q=\frac{z^q}{z}=r^{-1},
	\end{equation*}
	so $r\in\ker\No$.  This is an efficient operation denoted by $U_r$ in \cref{fig: GL2 branching gate ancillas}. The ratio is needed for the phases in the definition of $B_{\gamma}.$

	These phases are all that remains to be added to the circuit.
The gate applying these phases is denoted by $\Phi_\gamma$.
 On the split branch, apply $\chi_\gamma$ to the second root $v$.  The two ordered terms then become
	\begin{equation*}
	\frac{\chi_\gamma(y)\ket{\mathrm{sp},x/y}
		+\chi_\gamma(x)\ket{\mathrm{sp},y/x}}{\sqrt2}.
	\end{equation*}
	On the nonsplit branch, apply $\theta_0$ to the first root $u$ and also apply the constant phase $-1$.  This gives
	\begin{equation*}
	-\frac{\theta_0(z)\ket{\mathrm{ns},z/z^q}
		+\theta_0(z^q)\ket{\mathrm{ns},z^q/z}}{\sqrt2}.
	\end{equation*}
	These character phases are implemented by \cref{lem: controlled character phase}. Note that the nonsplit branch works over $\F_{q^2}$.
	The exponent label of $\theta_0$ is computed reversibly from the label of $\gamma$ using the lift in the definition of $B_\gamma$.

	Lastly, we explain why all work registers can be returned to zero efficiently. This is gate $U_{\mathrm{clean}}$ in \cref{fig: GL2 branching gate ancillas}.  For an ordered pair of nonboundary roots $(u,v)$, we have $u+v=1$ and $r=u/v$.
	Therefore
	\begin{equation}
		u=\frac r{1+r},
		\qquad
		v=\frac1{1+r},
		\qquad
		t=uv=\frac r{(1+r)^2}.
		\label{eq: GL2 router reconstruction}
	\end{equation}
	The denominator $1+r$ is nonzero because $u+v=1$.  Thus the output branch and ratio determine the oriented roots and the original input $t$.
Using field arithmetic, the ordered root pairs can be returned to zero and only the ratio $r$ remains.
 Note that the order can also be recovered from the root reconstruction. Similarly the square, trace, and branch test of $t$ can be undone.

	It remains to implement the boundary columns.  For odd $q$, set the ratio register to $1$ and apply the two-dimensional transformation
	\begin{align*}
		\ket0&\longmapsto
		\frac{\ket{\mathrm{sp},1}+\ket{\mathrm{ns},1}}{\sqrt2},\\
		\ket{1/4}&\longmapsto
		\frac{\chi_\gamma(1/2)}{\sqrt2}
		\bigl(\ket{\mathrm{sp},1}-\ket{\mathrm{ns},1}\bigr).
	\end{align*}
	The two output vectors are orthonormal.  For even $q$, only the first column is required; complete it arbitrarily on the unused orthogonal boundary state.  These are constant-dimensional operations.

	The whole circuit uses a constant number of reversible field operations, discrete logarithms, exponentiations, character phases, and fixed linear maps.  Assigning appropriate error parts to each approximate subroutine gives total operator-norm error at most $\varepsilon$, with $\operatorname{poly}(\log q,\log(1/\varepsilon))$ gates.
\end{proof}

\begin{figure}[H]
\centering
{\scriptsize
\begin{quantikz}[row sep=0.24cm,column sep=0.34cm]
\lstick{$\ket\gamma$}
 & \qw
 & \qw
 & \qw
 & \qw
 & \ctrl{2}
 & \qw
 & \rstick{$\ket\gamma$}
\\
\lstick{$\ket t$}
 & \gate[6]{U_{\mathrm{roots}}}
 & \qw
 & \qw
 & \qw
 & \qw
 & \gate[6]{U_{\mathrm{clean}}}
 & \rstick{$\ket0$}
\\
\lstick{$\ket0_\tau$}
 &
 & \qw
 & \qw
 & \qw
 & \gate[4]{\Phi_\gamma}
 &
 & \rstick{$\ket{\tau(t)}$}
\\
\lstick{$\ket0_{u,v}$}
 &
 & \qw
 & \gate{\operatorname{swap}}
 & \gate[3]{U_r}
 &
 &
 & \rstick{$\ket0_{u,v}$}
\\
\lstick{$\ket0_o$}
 &
 & \gate{H}
 & \ctrl{-1}
 &
 &
 &
 & \rstick{$\ket0_o$}
\\
\lstick{$\ket0_r$}
 &
 & \qw
 & \qw
 &
 &
 &
 & \rstick{$\ket{r(t,o)}$}
\\
\lstick{$\ket0_{\mathsf a_B}$}
 &
 & \qw
 & \qw
 & \qw
 & \qw
 &
 & \rstick{$\ket0_{\mathsf a_B}$}
\end{quantikz}
}
\caption{Circuit for the nonboundary columns of $B_\gamma$.
The gate $U_{\mathrm{roots}}$ computes the split/nonsplit branch (the $\tau$-register), a canonical root pair $(u,v)$, and the square-root or trace data used in that computation.
The order qubit $o$ creates the two root orders, $U_r$ writes $r=u/v$, and $\Phi_\gamma$ applies $\chi_\gamma(v)$ on the split branch or $-\theta_0(u)$ on the nonsplit branch.
Finally, $U_{\mathrm{clean}}$ uses $t=r/(1+r)^2$ to erase $t$, the roots, the order qubit, and the remaining work.
The lowest register are ancillas for the root computations.
The boundary columns $t=0$ and, for odd $q$, $t=1/4$ are handled separately by the two-dimensional boundary operation in the proof.}
\label{fig: GL2 branching gate ancillas}
\end{figure}
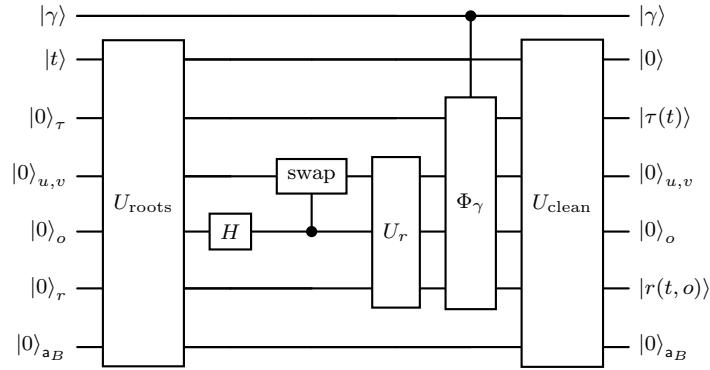

\begin{lemma}[Preparation of an exceptional target column, \protect{\cref{lem: GL2 exceptional target preparation}}]
	If $\alpha^2=\gamma$, then $A_\gamma\ket{c_\alpha}$ can be prepared to error at most $\varepsilon$ with $\operatorname{poly}(\log q,\log(1/\varepsilon))$ gates.
\end{lemma}
\begin{proof}
	We construct the state from the three sums in \cref{eq: corrected A nontrivial column}.  Since $\alpha^2=\gamma$, the label $b=\alpha^{-1}\gamma$ occurring there is $b=\alpha$.

	By the known norms of Gauss sums \cref{lem: Gauss sum property}, the squared norms of the cuspidal, principal, and Steinberg parts are
	\begin{equation}
		N_{\mathrm C}=|\mathcal C_\gamma|
		\frac q{(q-1)(q+1)}, \qquad
		N_{\mathrm P}=|\mathcal P_\gamma|
		\frac q{(q-1)^2}, \qquad
		N_{\mathrm S}=  \frac{1+(s_\gamma-1)q^2}{(q-1)^2(q+1)},
		\label{eq: GL2 exceptional branch norms}
	\end{equation}
	where $s_\gamma$ is the number of roots of $\gamma$ in $\F_q$. The three terms sum to one. We prepare a three-valued branch register with amplitudes $\sqrt{N_{\mathrm C}},\sqrt{N_{\mathrm P}},\sqrt{N_{\mathrm S}}$.
	These numbers depend only on $q$ and $s_\gamma$ and can be synthesized to precision $O(\varepsilon)$ with $\operatorname{poly}(\log q,\log(1/\varepsilon))$ gates.

	We next prepare the label superpositions within the cuspidal and principal branches.  Both can be created using a cyclic set with an involution.  For the principal branch the cyclic set has order $q-1$ and the involution is
	\begin{equation*}
	\delta\longmapsto\gamma\delta^{-1}.
	\end{equation*}
	For the cuspidal branch the set labeling the character coordinates  has order $q+1$ and the involution is Frobenius
	\begin{equation*}
	\theta\mapsto \theta^q,
	\end{equation*}
	expressed in exponent labels by \cref{eq: GL2 affine Frobenius coordinate}.  In both cases there are $s_\gamma$ fixed exponents.

	For a cyclic set of order $L$ (so for us $L=q-1, q+1$), compute these fixed exponents, which can be done efficiently using field arithmetic. There are zero or two when $q$ is odd and exactly one when $q$ is even.
	If $L-s_\gamma>0$, use a cyclic QFT to prepare the uniform state on $\Z/(L-s_\gamma)\Z$.  Map its $j$th element to the $j$th exponent in $\Z/L\Z$ that is not fixed. This uses only comparisons and additions.  For every resulting exponent, compute its partner under the involution, choose a canonical representative of the pair, and record in one bit which member was present.
	The uniform state on nonfixed elements then has the form
	\begin{equation*}
	\frac1{\sqrt{\#\text{ orbits}}}
	\sum_{\text{nonfixed orbits}}
	\ket{\text{representative}}
	\frac{\ket0+\ket1}{\sqrt2}.
	\end{equation*}
	We can clear the order bit using a Hadamard.  Reversible character arithmetic then attaches the physical label $\ket{I_{\delta,\eta}}$ or $\ket{\pi_\theta}$.  This gives uniform superpositions over $\mathcal P_\gamma$ and $\mathcal C_\gamma$.  Note that the case $L-s_\gamma=0$ can occur, in which case we do not have to create a superposition. In particular the principal branch is empty for $q=3$ and $s_\gamma=2$.

	With correct superpositions in place, it remains to attach the coefficients from \cref{eq: corrected A nontrivial column}.  On the cuspidal branch, compute $\Theta=\theta\cdot(\chi_{\alpha^{-1}}\circ\No)$ from the orbit representative.  Apply \cref{thm: coherent Gauss phase} over $\F_{q^2}$, with additive character $\psi\circ\tr$, to add the phase
	\begin{equation*}
	\frac{g_{q^2}(\Theta)}q.
	\end{equation*}
	Also apply the branch phase $-\chi_\alpha(-1)$.  The resulting amplitude of each cuspidal label is
	\begin{equation*}
	\frac{\sqrt{N_{\mathrm C}}}{\sqrt{|\mathcal C_\gamma|}}
	\left(-\chi_\alpha(-1)\frac{g_{q^2}(\Theta)}q\right)
	=-\frac{\chi_\alpha(-1)}{\sqrt{q(q^2-1)}}
	g_{q^2}(\Theta),
	\end{equation*}
	which is the required coefficient.  This phase is independent of the chosen member of the Frobenius orbit because the substitution $z\mapsto z^q$ leaves the Gauss sum unchanged.

	On the principal branch, apply the branch phase $\chi_\alpha(-1)$ and the additional phase $\chi_{\alpha^{-1}\delta}(-1)$.  Since $g_q(\xi)g_q(\xi^{-1})=q\chi_\xi(-1)$ for any $\xi$, the resulting amplitude is
	\begin{align*}
		&\frac{\sqrt{N_{\mathrm P}}}{\sqrt{|\mathcal P_\gamma|}}
		\chi_\alpha(-1)\chi_{\alpha^{-1}\delta}(-1)\\
		&\qquad=
		\frac{\chi_\alpha(-1)}{(q-1)\sqrt q}
		g_q(\alpha^{-1}\delta)g_q(\alpha^{-1}\eta),
	\end{align*}
	again the coefficient in \cref{eq: corrected A nontrivial column}.  This phase is well defined on the unordered pair because a character and its inverse have the same value at $-1$.

	The Steinberg branch contains at most two labels and can be prepared directly.  Its normalized internal state is
	\begin{equation*}
	\frac1{\sqrt{1+(s_\gamma-1)q^2}}
	\left(
	\ket{\St_\alpha}
	+\sum_{\delta\in\mathcal S_\gamma\setminus\{\alpha\}}
	q\frac{g_q(\alpha^{-1}\delta)^2}{q}
	\ket{\St_\delta}
	\right).
	\end{equation*}
	If the second term exists, one two-dimensional rotation prepares the relative magnitudes $1$ and $q$.  Its phase $g_q(\alpha^{-1}\delta)^2/q$ is obtained by applying the normalized Gauss phase twice.  Multiplying this state by the branch amplitude $\sqrt{N_{\mathrm S}}$ and the branch phase $\chi_\alpha(-1)$ gives exactly the Steinberg sum in \cref{eq: corrected A nontrivial column}.

	All orbit partners, character products, fixed-point tests, and phase labels are reversible functions of the output labels and can therefore be uncomputed.  The construction uses a constant number of approximate rotations, cyclic QFTs, character operations, and Gauss-phase circuits.  Dividing the error among them gives total error at most $\varepsilon$ and gate count $\operatorname{poly}(\log q,\log(1/\varepsilon))$.

\end{proof}

\section{Additional proofs of \texorpdfstring{\cref{sec: clifford theory}}{the Clifford theory section}}\label{app: proofs of clifford}

\begin{proof}[Completion of the proof of \cref{thm:little-group-qft-correctness}]
	We track the left and right regular actions through each step of \cref{alg:little-group-qft} and verify that by the end of the circuit, for each output label $(\sigma,\eta')$, the left and right Fourier registers carry $\Rep_{(\sigma,\eta')}$ and $(\Rep_{(\sigma,\eta')})^\vee$, respectively.

\paragraph{Steps 1--3.}
Fix $x\in G$ and define the factorizations
\begin{equation}\label{eq:little-proof-input-factorizations}
    \begin{aligned}
        g&=nh, & n&\in N, & h&\in\Transversal{N},\\
        xh&=m_Lh', & m_L&\in N, & h'&\in\Transversal{N},\\
        hx\inv&=m_Rh'', & m_R&\in N, & h''&\in\Transversal{N}.
    \end{aligned}
\end{equation}
By \cref{eq:little-proof-input-factorizations}, the two regular actions on the encoded input $\ket g$ are determined by
\begin{equation}\label{eq:little-proof-regular-actions}
    xg=(xnx\inv)m_Lh',\qquad gx\inv=nm_Rh''.
\end{equation}
After the $N$-QFT and the orbit encoding, the registers are
\begin{equation}\label{eq:little-proof-orbit-registers}
    \ket{\sigma,t_L}\otimes\ket{i_\la}\otimes\ket{k_\la^\vee}\otimes\ket{h},
    \qquad \la=\act{t_L}\sigma.
\end{equation}
The right action in \cref{eq:little-proof-regular-actions} fixes $\sigma,t_L,i_\la$, sends $h$ to $h''$, and acts on the dual register by $\Rep_\la^\vee(m_R\inv)$.
For the left action, factor
\begin{equation}\label{eq:little-proof-left-factorization}
    \begin{aligned}
        xt_L&=t_L'c_L,\qquad t_L'\in\Transversal{\sigma},\quad c_L\in\Inertia{\sigma},\\
        \la'&:=\act{x}\la=\act{t_L'}\sigma.
    \end{aligned}
\end{equation}
Using \cref{eq:little-proof-orbit-registers,eq:little-proof-left-factorization}, the orbit-transport relation \cref{eq:orbit-frame-intertwining-relation} and the extension property give the intertwiner
\begin{equation}\label{eq:little-proof-intertwiner}
    \begin{aligned}
        W&:=U_{t_L',\sigma}\ExtRep_\sigma(c_L)U_{t_L,\sigma}^\dagger,\\
        \Rep_{\la'}(xnx\inv)&=W\Rep_\la(n)W\inv.
    \end{aligned}
\end{equation}
Combining \cref{eq:little-proof-regular-actions,eq:little-proof-intertwiner} gives
\begin{equation}\label{eq:little-proof-left-fourier-action}
    \Rep_{\la'}((xnx\inv)m_L)
    =W\Rep_\la(n)W\inv\Rep_{\la'}(m_L).
\end{equation}
In Fourier coordinates, \cref{eq:little-proof-left-fourier-action} means that the left action sends $(t_L,h)$ to $(t_L',h')$, applies $W$ to the left matrix-index register and $\Rep_{\la'}^\vee(m_L\inv)W^{-\mathsf T}$ to the right one. The orbit label $\sigma$ is unchanged under both actions.

\paragraph{Steps 4--5.}
Step 4 applies $U_{t_L,\sigma}^\dagger$ and $U_{t_L,\sigma}\tp$, placing the two matrix-index registers in $V_\sigma$ and $V_\sigma^\vee$.
Using the factors in \cref{eq:little-proof-input-factorizations,eq:little-proof-left-factorization}, set
\begin{equation}\label{eq:little-proof-normal-corrections}
    d_L:=(t_L')\inv m_L\inv t_L'\in N,
    \qquad d_R:=t_L\inv m_R\inv t_L\in N.
\end{equation}
Substituting \cref{eq:little-proof-intertwiner,eq:little-proof-normal-corrections} when conjugating the preceding operators by the input and output transport gates gives
\begin{equation}\label{eq:little-proof-transported-actions}
\begin{aligned}
    U_{t_L',\sigma}^\dagger WU_{t_L,\sigma}
    &=\ExtRep_\sigma(c_L),\\
    U_{t_L',\sigma}\tp\Rep_{\la'}^\vee(m_L\inv)
    W^{-\mathsf T}(U_{t_L,\sigma}\tp)\inv
    &=\Rep_\sigma^\vee(d_L)(\ExtRep_\sigma)^\vee(c_L),\\
    U_{t_L,\sigma}\tp\Rep_\la^\vee(m_R\inv)(U_{t_L,\sigma}\tp)\inv
    &=\Rep_\sigma^\vee(d_R).
\end{aligned}
\end{equation}
Thus the left action has the required extension action on $i_\sigma$, but also acts on the right matrix-index register.

Step 5 replaces $h$ by $(a,t_R)$. Define these labels and the factors needed for the right action by
\begin{equation}\label{eq:little-proof-right-factorizations}
    \begin{aligned}
        h&=t_Lat_R\inv,\qquad t_R\in\Transversal{\sigma},\quad a\in\Inertia{\sigma},\\
        xt_R&=t_R'c_R,\qquad t_R'\in\Transversal{\sigma},\quad c_R\in\Inertia{\sigma}.
    \end{aligned}
\end{equation}
The factorizations in \cref{eq:little-proof-input-factorizations,eq:little-proof-left-factorization,eq:little-proof-right-factorizations} and the definitions in \cref{eq:little-proof-normal-corrections} give
\begin{equation}\label{eq:little-proof-transversal-updates}
\begin{aligned}
    h'&=m_L\inv xh=t_L'd_Lc_Lat_R\inv,\\
    h''&=m_R\inv hx\inv=t_Ld_Rac_R\inv(t_R')\inv
\end{aligned}
\end{equation}
Define the updated inertia elements by
\begin{equation}\label{eq:little-proof-inertia-updates}
    a_L':=d_Lc_La,\qquad a_R':=d_Rac_R\inv.
\end{equation}
Then \cref{eq:little-proof-transversal-updates,eq:little-proof-inertia-updates} show that the left action fixes $t_R$ and sends $a$ to $a_L'$, while the right action fixes $t_L$ and sends $(a,t_R)$ to $(a_R',t_R')$.

\paragraph{Step 6.}
The controlled gate $(\ExtRep_\sigma)^\vee(a\inv)$ makes the right Fourier register invariant under the left regular action;
when conjugating the actions in \cref{eq:little-proof-transported-actions} by this gate, its control changes from $a$ to $a_L'$ or $a_R'$ as defined in \cref{eq:little-proof-inertia-updates}. Since $\ExtRep_\sigma$ restricts to $\Rep_\sigma$ on $N$, substituting these definitions gives
\begin{align}
    (\ExtRep_\sigma)^\vee((a_L')\inv)
    \Rep_\sigma^\vee(d_L)(\ExtRep_\sigma)^\vee(c_L)
    (\ExtRep_\sigma)^\vee(a)
    &=(\ExtRep_\sigma)^\vee((a_L')\inv d_Lc_La)
      =\id_{V_\sigma^\vee},
    \label{eq:left-step6-dual-cancellation}\\
    (\ExtRep_\sigma)^\vee((a_R')\inv)
    \Rep_\sigma^\vee(d_R)(\ExtRep_\sigma)^\vee(a)
    &=(\ExtRep_\sigma)^\vee((a_R')\inv d_Ra)
      =(\ExtRep_\sigma)^\vee(c_R).
    \label{eq:right-action-corrected-k-register-final}
\end{align}
The cancellations in \cref{eq:left-step6-dual-cancellation,eq:right-action-corrected-k-register-final} show that the left action now acts only on $t_L,i_\sigma,a$, and the right action only on $t_R,k_\sigma^\vee,a$.

\paragraph{Steps 7--8.}
Step 7 replaces $a$ by its coset $aN$, uncomputing the redundant $N$-register as justified by \cref{obs:factor-coset-suffices}.
Using $d_L,d_R\in N$ from \cref{eq:little-proof-normal-corrections}, the updates in \cref{eq:little-proof-inertia-updates} reduce to
\begin{equation}\label{eq:little-proof-quotient-actions}
    a_L'N=(c_LN)(aN),\qquad a_R'N=(aN)(c_RN)\inv.
\end{equation}
Thus \cref{eq:little-proof-quotient-actions} gives left translation by $c_LN$ and right translation by $c_RN$ on the $aN$-register, while \cref{eq:little-proof-transported-actions,eq:left-step6-dual-cancellation,eq:right-action-corrected-k-register-final} give the matrix-index actions $\ExtRep_\sigma(c_L)$ and $(\ExtRep_\sigma)^\vee(c_R)$, respectively.

Step 8 applies $\Fou_{\Little{\sigma}}$ to $aN$. The translations in \cref{eq:little-proof-quotient-actions} become the actions
\begin{equation}\label{eq:little-proof-little-group-actions}
    \begin{aligned}
        \Rep_{\eta'}(c_LN)&=\InfRep_{\eta'}(c_L) &&\text{on }j_{\eta'},\\
        \Rep_{\eta'}^\vee(c_RN)&=(\InfRep_{\eta'})^\vee(c_R) &&\text{on }l_{\eta'}^\vee.
    \end{aligned}
\end{equation}
The final basis is
\begin{equation}\label{eq:basis-after-step-8}
    \ket{\sigma,\eta'}\otimes\ket{t_L,i_\sigma,j_{\eta'}}
    \otimes\ket{t_R,k_\sigma^\vee,l_{\eta'}^\vee}.
\end{equation}
By \cref{eq:little-proof-left-factorization,eq:little-proof-transported-actions,eq:little-proof-little-group-actions}, the transformation of the left triple in \cref{eq:basis-after-step-8} is exactly the induced action of
\[
    \Rep_{(\sigma,\eta')}
    =\Ind_{\Inertia{\sigma}}^G
      \bigl(\ExtRep_\sigma\otimes\InfRep_{\eta'}\bigr).
\]
Likewise, \cref{eq:little-proof-right-factorizations,eq:right-action-corrected-k-register-final,eq:little-proof-little-group-actions} identify the transformation of the right triple with its dual, induced from $(\ExtRep_\sigma)^\vee\otimes(\InfRep_{\eta'})^\vee$.
The labels $\sigma,\eta'$ are fixed by both actions.
Together with the normalization and matrix-coefficient calculation in the main proof, which fixes the scalar phase in each irreducible block, this proves the claimed Fourier transform.

\end{proof}

\end{document}